\documentclass[sigplan,screen,authorversion,nonacm]{acmart}
\pdfoutput=1
\usepackage{macros}

\usepackage[leftmargin=1em,rightmargin=1ex,vskip=1ex]{quoting}
\usepackage{lscape}

\title{Monoidal Streams for Dataflow Programming}
\author{Elena Di Lavore}
\author{Giovanni de Felice}
\author{Mario Rom\'an}

\makeatletter
\def\@copyrightspace{\relax}
\makeatother

\begin{document}

\begin{abstract}
  We introduce  monoidal streams: a generalization of causal stream functions to monoidal categories.
  They have a feedback structure that gives semantics to signal flow graphs.
  In the same way that streams provide semantics to dataflow programming with pure functions, monoidal streams provide semantics to dataflow programming with theories of processes represented by a symmetric monoidal category.
  As an example, we study a stochastic dataflow language.
\end{abstract}

\keywords{Monoidal stream, Stream, Monoidal category, Dataflow programming, Feedback, Signal flow graph, Coalgebra, Stochastic process.}
\maketitle

\section{Introduction}\label{section:introduction}

\paragraph{Dataflow languages}
Dataflow (or \emph{stream-based}) programming languages, such as \textsc{Lucid}~\cite{wadge1985lucid,halbwachs1991lustre},
follow a paradigm in which every declaration represents an infinite list of values:
a \emph{stream}~\cite{benveniste93,uustalu05}.
The following program in a \textsc{Lucid}-like
language (\Cref{diagram:dataflowfibonacci}) computes the Fibonacci sequence,
thanks to a $\Fby$ (``followed by'') operator.
\begin{figure}[H]
  \centering $\fib = 0\ \Fby\ (\fib + (1\ \Fby\ \Wait(\fib)))$
\caption{\textit{The Fibonacci sequence is 0 followed
by the Fibonacci sequence plus the Fibonacci sequence preceded by a 1.}}
\label{diagram:dataflowfibonacci}
\end{figure}
The control structure of dataflow programs is inspired by \emph{signal flow graphs} \cite{benveniste93,shannon42,mason53}.
Signal flow graphs are diagrammatic specifications of processes with feedback loops,
widely used in control system engineering.
In a dataflow program, feedback loops represent how the current value of a stream may depend on its previous values.
For instance, the previous program (\Cref{diagram:dataflowfibonacci}) corresponds to the signal flow graph in \Cref{figure:fibonacci}.
\begin{figure}[h!]
  \begin{minipage}{0.3\linewidth}
 \tikzset{every picture/.style={line width=0.85pt}} %
\begin{tikzpicture}[x=0.75pt,y=0.75pt,yscale=-1,xscale=1]
\draw   (50,100) -- (90,100) -- (90,120) -- (50,120) -- cycle ;
\draw   (85,125) -- (105,125) -- (105,145) -- (85,145) -- cycle ;
\draw    (120,95) .. controls (120,135.2) and (116.29,135.8) .. (105,135) ;
\draw    (90,70) .. controls (89.8,55.8) and (120.2,56.6) .. (120,70) ;
\draw    (70,120) .. controls (72,132.09) and (72,134.66) .. (85,135) ;
\draw    (130,62) -- (130,90) ;
\draw [shift={(130,60)}, rotate = 90] [color={rgb, 255:red, 0; green, 0; blue, 0 }  ][line width=0.75]    (10.93,-3.29) .. controls (6.95,-1.4) and (3.31,-0.3) .. (0,0) .. controls (3.31,0.3) and (6.95,1.4) .. (10.93,3.29)   ;
\draw    (105.1,59.1) .. controls (104.9,44.9) and (130.2,46.6) .. (130,60) ;
\draw  [fill={rgb, 255:red, 0; green, 0; blue, 0 }  ,fill opacity=1 ] (102.2,59.1) .. controls (102.2,57.5) and (103.5,56.2) .. (105.1,56.2) .. controls (106.7,56.2) and (108,57.5) .. (108,59.1) .. controls (108,60.7) and (106.7,62) .. (105.1,62) .. controls (103.5,62) and (102.2,60.7) .. (102.2,59.1) -- cycle ;
\draw   (45,70) -- (65,70) -- (65,90) -- (45,90) -- cycle ;
\draw   (55,155) -- (95,155) -- (95,175) -- (55,175) -- cycle ;
\draw   (45,125) -- (65,125) -- (65,145) -- (45,145) -- cycle ;
\draw    (65,195) -- (65,210) ;
\draw   (70,70) -- (110,70) -- (110,90) -- (70,90) -- cycle ;
\draw    (55,90) .. controls (54.57,99.23) and (65.14,93.51) .. (65,100) ;
\draw    (120,70) -- (120,100) ;
\draw    (90,90) .. controls (89.57,99.23) and (80.14,93.51) .. (80,100) ;
\draw    (95,145) .. controls (94.57,154.23) and (85.14,148.51) .. (85,155) ;
\draw    (55,145) .. controls (54.57,154.23) and (65.14,148.51) .. (65,155) ;
\draw    (130,90) .. controls (130.25,158.17) and (105.14,170.09) .. (105,195) ;
\draw    (75,175) -- (75,183.87) ;
\draw    (65,195) .. controls (64.8,180.8) and (85.2,180.8) .. (85,195) ;
\draw  [fill={rgb, 255:red, 0; green, 0; blue, 0 }  ,fill opacity=1 ] (72.1,183.87) .. controls (72.1,182.27) and (73.4,180.97) .. (75,180.97) .. controls (76.6,180.97) and (77.9,182.27) .. (77.9,183.87) .. controls (77.9,185.47) and (76.6,186.77) .. (75,186.77) .. controls (73.4,186.77) and (72.1,185.47) .. (72.1,183.87) -- cycle ;
\draw    (85,195) .. controls (85,204.2) and (105,205) .. (105,195) ;
\draw (55,80) node    {$1$};
\draw (95,135) node    {$+$};
\draw (70,110) node    {$fby$};
\draw (55,135) node    {$0$};
\draw (75,165) node    {$fby$};
\draw (90,80) node    {$wait$};
\end{tikzpicture}
  \end{minipage} \begin{minipage}{0.5\linewidth}
    \begin{gather*} \mathsf{fib} \defn\fbk( \COPY ; \\
      \partial (1 \times \WAIT) \times \im ; \\
      \partial(\fby) \times \im ; \\
      \partial(+) ; \\
      0 \times \im ; \\
      \fby ; \\
      \COPY ) \end{gather*}
  \end{minipage}
  \caption{Fibonacci: signal flow graph and morphism.}\label{figure:fibonacci}\label{figure:finalfibonaccigraph}
\end{figure}
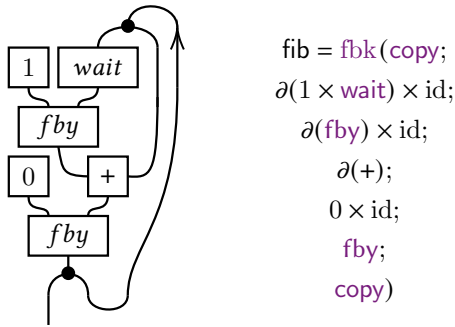

\paragraph{Monoidal categories}
Any theory of processes that \emph{compose sequentially and in parallel}, satisfying reasonable axioms, forms a \emph{monoidal category}.
Examples include functions~\cite{lambek1986a}, probabilistic channels~\cite{cho2019, fritz2020}, partial maps~\cite{cockett02}, database queries \cite{Bonchi18}, linear resource theories~\cite{coeckeFS16} and quantum processes~\cite{abramsky2009categorical}.
Signal flow graphs are the graphical syntax for \emph{\feedbackMonoidalCategories{}}~\cite{katis02,bonchi14,bonchi15,feedbackspans2020,kaye22}: they are the \emph{string diagrams} for any of these theories, extended with \emph{feedback}.

Yet, semantics of dataflow languages have been mostly restricted to theories of pure functions~\cite{benveniste93,uustalu2008comonadic,cousot19,delpeuch19,oliveira84}: what are called \emph{cartesian} monoidal categories.
We claim that this restriction is actually inessential;
dataflow programs may take semantics in non-cartesian monoidal categories, exactly as their signal flow graphs do.

The present work provides this missing semantics:
we construct \emph{monoidal streams} over a symmetric monoidal category, which form a \emph{\feedbackMonoidalCategory{}}.
Monoidal streams model the values of a \emph{monoidal dataflow language}, in the same way that streams model the values of a classical dataflow language.
This opens the door to stochastic, effectful, or quantum dataflow languages.
In particular, we give semantics and string diagrams for a \emph{stochastic dataflow programming language}, where the following code can be run.

\begin{figure}[H]
  \centering  $\walk = 0\ \Fby\ (\uniform{(-1,1)} + \walk)$
  \caption{\textit{A stochastic dataflow program. A random walk is 0 followed by the random walk plus a stochastic stream of steps to the left (-1) or to the right (1), sampled uniformly.}}
\label{diagram:dataflowwalk}
\end{figure}

\begin{figure}[h!]
\begin{minipage}{0.3\linewidth}
\tikzset{every picture/.style={line width=0.85pt}}
\begin{tikzpicture}[x=0.75pt,y=0.75pt,yscale=-1,xscale=1]
\draw   (190,110) -- (230,110) -- (230,130) -- (190,130) -- cycle ;
\draw   (170,140) -- (210,140) -- (210,160) -- (170,160) -- cycle ;
\draw   (160,110) -- (180,110) -- (180,130) -- (160,130) -- cycle ;
\draw    (190,160) -- (190,170) ;
\draw    (220,110) .. controls (220.2,97) and (239.4,97.4) .. (240,110) ;
\draw   (170,80) -- (210,80) -- (210,100) -- (170,100) -- cycle ;
\draw    (170,130) .. controls (169.57,139.23) and (180.14,133.51) .. (180,140) ;
\draw    (210,130) .. controls (209.57,139.23) and (200.14,133.51) .. (200,140) ;
\draw    (180,180) -- (180,195) ;
\draw    (180,180.65) .. controls (179.8,166.45) and (200.2,166.45) .. (200,180.65) ;
\draw  [fill={rgb, 255:red, 0; green, 0; blue, 0 }  ,fill opacity=1 ] (187.1,170) .. controls (187.1,168.4) and (188.39,167.1) .. (190,167.1) .. controls (191.6,167.1) and (192.9,168.4) .. (192.9,170) .. controls (192.9,171.6) and (191.6,172.9) .. (190,172.9) .. controls (188.39,172.9) and (187.1,171.6) .. (187.1,170) -- cycle ;
\draw    (200,180.65) .. controls (200,189.85) and (220,190) .. (220,180) ;
\draw    (190,100) .. controls (189.57,109.23) and (200.14,103.51) .. (200,110) ;
\draw    (220,180) .. controls (220.2,150.25) and (239.99,148.64) .. (240.01,111.71) ;
\draw [shift={(240,110)}, rotate = 89.11] [color={rgb, 255:red, 0; green, 0; blue, 0 }  ][line width=0.75]    (10.93,-3.29) .. controls (6.95,-1.4) and (3.31,-0.3) .. (0,0) .. controls (3.31,0.3) and (6.95,1.4) .. (10.93,3.29)   ;

\draw (210,120) node    {$+$};
\draw (170,120) node    {$0$};
\draw (190,150) node    {$fby$};
\draw (190,90) node    {$unif$};
\end{tikzpicture}
\end{minipage} \begin{minipage}{0.5\linewidth}
  \begin{gather*} \mathsf{walk} \defn
    \fbk( \\
    \partial(\mathsf{unif}) \tensor  \im; \\
    0 \tensor \partial(+) ; \\
    \fby ; \\
    \COPY )
  \end{gather*}
\end{minipage}
\caption{Random walk: signal flow graph and morphism.}
\label{string:walk}
\end{figure}
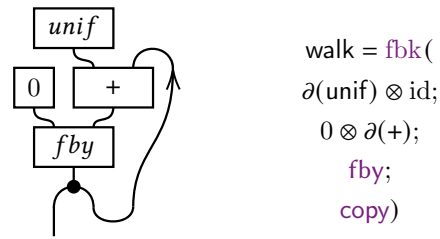

\subsection{Contributions}

Our main novel contribution is the definition of a \feedbackMonoidalCategory{} of \emph{monoidal streams}
over a \symmetricMonoidalCategory{} ($\STREAM$,~\Cref{def:monoidalstream,th:monoidalstreamsfeedback}).
Monoidal streams form a \finalCoalgebra{}; for sufficiently well-behaved monoidal categories (\Cref{def:productive}), we give an explicit construction of this coalgebra (\Cref{def:observationalsequence}).

In \cartesianCategories{}, the \emph{causal functions} of Uustalu and Vene~\cite{uustalu05} (see also \cite{jacobs:causalfunctions})
are a particular case of our monoidal streams (\Cref{th:cartesianstreams,th:nelist}).
In the category of stochastic functions, our construction captures the notion of
\emph{controlled stochastic process}~\cite{fleming1975,ross1996stochastic} (\Cref{th:stochasticprocesses}).

In order to arrive to this definition, we unify the previous literature:
we characterize the cartesian ``intensional stateful sequences'' of Katsumata and Sprunger
with a final coalgebra (\Cref{th:intensionalcoalgebra}), and then ``extensional stateful sequences''
in terms of the ``feedback monoidal categories'' of Katis, Sabadini and Walters~\cite{katis02} (\Cref{th:ext-stateful-sequences}).
We justify observational equivalence with a refined fixpoint equation that employs \emph{coends} (\Cref{theorem:observationalfinalcoalgebra}).
We strictly generalize ``stateful sequences'' from the cartesian to the monoidal case.

Finally, we extend a type theory of \symmetricMonoidalCategories{} with a feedback operator (\Cref{section:typetheory}) and we use it as a stochastic dataflow programming language.

\subsection{Related work}

\paragraph{Coalgebraic streams.}
Uustalu and Vene~\cite{uustalu05} provide elegant \emph{comonadic} semantics for a (cartesian) \Lucid-like programming language.
We shall prove that their exact technique cannot be possibly extended to arbitrary monoidal categories (\Cref{th:nelist}).
However, we recover their semantics as a particular case of our monoidal streams (\Cref{th:cartesianstreams}).

\paragraph{Stateful morphism sequences.}
Sprunger and Katsumata constructed the category of \emph{stateful sequences} in the cartesian case~\cite{katsumata19}.
Our work is based on an unpublished work by Rom{\'a}n~\cite{roman2020} that first exteneded this definition to the symmetric monoidal case,
using \emph{coends} to justify \emph{extensional equality}.
Shortly after, Carette, de Visme and Perdrix~\cite{carette21} rederived this construction and applied it to the case of completely positive maps between Hilbert spaces, using (a priori) a slightly different notion of equality.
We synthetise some of this previous work, we justify it for the first time using coalgebra and we particularize it to some cases of interest.

\paragraph{Feedback.}
Feedback monoidal categories are a weakening of \emph{traced monoidal categories}.
The construction of the free such categories is originally due to Katis, Sabadini and Walters~\cite{katis02}.
Feedback monoidal categories and their free construction have been repurposed and rediscovered multiple times in the literature~\cite{sabadini95,hoshino14,bonchi19,gay03}.
Di Lavore et al.~\cite{feedbackspans2020} summarize these uses and introduce \emph{delayed feedback}.

\paragraph{General and dependent streams.}
Our work concerns \emph{synchronous} streams: those where, at each point in time $t = 0,1,\dots$,
the stream process takes exactly one input and produces exactly one output.
This condition is important in certain contexts like, for instance, real-time embedded systems; %
but it is not always present.
The study of asynchronous stream transformers and their universal properties is considerably different~\cite{abadi15},
and we refer the reader to the recent work of Garner~\cite{garner2021stream} for a discussion on \emph{non-synchronous} streams.
Finally, when we are concerned with \emph{dependent streams} indexed by time steps,
a possible approach, when our base category is a topos, is to use the \emph{topos of trees}~\cite{birkedal11}.

\paragraph{Categorical dataflow programming.}
Category theory is a common tool of choice for dataflow programming~\cite{rutten00,gay03,mamouras20}.
In particular, profunctors and coends are used by Hildebrandt, Panangaden and Winskel~\cite{hildebrandt1998relational}
to generalise a model of non-deterministic dataflow,
which has been the main focus~\cite{panangaden1988computations,lynch1989proof,lee09} outside cartesian categories.

\subsection{Synopsis}

This manuscript contains three main definitions in terms of universal properties (\emph{intensional, extensional and observational streams}, \Cref{def:intensionalmonoidalstream,def:extensionalmonoidalstreams,def:observationalmonoidalstream}); and three explicit constructions for them (\emph{intensional, extensional and observational sequences}, \Cref{def:intensionalsequence,def:extensional-sequence,def:observationalsequence}).
Each definition is refined into the next one: each construction is a quotienting of the previous one.

\Cref{section-prelude-coalgebra,sec:dinaturality} contain expository material on coalgebra and dinaturality.
\Cref{sec:int-stateful-sequences} presents intensional monoidal streams.
\Cref{section:extensional} introduces extensional monoidal streams in terms of feedback monoidal categories.
\Cref{section:observational} introduces the definitive \emph{observational} equivalence and defines \emph{monoidal streams}.
\Cref{sec:monoidal-streams} constructs the feedback monoidal category of monoidal streams.
\Cref{section:classicalstreams,section:stochastic-streams} present two examples: cartesian and stochastic streams.
\Cref{section:typetheory} introduces a type theory for feedback monoidal categories.
\subsection{Prelude: Coalgebra}
\label{section-prelude-coalgebra}

In this preparatory section, we introduce some background material on coalgebra~\cite{rutten00,jacobs2005coalgebras,adamek2005introduction}.
Coalgebra is the category-theoretic study of stateful systems and infinite data-structures, such as streams.
These structures arise as \emph{final coalgebras}: universal solutions to certain functor equations.

Let us fix an endofunctor $F \colon \catC \to \catC$ through the section.

\begin{definition}
  A \emph{coalgebra} $(Y,\beta)$ is an object $Y \in \catC$, together with a
  morphism $\beta \colon Y \to FY$. A \emph{coalgebra
  morphism} $g \colon (Y,\beta) \to (Y',\beta')$ is a morphism
  $g \colon Y \to Y'$ such that $g ; \beta' = \beta ; Fg$.
\end{definition}

Coalgebras for an endofunctor form a category with coalgebra morphisms between them.
A \emph{final coalgebra} is a final object in this category.
As such, final coalgebras are unique up to isomorphism when they exist.

\begin{definition}\defining{linkfinalcoalgebra}{}
A \emph{final coalgebra} is a coalgebra $(Z,\gamma)$ such that for any other coalgebra $(Y,\beta)$ there exists a unique coalgebra morphism $g \colon (Y,\beta) \to (Z,\gamma)$.
\end{definition}

Our interest in final coalgebras derives from the fact that they are canonical fixpoints of an endofunctor.
Specifically, Lambek's theorem (\Cref{th:lambektheorem}) states that whenever the final coalgebra exists, it is a fixpoint.

\begin{definition}
  A \emph{fixpoint} is a coalgebra $(Y,\beta)$ such that $\beta \colon Y \to FY$ is an
  isomorphism. A \emph{fixpoint morphism} is a coalgebra morphism between fixpoints:
  fixpoints and fixpoint morphisms form a full subcategory of the category of coalgebras.
  A \emph{final fixpoint} is a final object in this category.
\end{definition}

\begin{theorem}[Lambek, \cite{lambek68}]\label{th:lambektheorem}
  Final coalgebras are fixpoints.
  As a consequence, when they exist, they are final fixpoints.
\end{theorem}

The last question before continuing is how to explicitly construct a final coalgebra.
This is answered by Adamek's theorem (\Cref{th:adamek}).
The reader may be familiar with Kleene's theorem for constructing fixpoints~\cite{stoltenberg}:
the least fixpoint of a monotone function $f \colon X \to X$ in a directed-complete partial order $(X,\leq)$ is
the supremum of the chain $\bot \leq f(\bot) \leq f(f(\bot)) \leq \dots$,
where $\bot$ is the least element of the partial order, whenever this supremum is preserved by $f$.
This same result can be categorified into a fixpoint theorem for constructing final coalgebras:
the directed-complete poset becomes a category with $\omega$-chain limits;
the monotone function becomes an endofunctor;
and the least element becomes the final object.

\begin{theorem}[Adamek, \cite{adamek74}]\label{th:adamek}
  Let $\catD$ be a category with a final object $1$ and $\omega$-shaped
  limits. Let $F \colon \catD \to \catD$ be an endofunctor. We write $L \defn \lim\nolimits_{n} F^{n}1$ for
  the limit of the following $\omega$-chain, which is called the \emph{terminal sequence} of $F$.
  \[1 \overset{!}\longleftarrow F1 \overset{F!}\longleftarrow FF1 \overset{FF!}\longleftarrow FFF1 \overset{FFF!}\longleftarrow \dots \]
  Assume that $F$ preserves this limit, meaning that the canonical morphism $FL \to L$ is an isomorphism.
  Then, $L$ is the final $F$-coalgebra.
\end{theorem}
\section{Intensional Monoidal Streams}
\label{sec:int-stateful-sequences}

This section introduces a preliminary definition of \emph{monoidal stream} in terms of a fixpoint equation (in \Cref{eq:intensionalstreamshort}).
In later sections, we refine both this definition and its characterization into the definitive notion of \emph{monoidal stream}.

Let $(\catC,\tensor,I)$ be a fixed \symmetricMonoidalCategory{}.

\subsection{The fixpoint equation}\label{section:fixpoint}

Classically, type-variant streams have a neat coinductive definition~\cite{jacobs2005coalgebras,rutten00} that says:
\begin{quoting}\defining{linkstreamfun}{}
\emph{``A stream of type \(\stream{A} = \streamExpr{A}\) is an
  element of \(A_{0}\) together with a stream of type
  \(\tail{\stream{A}} = (A_{1}, A_{2}, \ldots)\)''}.
\end{quoting}
Formally, streams are the final fixpoint of the equation
\[\streamFun(A_{0},A_{1},\ldots) \cong A_{0} \times \streamFun(A_{1},A_{2},\ldots);\]
and this fixpoint is computed to be \(\streamFun(\stream{A}) = \prod_{n \in \naturals}^{\infty} A_{n}\).

In the same vein, we want to introduce not only streams but \emph{stream processes}
over a fixed theory of processes.
\begin{quoting}
\emph{``A stream process from
  \(\stream{X} = \streamExpr{X}\) to \(\stream{Y} = \streamExpr{Y}\)
  is a process from \(X_{0}\) to \(Y_{0}\) communicating along a channel $M$
  with a stream process from \(\tail{\stream{X}} = (X_{1}, X_{2}, \ldots)\) to
  \(\tail{\stream{Y}} = (Y_{1}, Y_{2}, \ldots)\).''}
\end{quoting}
Streams are recovered as stream processes on an empty input,
so we take this more general slogan as our preliminary definition of \emph{monoidal stream} (in~\Cref{def:intensionalmonoidalstream}).
Formally, they are the final fixpoint of the equation in \Cref{eq:intensionalstreamshort}.
\begin{figure}[!h]
  \centering \defining{linkistream}{}
  $\displaystyle \fun{T}(\stream{X},\stream{Y}) \cong \sum_{M \in \catC}
    \hom{}(X_{0}, M \tensor Y_{0}) \times \fun{T}(\act{M}{\tail{\stream{X}}}, \tail{\stream{Y}}).$
  \caption{Fixpoint equation for intensional streams.}
\label{eq:intensionalstreamshort}
\end{figure}
\begin{remark}[Notation]
  Let  $\stream{X} \in \NcatC$ be a sequence of objects $\streamExpr{X}$.
  We write $\defining{linktail}{\tail{\stream{X}}}$ for its \emph{tail} $(X_{1},X_{2},\dots)$.
  Given $M \in \catC$, we write $\act{M}{\stream{X}}$ for the sequence $(M \tensor X_{0},X_{1},X_{2},\dots)$;
  As a consequence, we write $\act{M}{\tail{\stream{X}}}$ for $(M \otimes X_{1},X_{2},X_{3},\dots)$.

\end{remark}

\begin{definition}\label{def:intensionalmonoidalstream}\defining{linkintensionalmonoidalstream}{}
  The set of \emph{intensional monoidal streams} $\fun{T} \colon \NcatC^{op} \times \NcatC \to \Set$, depending on inputs and outputs, %
  is the final fixpoint of the equation in \Cref{eq:intensionalstreamshort}.
\end{definition}

\begin{remark}[Initial fixpoint]
  There exists an obvious fixpoint for the equation in
  \Cref{eq:intensionalstreamshort}: the constant empty set.
  This solution is the \emph{initial} fixpoint, a minimal solution.
  The \emph{final} fixpoint will be realized by the set of \emph{intensional sequences}.
\end{remark}

\subsection{Intensional sequences}
We now construct the set of \intensionalStreams{} explicitly (\Cref{th:intensionalcoalgebra}).
For this purpose, we generalize the ``stateful morphism sequences'' of Katsumata and Sprunger~\cite{katsumata19} from cartesian to arbitrary \symmetricMonoidalCategories{} (\Cref{def:intensionalsequence}). We derive a novel characterization of these ``sequences'' as the desired final fixpoint (\Cref{th:intensionalcoalgebra}).

In the work of Katsumata and Sprunger, a stateful sequence is a sequence of morphisms $f_{n} \colon M_{n-1} \times X_{n} \to M_{n} \times Y_{n}$ in a \cartesianMonoidalCategory{}.
These morphisms represent a process at each point in time $n=0,1,2,\dots$.
At each step \(n\), the process takes an input \(X_{n}\) and,
together with the stored memory \(M_{n-1}\), produces some output \(Y_{n}\) and
writes to a new memory \(M_{n}\). The memory is initially empty, with
\(M_{-1} \defn \mathbf{1}\) being the final object by convention.
We extend this definition to any \symmetricMonoidalCategory{}.
\begin{definition}\label{def:intensionalsequence}\defining{linkintensionalstatefulsequence}{}
  Let $\stream{X}$ and $\stream{Y}$ be two sequences of objects
  representing inputs and outputs, respectively. An
  \defining{linkstreamtransducer}{\emph{intensional sequence}} is a
  sequence of objects $\streamExpr{M}$ together with a sequence of morphisms
  \[\intseq{f_{n} \colon M_{n-1} \tensor X_{n} \to M_{n} \tensor Y_{n}},\]
  where, by convention, $M_{-1} \defn I$ is the unit of the monoidal category.
  In other words, the set of intensional sequences is
  \[\iSeq(\stream{X},\stream{Y}) \coloneqq
    \sum_{M \in [\naturals,\catC]} \prod_{n = 0}^{\infty} \hom{}(M_{n-1} \tensor X_{n} , M_{n} \tensor Y_{n}).\]
\end{definition}

We now prove that \intensionalSequences{} are the final fixpoint of the equation in \Cref{eq:intensionalstreamshort}.
The following~\Cref{th:intensionalcoalgebra} serves two purposes:
it gives an explicit final solution to this fixpoint equation and
it gives a novel universal property to \intensionalSequences{}.

\begin{theorem}\label{th:intensionalcoalgebra}\label{corollary:intensionalstreams}
\IntensionalSequences{} are the explicit construction of \intensionalStreams{}, $\iStream{} \cong \iSeq{}$.
  In other words, they are a fixpoint of the equation in \Cref{eq:intensionalstreamshort},
and they are the final such one.
\end{theorem}
\begin{proof}[Proof sketch]
  It is known that categories of functors over sets, such as $[\NcatC^{op}\times\NcatC,\Set]$,
have all limits.
  Adamek's theorem (\Cref{th:adamek}) states that, if the following limit is a fixpoint, it is indeed the final one.
  \begin{equation}\label{eq:limit-intensional}
    \lim_{n \in \naturals} \sum_{M_{0},\dots,M_{n}} \prod^{n}_{t=0} \hom{}(M_{t-1} \tensor X_{t}, M_{t} \tensor Y_{t})
  \end{equation}
  \hyperlink{linkconnectedlimits}{Connected limits commute with coproducts} and the limit of the nth-product is the infinite product. Thus,~\Cref{eq:limit-intensional} is isomorphic to
  $\iSeq(\stream{X},\stream{Y})$. It only remains to show that $\iSeq(\stream{X},\stream{Y})$ is
  a fixpoint, which means it should be isomorphic to the following expression.
  \begin{equation}\label{eq:fixpoint-expr-intensional}
    \sum_{M_{0}} \hom{}(X_{0},M_{0} \tensor Y_{0}) \times \sum_{M \in \NcatC} \prod_{n = 1}^{\infty} \hom{}(M_{n-1} \tensor X_{n}, M_{n} \tensor Y_{n}).
  \end{equation}
  Cartesian products distribute over coproducts, so~\Cref{eq:fixpoint-expr-intensional} is again isomorphic to $\iSeq(\stream{X},\stream{Y})$.
\end{proof}
\subsection{Interlude: Dinaturality}
\label{sec:dinaturality}

During the rest of this text, we deal with two different definitions of what it means for two processes to be equal: \emph{extensional} and \emph{observational equivalence}, apart from pure \emph{intensional equality}.
Fortunately, when working with functors of the form $P \colon \catC^{op} \times \catC \to \Set$, the so-called \emph{endoprofunctors}, we already have a canonical notion of equivalence.

Endoprofunctors $P \colon \catC^{op} \times \catC \to \Set$ can be thought as indexing families
of processes $P(M,N)$ by the types of an input channel $M$ and an output channel $N$. A process
$p \in P(M,N)$ writes to a channel of type $N$ and then reads from a channel of type $M$.

Now, assume we also have a transformation $r \colon N \to M$ translating from output to input types.
Then, we can \emph{plug the output to the input}:
the process $p$ writes with type $N$, then $r$ translates from $N$ to $M$, and then $p$ uses this same output as its input $M$.
This composite process can be given two sligthly different descriptions; the process could
\begin{itemize}
  \item translate \emph{after writing}, $P(M,r)(p) \in P(M,M)$, or
  \item translate \emph{before reading}, $P(r,N)(p) \in P(N,N)$.
\end{itemize}
These two processes have different types.
However, with the output plugged to the input, it does not really matter when to apply the translation.
These two descriptions represent the same process: they are \emph{dinaturally equivalent}.

\begin{definition}[Dinatural equivalence] \label{def:dinaturality}\defining{linkdinaturality}
  For any functor $P \colon \catC^{op} \times \catC \to \Set$, consider the set
  \[S_{P} \defn \sum_{M \in \catC} P(M,M).\]
  \emph{Dinatural equivalence}, $(\sim)$, on the set $S_{P}$
  is the smallest equivalence relation satisfying
  $P(M,r)(p) \sim P(r,N)(p)$ for each $p \in P(M,N)$ and each $r \in \hom{}(N,M)$.
\end{definition}

Coproducts quotiented by \dinaturalEquivalence{} construct a particular form of colimit called a \emph{coend}.
Under the process interpretation of profunctors, taking a coend means \emph{plugging an output to an input}
of the same type.

\begin{definition}[Coend, \cite{maclane78,loregian2021}]
  Let $P \colon \catC^{op} \times \catC \to \Set$ be a functor.
  Its \emph{coend} is the coproduct of $P(M,M)$ indexed by $M \in \catC$, quotiented by \dinaturalEquivalence{}.
  \[\coend{M\in\catC}P(M,M) \defn \left(\sum_{M \in \catC} P(M,M) \bigg/ \sim \right).\]
  That is, the coend is the colimit of the diagram with a \emph{cospan} $P(M,M) \gets P(M,N) \to P(N,N)$
  for each $f \colon N \to M$.
\end{definition}

\subsection{Towards extensional memory channels}

Let us go back to intensional monoidal streams.
Consider a family of processes $f_{n} \colon M_{n-1} \tensor X_{n} \to Y_{n} \tensor N_{n}$ reading from memories of type $M_{n}$ but writing to memories of type $N_{n}$.
Assume we also have processes $r_{n} \colon N_{n} \to M_{n}$ translating from output to input memory.
Then, we can consider the process that does \(f_{n}\), translates from memory \(N_{n}\) to memory \(M_{n}\) and then does \(f_{n+1}\).
This process is described by two different \intensionalSequences{},
\begin{itemize}
  \item $\intseq{f_{n} ; (r_{n} \otimes \im) \colon M_{n-1} \tensor X_{n} \to M_{n} \tensor Y_{n}}$, and
  \item $\intseq{(r_{n-1} \otimes \im) ; f_{n} \colon N_{n-1} \tensor X_{n} \to N_{n} \tensor Y_{n}}$.
\end{itemize}
These two \intensionalSequences{} have different types for the memory channels.
However, in some sense, they represent \emph{the same process description}.
If we do not care about what exactly it is that we save to memory, we should consider two such processes to be equal (as in \Cref{diagram:dataflowwalk2}, where ``the same process'' can keep two different values in memory).
Indeed, dinaturality in the memory channels $M_{n}$ is the smallest equivalence relation $(\sim)$ satisfying
\[(f_{n} ; (r_{n} \otimes \im))_{n \in \naturals} \sim ((r_{n-1} \otimes \im); f_{n})_{n \in \naturals}.\]
This is precisely the quotienting that we
perform in order to define \emph{extensional sequences}.

\begin{definition}\label{def:extensionalequality}\defining{linkextensionalequality}{}\label{def:extensional-sequence}
\emph{Extensional equivalence} of \intensionalSequences{}, $(\sim)$, is dinatural equivalence in the memory channels $M_{n}$.
An \defining{linkextensionalstatefulsequence}{\emph{extensional sequence}} from \(\stream{X}\) to \(\stream{Y}\) is an equivalence class
  \[\extseq{f_{n} \colon M_{n-1} \tensor X \to M_{n} \tensor Y}\]
of \intensionalSequences{} under extensional equivalence.

In other words, the set of extensional sequences is the set of intensional sequences substituting the coproduct by a coend,
  \[\eSeq(\stream{X},\stream{Y}) \defn \coend{M \in [\naturals,\catC]} \prod^{\infty}_{i=0} \hom{}(X_{i} \tensor M_{i-1}, Y_{i} \tensor M_{i}).\]
\end{definition}

\begin{figure}[h!]
\tikzset{every picture/.style={line width=0.85pt}} %
\begin{tikzpicture}[x=0.75pt,y=0.75pt,yscale=-1,xscale=1]
\draw   (190,110) -- (230,110) -- (230,130) -- (190,130) -- cycle ;
\draw   (170,140) -- (210,140) -- (210,160) -- (170,160) -- cycle ;
\draw   (160,110) -- (180,110) -- (180,130) -- (160,130) -- cycle ;
\draw    (190,160) -- (190,170) ;
\draw    (220,110) .. controls (220.2,97) and (239.4,97.4) .. (240,110) ;
\draw   (170,80) -- (210,80) -- (210,100) -- (170,100) -- cycle ;
\draw    (170,130) .. controls (169.57,139.23) and (180.14,133.51) .. (180,140) ;
\draw    (210,130) .. controls (209.57,139.23) and (200.14,133.51) .. (200,140) ;
\draw    (180,180) -- (180,195) ;
\draw    (180,180.65) .. controls (179.8,166.45) and (200.2,166.45) .. (200,180.65) ;
\draw  [fill={rgb, 255:red, 0; green, 0; blue, 0 }  ,fill opacity=1 ] (187.1,170) .. controls (187.1,168.4) and (188.39,167.1) .. (190,167.1) .. controls (191.6,167.1) and (192.9,168.4) .. (192.9,170) .. controls (192.9,171.6) and (191.6,172.9) .. (190,172.9) .. controls (188.39,172.9) and (187.1,171.6) .. (187.1,170) -- cycle ;
\draw    (200,180.65) .. controls (200,189.85) and (220,190) .. (220,180) ;
\draw    (190,100) .. controls (189.57,109.23) and (200.14,103.51) .. (200,110) ;
\draw    (220,180) .. controls (220.2,150.25) and (239.99,148.64) .. (240.01,111.71) ;
\draw [shift={(240,110)}, rotate = 89.11] [color={rgb, 255:red, 0; green, 0; blue, 0 }  ][line width=0.75]    (10.93,-3.29) .. controls (6.95,-1.4) and (3.31,-0.3) .. (0,0) .. controls (3.31,0.3) and (6.95,1.4) .. (10.93,3.29)   ;
\draw   (299.99,165) -- (339.99,165) -- (339.99,185) -- (299.99,185) -- cycle ;
\draw   (319.97,110) -- (359.97,110) -- (359.97,130) -- (319.97,130) -- cycle ;
\draw   (309.97,80) -- (329.97,80) -- (329.97,100) -- (309.97,100) -- cycle ;
\draw   (279.99,135.56) -- (319.99,135.56) -- (319.99,155.56) -- (279.99,155.56) -- cycle ;
\draw    (319.97,100) .. controls (319.55,109.23) and (330.12,103.51) .. (329.97,110) ;
\draw    (329.97,150) -- (329.99,165) ;
\draw    (319.99,185.65) .. controls (319.99,195.63) and (350.32,196.97) .. (349.99,185) ;
\draw    (349.99,185) .. controls (350.18,155.25) and (369.98,153.64) .. (370,116.71) ;
\draw [shift={(369.99,115)}, rotate = 89.11] [color={rgb, 255:red, 0; green, 0; blue, 0 }  ][line width=0.75]    (10.93,-3.29) .. controls (6.95,-1.4) and (3.31,-0.3) .. (0,0) .. controls (3.31,0.3) and (6.95,1.4) .. (10.93,3.29)   ;
\draw    (369.99,110) -- (369.99,115) ;
\draw    (339.98,130) -- (339.98,140) ;
\draw    (329.99,150.65) .. controls (329.79,136.45) and (350.19,136.45) .. (349.99,150.65) ;
\draw  [fill={rgb, 255:red, 0; green, 0; blue, 0 }  ,fill opacity=1 ] (337.08,140) .. controls (337.08,138.4) and (338.38,137.1) .. (339.98,137.1) .. controls (341.58,137.1) and (342.88,138.4) .. (342.88,140) .. controls (342.88,141.6) and (341.58,142.9) .. (339.98,142.9) .. controls (338.38,142.9) and (337.08,141.6) .. (337.08,140) -- cycle ;
\draw    (299.99,155) .. controls (299.56,164.23) and (310.13,158.51) .. (309.99,165) ;
\draw    (349.99,110) .. controls (350.19,97) and (369.39,97.4) .. (369.99,110) ;
\draw    (349.99,150) .. controls (349.99,171.97) and (360.32,171.63) .. (359.99,195) ;
\draw (210,120) node    {$+$};
\draw (170,120) node    {$0$};
\draw (190,150) node    {$fby$};
\draw (190,90) node    {$unif$};
\draw (319.98,175) node    {$+$};
\draw (319.97,90) node    {$0$};
\draw (339.97,120) node    {$fby$};
\draw (299.99,145.56) node    {$unif$};
\draw (251,123.4) node [anchor=north west][inner sep=0.75pt]    {$\ \sim $};
\end{tikzpicture}
\caption{Extensionally equivalent walks keeping different memories: the \emph{current} position vs. the \emph{next} position.}
\label{diagram:dataflowwalk2}
\end{figure}
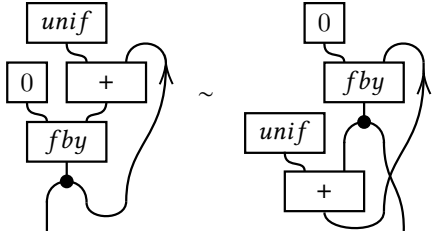

\section{Extensional Monoidal Streams}
\label{section:extensional}

In this section, we introduce \emph{extensional monoidal streams} in terms of a universal property:
\emph{extensional streams are the morphisms of the free delayed-feedback monoidal category} (\Cref{th:extensionalfreefeedback}).

Feedback monoidal categories come from the work of Katis, Sabadini and Walters~\cite{katis02}.
They are instrumental to our goal of describing and composing signal flow graphs: they axiomatize a graphical calculus that extends the well-known string diagrams for monoidal categories with \emph{feedback loops}~\cite{katis02,feedbackspans2020}.
Constructing the \emph{free feedback monoidal category}~(\Cref{def:stconstruction}) will lead to the main result of this section:
extensional sequences are the explicit construction of extensional streams (\Cref{th:extensionalfreefeedback}).

We finish the section by exploring how extensional equivalence may not be enough to capture true observational equality of processes (\Cref{example:observe}).

\subsection{Feedback monoidal categories}

\emph{Feedback monoidal categories} are \symmetricMonoidalCategories{} with a ``feedback'' operation that connects outputs to inputs.
They have a natural axiomatization (\Cref{def:feedback}) that has been rediscovered independently multiple times, with only slight variations~\cite{bloom93,katis02,katis99,bonchi19}.
It is weaker than that of \emph{traced monoidal categories}~\cite{feedbackspans2020} while still satisfying a normalization property (\Cref{th:free-feedback}).
We present a novel definition that generalizes the previous ones by allowing the feedback operator to be guarded by a \monoidalEndofunctor{}.

\begin{definition}\label{def:feedback}\defining{linkcategorywithfeedback}{}
  A \emph{feedback monoidal category} is a symmetric monoidal category $(\catC,\otimes,I)$ endowed with a monoidal endofunctor $\fun{F} \colon \catC \to \catC$ and an operation
  \[\defining{linkfbkop}{\ensuremath{\mathbf{fbk}}}_{S} \colon
     \idProf(\fun{F}(S) \tensor X, S \tensor Y) \to \idProf(X,Y)\]
  for all \(S\), \(X\) and \(Y\) objects of \(\catC\);
  this operation needs to satisfy the following axioms.
  \begin{enumerate}[ label={(A\arabic*).}, ref={\textbf{(A\arabic*)}}, start=1 ]
    \item\label{axiom:tight} Tightening: \(u \dcomp \fbk[S](f) \dcomp v = \fbk[S]((\id{\fun{F}S} \tensor u) \dcomp f \dcomp(\id{S} \tensor v))\).
    \item\label{axiom:vanish} Vanishing: \(\fbk[\monoidalunit](f) = f\).
    \item\label{axiom:join} Joining: \(\fbk[T](\fbk[S](f)) = \fbk[S \tensor T](f)\)
    \item\label{axiom:strength} Strength: \(\fbk[S](f) \tensor g = \fbk[S](f \tensor g)\).
    \item\label{axiom:slide} Sliding: \(\fbk[S]((\fun{F}h \tensor \id{X}) \dcomp f) = \fbk[T](f \dcomp (h \tensor \id{Y}))\).
  \end{enumerate}
\end{definition}

\begin{figure}[H]
\tikzset{every picture/.style={line width=0.85pt}} %
\begin{tikzpicture}[x=0.75pt,y=0.75pt,yscale=-1,xscale=1]
\draw    (40,60) -- (40,70) ;
\draw   (30,70) -- (70,70) -- (70,90) -- (30,90) -- cycle ;
\draw    (10,40.35) .. controls (10,20.2) and (40,20.2) .. (40,40.35) ;
\draw   (29.5,40) -- (49.5,40) -- (49.5,60) -- (29.5,60) -- cycle ;
\draw    (10,90) .. controls (10,109) and (40.4,110.2) .. (40,90) ;
\draw    (10,62) -- (10,90) ;
\draw [shift={(10,60)}, rotate = 90] [color={rgb, 255:red, 0; green, 0; blue, 0 }  ][line width=0.75]    (10.93,-3.29) .. controls (6.95,-1.4) and (3.31,-0.3) .. (0,0) .. controls (3.31,0.3) and (6.95,1.4) .. (10.93,3.29)   ;
\draw    (60,20) -- (60,70) ;
\draw    (60,90) -- (60,110) ;
\draw    (140,60) -- (140,70) ;
\draw   (130,40) -- (170,40) -- (170,60) -- (130,60) -- cycle ;
\draw    (110,40.35) .. controls (109.6,19.8) and (140.4,19.8) .. (140,40.35) ;
\draw   (129.5,70) -- (149.5,70) -- (149.5,90) -- (129.5,90) -- cycle ;
\draw    (110,90) .. controls (110.4,109.4) and (140.4,109.8) .. (140,90) ;
\draw    (110,40) -- (110,90) ;
\draw    (160,60) -- (160,110) ;
\draw    (160,20) -- (160,40) ;
\draw    (10,40) -- (10,80) ;
\draw    (110,62) -- (110,90) ;
\draw [shift={(110,60)}, rotate = 90] [color={rgb, 255:red, 0; green, 0; blue, 0 }  ][line width=0.75]    (10.93,-3.29) .. controls (6.95,-1.4) and (3.31,-0.3) .. (0,0) .. controls (3.31,0.3) and (6.95,1.4) .. (10.93,3.29)   ;
\draw (50,80) node    {$f$};
\draw (39.5,50) node    {$Fh$};
\draw (81,53.4) node [anchor=north west][inner sep=0.75pt]    {$=$};
\draw (150,50) node    {$f$};
\draw (139.5,80) node    {$h$};
\end{tikzpicture}
\caption{The sliding axiom (A5).}\label{diagram:sliding}
\end{figure}
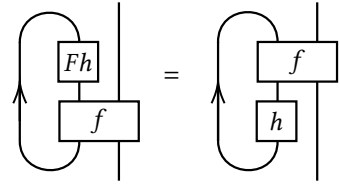
A \hyperlink{linkfeedbackfunctor}{\emph{feedback functor}} is a symmetric \monoidalFunctor{} that preserves the feedback structure (Appendix, \Cref{def:feedbackfunctor}).

\begin{remark}[Wait or trace]\label{remark:wait}\defining{linkmorphwait}
  In a \feedbackMonoidalCategory{} $(\catC,\fbk)$,
  we construct the morphism $\mathsf{wait}_{X} \colon X \to FX$ as a feedback loop over the symmetry,
  $\mathsf{wait}_{X} = \fbk(\sigma_{X,X}).$
  A \emph{traced monoidal category}~\cite{joyal96} is a \feedbackMonoidalCategory{} guarded by the identity functor
  such that $\mathsf{wait}_{X} = \im_{X}$.
\end{remark}

The ``state construction'', \(\St(\bullet)\), realizes the \emph{free} feedback monoidal category.
As it happens with feedback monoidal categories, this construction has appeared in the literature in slightly different forms. It has been used for describing a ``memoryful geometry of interaction''~\cite{hoshino14}, ``stateful resource calculi''~\cite{bonchi19}, and ``processes with feedback''~\cite{sabadini95,katis02}.

The idea in all of these cases is the same: we allow the morphisms of a monoidal category to depend on a ``state space'' $S$, possibly guarded by a functor.
Adding a state space is equivalent to freely adding feedback~\cite{feedbackspans2020}.

\begin{definition}
  A \emph{stateful morphism} is a pair $(S,f)$ consisting of a ``state space'' $S \in \catC$ and a morphism $f \colon \fun{F}S \tensor X \to S \tensor Y$.
  We say that two stateful morphisms are \emph{sliding equivalent} if they are related by the smallest equivalence relation satisfying
  $\reprCoend{(\fun{F}r \tensor \im) \dcomp h}{S} \sim \reprCoend{h \dcomp (r \tensor \im)}{T}$
  for each $h \colon X \tensor \fun{F}T \to S \tensor Y$ and each $r \colon S \to T$.

  In other words, sliding equivalence is \emph{dinaturality in $S$}.
\end{definition}

\begin{definition}[St($\bullet$) construction,~\cite{katis02,feedbackspans2020}]\label{def:stconstruction}
  We write \(\defining{linkSt}{\ensuremath{\mathsf{St}}}_{\fun{F}}(\catC)\) for the \symmetricMonoidalCategory{} that has the same objects as \(\catC\) and whose morphisms from $X$ to $Y$ are stateful morphisms $f \colon \fun{F}S \tensor X \to S \tensor Y$ \emph{up to sliding}.
  \[\hom{\St_{\fun{F}}(\catC)}(X,Y) \coloneqq \coend{S \in \catC} \idProf[\catC](\fun{F}S \tensor X, S \tensor Y).\]
\end{definition}

\begin{theorem}[see~\cite{katis02}]\label{th:free-feedback}
  \(\St_{\fun{F}}(\catC)\) is the free feedback monoidal category over \((\catC,\fun{F})\).
\end{theorem}

\subsection{Extensional monoidal streams}

Monoidal streams should be, in some sense, the minimal way of adding feedback to a theory of processes.
The output of this feedback, however, should be delayed by one unit of time: the category $\NcatC$ is naturally equipped with a \emph{delay} endofunctor that shifts a sequence by one.
Extensional monoidal streams form the free delayed-feedback category.

\begin{definition}[Delay functor]\label{def:delay-fun-cn}
  Let \(\defining{linkdelay}{\delay} \colon \sequencesCat{\catC} \to \sequencesCat{\catC}\) be the endofunctor defined on objects \(\stream{X} = \streamExpr{X}\), as \(\delay[\stream{X}] \defn \streamExprExp{\monoidalunit, X_0, X_1}\); and on morphisms \(\stream{f} = \streamExpr{f}\) as \(\delay[\stream{f}] \defn \streamExprExp{\id{\monoidalunit}, f_0, f_1}\).
\end{definition}

\begin{definition}\label{def:extensionalmonoidalstreams}
  \defining{linkestream}{}\defining{linkextensionalmonoidalstream}{}
  The set of \emph{extensional monoidal streams}, depending on inputs and outputs, $\eStream{} \colon \NcatC^{op} \times \NcatC \to \Set$, is the hom-set of the free feedback monoidal category over \((\NcatC, \delay)\).
\end{definition}

We characterize now extensional streams in terms of extensional sequences and the $\St(\bullet)$-construction.

\begin{theorem}\label{th:ext-stateful-sequences}\label{th:extensionalfreefeedback}
  \ExtensionalSequences{} are the explicit construction of \extensionalStreams{}, $\eStream \cong \eSeq$.
\end{theorem}
\begin{proof}
Note that $\eSeq(\stream{X},\stream{Y}) = \St_{\delay}(\NcatC)(\stream{X},\stream{Y})$.
That is, the \extensionalSequences{} we defined using dinaturality coincide with the morphisms of \(\St_{\delay}(\NcatC)\), the free feedback monoidal category over \((\NcatC,\delay)\) in \Cref{def:stconstruction}.
\end{proof}

As a consequence, the calculus of signal flow graphs given by the syntax of feedback monoidal categories is sound and complete for extensional equivalence over $\NcatC$.

\subsection{Towards observational processes}

\ExtensionalSequences{} were an improvement over \intensionalSequences{} because they allowed us to equate process descriptions that were \emph{essentially the same}.
However, we could still have two processes that are ``observationally the same'' without them being described in the same way.

\begin{remark}[Followed by]
  \defining{linkmorphfollowedby}{}
  As we saw in the Introduction, \emph{``followed by''} is a crucial operation in dataflow programming.
  Any sequence can be decomposed as $\stream{X} \cong \act{X_{0}}{\delay(\tail{\stream{X}})}$.\footnote{This can also be seen as the isomorphism making ``sequences'' a final coalgebra. That is, the first slogan we saw in~\Cref{section:fixpoint}.}
  We call \emph{``followed by''} to the coherence map in $\NcatC$ that witnesses this decomposition.
  \[\morphfby_{\stream{X}} \colon \act{X_{0}}{\delay{(\tail{\stream{X}})}} \to \stream{X}\]
  In the case of constant sequences $\stream{X} = (X,X,\dots)$, we have that $\tail{\stream{X}} = \stream{X}$; which means that ``followed by'' has type
  $\morphfby_{\stream{X}} \colon \act{X}{\delay\stream{X}} \to \stream{X}$.

\end{remark}

\begin{example}\label{example:observe}
  Consider the \extensionalStatefulSequence{}, in any \cartesianMonoidalCategory{}, that saves the first input to memory without ever producing an output.
  \emph{Observationally}, this is no different from simply discarding the first input, $(\blackComonoidUnit)_{X} \colon X \to 1$.
  However, in principle, we cannot show that these are \emph{extensionally equal}, that is, $\fbk(\morphfby_{\stream{X}}) \neq (\blackComonoidUnit)_{X}$.

More generally, discarding the result of any stochastic or deterministic signal flow graph is, \emph{observationally}, the same as doing nothing (\Cref{figure:silentwalk}, consequence of \Cref{th:stochasticprocesses}).
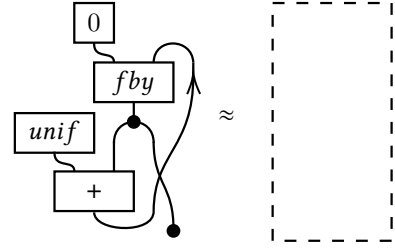
\begin{figure}
\tikzset{every picture/.style={line width=0.85pt}} %
\begin{tikzpicture}[x=0.75pt,y=0.75pt,yscale=-1,xscale=1]
\draw   (299.99,165) -- (339.99,165) -- (339.99,185) -- (299.99,185) -- cycle ;
\draw   (319.97,110) -- (359.97,110) -- (359.97,130) -- (319.97,130) -- cycle ;
\draw   (309.97,80) -- (329.97,80) -- (329.97,100) -- (309.97,100) -- cycle ;
\draw   (279.99,135.56) -- (319.99,135.56) -- (319.99,155.56) -- (279.99,155.56) -- cycle ;
\draw    (319.97,100) .. controls (319.55,109.23) and (330.12,103.51) .. (329.97,110) ;
\draw    (329.97,150) -- (329.99,165) ;
\draw    (319.99,185.65) .. controls (319.99,195.63) and (350.32,196.97) .. (349.99,185) ;
\draw    (349.99,185) .. controls (350.18,155.25) and (369.98,153.64) .. (370,116.71) ;
\draw [shift={(369.99,115)}, rotate = 89.11] [color={rgb, 255:red, 0; green, 0; blue, 0 }  ][line width=0.75]    (10.93,-3.29) .. controls (6.95,-1.4) and (3.31,-0.3) .. (0,0) .. controls (3.31,0.3) and (6.95,1.4) .. (10.93,3.29)   ;
\draw    (369.99,110) -- (369.99,115) ;
\draw    (339.98,130) -- (339.98,140) ;
\draw    (329.99,150.65) .. controls (329.79,136.45) and (350.19,136.45) .. (349.99,150.65) ;
\draw  [fill={rgb, 255:red, 0; green, 0; blue, 0 }  ,fill opacity=1 ] (337.08,140) .. controls (337.08,138.4) and (338.38,137.1) .. (339.98,137.1) .. controls (341.58,137.1) and (342.88,138.4) .. (342.88,140) .. controls (342.88,141.6) and (341.58,142.9) .. (339.98,142.9) .. controls (338.38,142.9) and (337.08,141.6) .. (337.08,140) -- cycle ;
\draw    (299.99,155) .. controls (299.56,164.23) and (310.13,158.51) .. (309.99,165) ;
\draw    (349.99,110) .. controls (350.19,97) and (369.39,97.4) .. (369.99,110) ;
\draw    (349.99,150) .. controls (349.99,171.97) and (360.32,171.63) .. (359.99,195) ;
\draw  [fill={rgb, 255:red, 0; green, 0; blue, 0 }  ,fill opacity=1 ] (357.09,195) .. controls (357.09,193.4) and (358.39,192.1) .. (359.99,192.1) .. controls (361.59,192.1) and (362.89,193.4) .. (362.89,195) .. controls (362.89,196.6) and (361.59,197.9) .. (359.99,197.9) .. controls (358.39,197.9) and (357.09,196.6) .. (357.09,195) -- cycle ;
\draw  [dash pattern={on 4.5pt off 4.5pt}] (410,80) -- (470,80) -- (470,200) -- (410,200) -- cycle ;
\draw (319.98,175) node    {$+$};
\draw (319.97,90) node    {$0$};
\draw (339.97,120) node    {$fby$};
\draw (299.99,145.56) node    {$unif$};
\draw (381,132.4) node [anchor=north west][inner sep=0.75pt]    {$\approx $};
\end{tikzpicture}
\caption{\emph{Observationally}, a silent process does nothing.}
\label{figure:silentwalk}
\end{figure}
\end{example}
\section{Observational Monoidal Streams}
\label{section:observational}

In this section, we introduce our definitive \emph{monoidal streams}: \emph{\observationalStreams{}} (\Cref{def:observationalmonoidalstream}).
Their explicit construction is given by \observationalSequences{}: \extensionalSequences{} quotiented by \observationalEquivalence{}.

Intuitively, two processes are observationally equal if they are ``equal up to stage \(n\)'', for any \(n \in \naturals\).
We show that, in sufficiently well-behaved monoidal categories (which we call \emph{productive}, \Cref{def:productive}), the set of observational sequences given some inputs and outputs is the final coalgebra of a fixpoint equation (\Cref{eq:observationalstreamshort}).
The name ``observational equivalence'' is commonly used to denote equality on the final coalgebra: \Cref{theorem:observationalfinalcoalgebra} justifies our use of the term.

\subsection{Observational streams}

We saw in~\Cref{sec:int-stateful-sequences} that we can define \intensionalSequences{} as a solution to a fixpoint equation.
We now consider the same equation, just substituting the coproduct for a coend.

\begin{definition}[Observational streams]\label{def:observationalmonoidalstream}\defining{linkobservationalmonoidalstream}{}
  The set of \emph{observational monoidal streams}, depending on inputs and outputs, %
  is the functor $\fun{Q} \colon \NcatC^{op} \times \NcatC \to \Set$ given by %
  the final fixpoint of the equation in \Cref{eq:observationalstreamshort}.
\end{definition}
\vspace{-1em}
\begin{figure}[!h]
  \centering
  $\displaystyle \fun{Q}(\stream{X},\stream{Y}) \cong \coend{M \in \catC}
    \hom{}(X_{0}, M \tensor Y_{0}) \times \fun{Q}(\act{M}{\tail{\stream{X}}}, \tail{\stream{Y}}).$
  \caption{Fixpoint equation for observational streams.}
\label{eq:observationalstreamshort}
\end{figure}

The explicit construction of this final fixpoint will be given by \observationalSequences{} (\Cref{theorem:observationalfinalcoalgebra}).

\subsection{Observational sequences}

We said that observational equivalence is ``equality up to stage $n$'', so our first step will be to define what it means to truncate an \extensionalSequence{} at any given $n \in \naturals$.

\begin{definition}[$n$-Stage process]\label{def:n-stages-process}
  An \emph{n-stage process} from inputs $\stream{X} = \streamExpr{X}$ to outputs
  $\stream{Y} = \streamExpr{Y}$ is an element of the set
  \[\defining{linkstage}{\ensuremath{\fun{Stage}_{n}}}(\stream{X},\stream{Y}) \defn \StageExpr{n}{X}{Y}.\]
\end{definition}

\begin{remark}\label{remark:stagenotation}
  In other words, $n$-stage processes are $n$-tuples $(f_{i} \colon M_{i-1} \tensor X_{i} \to M_{i} \tensor Y_{i})^{n}_{i=0}$, for some choice of $M_{i}$ \emph{up to dinaturality}, that we write as
  \[\bra{f_{0}|f_{1}|\dots|f_{n}} \in \Stage{n}(\stream{X},\stream{Y}).\]
  In this notation, \emph{dinaturality} means that morphisms can \emph{slide past the bars}.
  That is, for any $r_{i} \colon N_{i} \to M_{i}$ and any tuple, dinaturality says that
  \[\begin{aligned}
      &\bra{f_{0};(r_{0} \tensor \im)|f_{1};(r_{1} \tensor \im)|\dots|f_{n};(r_{n} \tensor \im)} \\
      & = \bra{f_{0}|(r_{0} \tensor \im);f_{1}|\dots|(r_{n-1} \tensor \im);f_{n}}.
    \end{aligned}\]
  Note that the last $r_{n}$ is removed by dinaturality.
\end{remark}

\begin{definition}[Truncation]
  The \emph{$k$-truncation} of an \extensionalSequence{} $\braket{f_{n} \colon M_{n-1} \tensor X_{n} \to M_{n} \tensor Y_{n}} \in \eSeq(\stream{X},\stream{Y})$ is  $\bra{f_{0}|\dots|f_{k}} \in \Stage{k}(\stream{X},\stream{Y})$.
  Truncation is well-defined under dinatural equivalence (\Cref{remark:stagenotation}).

  For $k \leq n$, the \emph{$k$-truncation} of an n-stage process given by $\bra{f_{0}|f_{1}|\dots|f_{n}} \in \Stage{n}(\stream{X},\stream{Y})$
  is $\bra{f_{0}|\dots|f_{k}} \in \Stage{k}(\stream{X},\stream{Y})$. This induces functions $\pi_{n,k} \colon \Stage{n}(\stream{X},\stream{Y}) \to \Stage{k}(\stream{X},\stream{Y})$, with
  the property that $\pi_{n,m} ; \pi_{m,k} = \pi_{n,k}$.
\end{definition}

\begin{definition}[\defining{linkobservationallyequal}{Observational equivalence}]\label{def:observationallyequal}
  Two extensional stateful sequences
  \[ \extseq{f}, \extseq{g} \in \coend{M \in [\naturals,\catC]} \prod^{\infty}_{i=0} \hom{}(M_{i-1} \tensor X_{i}, Y_{i} \tensor M_{i})\]
  are \emph{observationally equivalent} when all their n-stage truncations are equal. That is, $\bra{f_{0}|\dots | f_{n}} = \bra{g_{0}|\dots | g_{n}}$, for each $n \in \naturals$.
  We write this as $f \obsEqRel g$.
\end{definition}

\begin{remark}
  Formally, this is to say that the sequences \(\extseq{f}\) and \(\extseq{g}\) have the same image on the limit
  \[\lim\nolimits_{n} \Stage{n}(\stream{X},\stream{Y}),\]
  over the chain $\pi_{n,k} \colon \Stage{n}(\stream{X},\stream{Y}) \to \Stage{k}(\stream{X},\stream{Y})$.
\end{remark}

\begin{definition}\label{def:observationalsequence}
  An \defining{linkobservationalsequence}{\emph{observational sequence}} from \(\stream{X}\) to \(\stream{Y}\) is an
  equivalence class
  \[[\extseq{f_{n} \colon M_{n-1} \tensor X_{n} \to M_{n} \tensor Y_{n}}]_{\approx}\]
  of \extensionalSequences{} under \observationalEquality{}.
  In other words, the set of observational sequences is
  \[\oSeq(\stream{X},\stream{Y}) \cong \left(\coend{M \in [\naturals,\catC]} \prod^{\infty}_{i=0} \hom{}(M_{i-1} \tensor X_{i}, M_{i} \tensor Y_{i})\right)
  \bigg/\approx\]
\end{definition}
\subsection{Productive categories}
\label{section:productive}

The interaction between extensional and observational equivalence is of particular interest in some well-behaved categories that we call \emph{productive categories}.
In productive categories, observational sequences are the final fixpoint of an equation (\Cref{theorem:observationalfinalcoalgebra}), analogous to that of \Cref{sec:int-stateful-sequences}.

An important property of programs is \emph{termination}: a terminating program always halts in a finite amount of time.
However, some programs (such as servers, drivers) are not actually intended to terminate but to produce infinite output streams.
A more appropriate notion in these cases is that of \emph{productivity}:
a program that outputs an infinite stream of data is productive if each individual component of the stream is produced in finite time.
To quip, \emph{``a productive stream is a terminating first component followed by a productive stream''}.

The first component of our streams is only defined \emph{up to some future}.
It is an equivalence class $\alpha \in \Stage{1}(\stream{X},\stream{Y})$,
with representatives $\alpha_{i} \colon X_{0} \to M_{i} \tensor Y_{0}$.
But, if it does terminate, there is a process $\alpha_{0} \colon X_{0} \to M_{0} \tensor Y_{0}$ in our theory representing the process just until $Y_{0}$ is output.

\begin{definition}[Terminating component]
  A 1-stage process $\alpha \in \Stage{1}(\stream{X},\stream{Y})$ is \emph{terminating relative to $\catC$} if there exists $\alpha_{0} \colon X_{0} \to M_{0} \tensor Y_{0}$ such that each one of its representatives, $\bra{\alpha_{i}} = \alpha$, can be written as $\alpha_{i} = \alpha_{0} ; (s_{i} \tensor \im)$ for some $s_{i} \colon M_{0} \to M_{i}$.

The morphisms $s_{i}$ represent what is unique to each representative, and so we ask that,
  for any $u\colon M_{0} \tensor A \to U \tensor B$ and $v \colon M_{0} \tensor A \to V \tensor B$,
the equality $\bra{\alpha_{i} \tensor \im_{A} ; u \tensor \im_{Y_{0}}} = \bra{\alpha_{j} \tensor \im_{A} ; v \tensor \im_{Y_{0}}}$  implies
  $\bra{s_{i} \tensor \im_{A} ; u} = \bra{s_{j} \tensor \im_{A} ; v}$.
\end{definition}

\begin{definition}[\defining{linkproductive}{Productive category}]\label{def:productive}
  A symmetric monoidal category $(\catC,\tensor,\monoidalunit)$ is \emph{productive} when
every 1-stage process is terminating relative to $\catC$.
\end{definition}

\begin{remark}
  Cartesian monoidal categories are \productive{} (\Cref{prop:cartesianproductive}).
  \emph{Markov categories} \cite{fritz2020} with \emph{conditionals} and \emph{ranges} are \productive{} (\Cref{appendix:productivemarkov}).
  Free symmetric monoidal categories and compact closed categories are always productive.
\end{remark}

\begin{theorem}\label{theorem:observationalfinalcoalgebra}
  \ObservationalSequences{} are the explicit construction of \observationalStreams{} when the category is \productive{}.
More precisely, in a \productive{} category, the final fixpoint of the equation in \Cref{eq:observationalstreamshort} is given by the set of \observationalSequences{}, $\oSeq$.
\end{theorem}
\begin{proof}[Proof sketch]
  The terminal sequence for this final coalgebra is given by $\Stage{n}(\stream{X},\stream{Y})$.
  In productive categories, we can prove that the limit $\lim_{n} \Stage{n}(\stream{X},\stream{Y})$ is a fixpoint of the equation in \Cref{eq:observationalstreamshort} (\Cref{lemma:coalgebraexists}).
  Finally, in productive categories, observational sequences coincide with this limit (\Cref{appendix:theorem:observationalfinalcoalgebra}).
\end{proof}

\section{The Category of Monoidal Streams}\label{sec:monoidal-streams}

We are ready to construct $\STREAM$: the \feedbackMonoidalCategory{} of monoidal streams.
Let us recast the definitive notion of monoidal stream (\Cref{def:observationalmonoidalstream}) coinductively.

\begin{definition}%
\label{def:monoidalstream}\defining{linkmonoidalstream}
A \emph{monoidal stream} $f \in \STREAM(\stream{X},\stream{Y})$ is a triple consisting of
  \begin{itemize}
    \item $M(f) \in \obj{\catC}$, the \emph{memory},
    \item $\defining{linknow}{\ensuremath{\now}}(f) \in \hom{}(X_{0}, M(f) \tensor Y_{0})$, the \emph{first action},
    \item $\defining{linklater}{\ensuremath{\later}}(f) \in \STREAM(\act{M(f)}{\tail{\stream{X}}},\tail{\stream{Y}})$, the \emph{rest of the action},
  \end{itemize}
  quotiented by dinaturality in $M$.
\end{definition}

Explicitly, monoidal streams are quotiented by the equivalence relation $f \sim g$ generated by
\begin{itemize}
  \item the existence of $r \colon M(g) \to M(f)$,
  \item such that $\now(f) = \now(g) ; r$,
  \item and such that $\act{r}{\later(f)} \sim \later(g)$.
\end{itemize}
Here, $\act{r}{\later(f)} \in \STREAM(\act{M(g)}{\tail{\stream{X}}},\tail{\stream{Y}})$ is obtained by
precomposition of the first action of
$\later(f)$ with $r$.

\begin{remark}
  This is a coinductive definition of the functor
  \[\STREAM \colon \NcatC^{op} \times \NcatC \to \Set.\]
  In principle, arbitrary final coalgebras do not need to exist.
  Moreover, it is usually difficult to explicitly construct such coalgebras~\cite{adamek74}.
  However, in \productive{} categories, this coalgebra does exist and is constructed by \observationalSequences{}.
  From now on, we reason \emph{coinductively}~\cite{kozen17}, a style particularly suited for all the following definitions.
\end{remark}

\subsection{The symmetric monoidal category of streams}
The definitions for the operations of sequential and parallel composition are described in two steps.
We first define an operation that takes into account an extra \emph{memory channel} (\Cref{figure:memories}); we use this extra generality to strengthen the \emph{coinduction hypothesis}.
We then define the desired operation as a particular case of this coinductively defined one.

\begin{figure}[h!]
\tikzset{every picture/.style={line width=0.85pt}} %
\begin{tikzpicture}[x=0.75pt,y=0.75pt,yscale=-1,xscale=1]
\draw   (60,50) -- (100,50) -- (100,70) -- (60,70) -- cycle ;
\draw    (70,20) .. controls (69.67,40.43) and (49.67,30.43) .. (50,50) ;
\draw    (50,20) .. controls (50,40.43) and (70.33,30.43) .. (70,50) ;
\draw    (90,20) -- (90,50) ;
\draw    (70,70) .. controls (69.67,90.43) and (49.67,80.43) .. (50,100) ;
\draw    (50,70) .. controls (50,90.43) and (70.33,80.43) .. (70,100) ;
\draw   (56.5,100) -- (96.5,100) -- (96.5,120) -- (56.5,120) -- cycle ;
\draw    (90,70) -- (90,100) ;
\draw    (50,100) -- (50,140) ;
\draw    (50,50) -- (50,70) ;
\draw    (70,120) -- (70,140) ;
\draw    (90,120) -- (90,140) ;
\draw   (190,70) -- (230,70) -- (230,90) -- (190,90) -- cycle ;
\draw   (136,70) -- (176,70) -- (176,90) -- (136,90) -- cycle ;
\draw    (150,20) -- (150,70) ;
\draw    (170,20) .. controls (170,40.43) and (200.33,50.43) .. (200,70) ;
\draw    (200,20) .. controls (199.67,40.43) and (169.67,50.43) .. (170,70) ;
\draw    (220,20) -- (220,70) ;
\draw    (150,90) -- (150,140) ;
\draw    (220,90) -- (220,140) ;
\draw    (170,90) .. controls (170.01,110.43) and (200.33,120.43) .. (200,140) ;
\draw    (200,90) .. controls (199.67,110.43) and (169.67,120.43) .. (170,140) ;
\draw (80,60) node  [font=\small]  {${\textstyle \mathsf{now}( f)}$};
\draw (76.5,110) node  [font=\small]  {${\textstyle \mathsf{now}( g)}$};
\draw (50,16.6) node [anchor=south] [inner sep=0.75pt]  [font=\small]  {$A$};
\draw (70,16.6) node [anchor=south] [inner sep=0.75pt]  [font=\small]  {$B$};
\draw (89.99,16.6) node [anchor=south] [inner sep=0.75pt]  [font=\small]  {$X$};
\draw (90,143.4) node [anchor=north] [inner sep=0.75pt]  [font=\small]  {$Z$};
\draw (70,143.4) node [anchor=north] [inner sep=0.75pt]  [font=\small]  {$M_{g}$};
\draw (50,143.4) node [anchor=north] [inner sep=0.75pt]  [font=\small]  {$M_{f}$};
\draw (156,80) node  [font=\small]  {${\textstyle \mathsf{now}( f)}$};
\draw (210,80) node  [font=\small]  {${\textstyle \mathsf{now}( g)}$};
\draw (200,16.6) node [anchor=south] [inner sep=0.75pt]  [font=\small]  {$X$};
\draw (220,16.6) node [anchor=south] [inner sep=0.75pt]  [font=\small]  {$X'$};
\draw (150,16.6) node [anchor=south] [inner sep=0.75pt]  [font=\small]  {$A$};
\draw (170,16.6) node [anchor=south] [inner sep=0.75pt]  [font=\small]  {$B$};
\draw (170,143.4) node [anchor=north] [inner sep=0.75pt]  [font=\small]  {$M_{g}$};
\draw (150,143.4) node [anchor=north] [inner sep=0.75pt]  [font=\small]  {$M_{f}$};
\draw (200,143.4) node [anchor=north] [inner sep=0.75pt]  [font=\small]  {$Y$};
\draw (220,143.4) node [anchor=north] [inner sep=0.75pt]  [font=\small]  {$Y'$};
\end{tikzpicture}
\caption{String diagrams for the first action of sequential and parallel composition with memories.}
\label{figure:memories}
\end{figure}
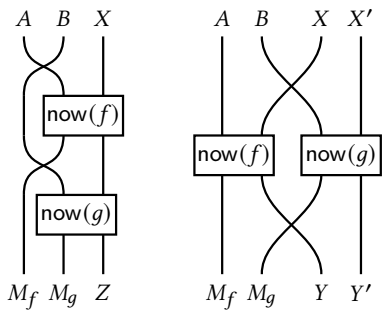

\begin{definition}[Sequential composition]\label{def:sequentialstream}
  Given two streams $f \in \STREAM(\act{\sA}{\stream{X}},\stream{Y})$ and
  $g \in \STREAM(\act{\sB}{\stream{Y}},\stream{Z})$, we compute
  $(\Ncomp{f}{\sA}{g}{\sB}) \in\STREAM(\act{(\sA \tensor \sB)}{\stream{X}},\stream{Z})$,
  their \emph{sequential composition with memories $\sA$ and $\sB$}, as
  \begin{itemize}
    \item $M(\Ncomp{f}{\sA}{g}{\sB}) = M(f) \tensor M(g)$,
    \item $\now(\Ncomp{f}{\sA}{g}{\sB}) = \sigma \dcomp (\now(\fm) \tensor \im) \dcomp \sigma \dcomp (\now(\gm) \tensor \im)$,
    \item $\later(\Ncomp{f}{\sA}{g}{\sB}) = \Ncomp{\later(f)}{M(f)}{\later(g)}{M(g)}$.
  \end{itemize}
  We write $(f \dcomp g)$ for $(\Ncomp{f}{\monoidalunit}{g}{\monoidalunit}) \in \STREAM(\stream{X},\stream{Z})$; the \emph{sequential composition} of $f  \in \STREAM(\stream{X},\stream{Y})$ and $g \in \STREAM(\stream{Y},\stream{Z})$.
\end{definition}

\begin{definition} The \emph{identity}
  $\im_{\stream{X}} \in \STREAM(\stream{X},\stream{X})$ is defined by $M(\im_{\stream{X}}) = \monoidalunit$,
  $\now(\im_{\stream{X}}) = \im_{X_{0}}$, and $\later(\im_{\stream{X}}) = \im_{\tail{\stream{X}}}$.
\end{definition}

\begin{definition}[Parallel composition]\label{def:parallelstream}
  Given two streams $f \in \STREAM(\act{\sA}{\stream{X}}, \stream{Y})$
  and $g \in \STREAM(\act{\sB}{\stream{X}'},\stream{Y}')$, we compute $(\Ntensor{f}{\sA}{g}{\sB}) \in \STREAM(\act{(\sA \tensor \sB)}{(\stream{X} \tensor \stream{X'})},\stream{Y}  \tensor \stream{Y'})$, their \emph{parallel composition
  with memories $\sA$ and $\sB$}, as
  \begin{itemize}
    \item $M(\Ntensor{f}{\sA}{g}{\sB}) = M(f) \tensor M(g)$,
    \item $\now(\Ntensor{f}{\sA}{g}{\sB}) = \sigma ; (\now(f)\tensor\now(g)); \sigma$,
    \item $\later(\Ntensor{f}{\sA}{g}{\sB}) = \Ntensor{\later(f)}{M(f)}{\later(g)}{M(g)}$.
  \end{itemize}
  We write $(f \tensor g)$ for $(\Ntensor{f}{\monoidalunit}{g}{\monoidalunit}) \in \STREAM(\stream{X} \tensor \stream{X'},\stream{Y}  \tensor \stream{Y'})$; we call it the \emph{parallel composition} of $f \in \STREAM(\stream{X},\stream{Y})$ and $g \in \STREAM(\stream{X}',\stream{Y}')$.
\end{definition}

\begin{definition}[Memoryless and constant streams]\label{def:inclusion}
  Each sequence \(\stream{f} = \streamExpr{f}\), with $f_{n} \colon X_{n} \to Y_{n}$, induces a stream $\lift{f} \in \STREAM(\stream{X}, \stream{Y})$ defined by $M(\lift{f}) \defn \monoidalunit$, $\now(\lift{f}) \defn f_{0}$, and $\later(\lift{f}) \defn \lift{\tail{\stream{f}}}$.
  Streams of this form are called \emph{memoryless}, i.e. their memories are given by the monoidal unit.

  Moreover, each morphism $f_{0} \colon X \to Y$ induces a \emph{constant} memoryless stream that we also call
  $f \in \STREAM(X,Y)$, defined by $M(\lift{f}) \defn \monoidalunit$, $\now(\lift{f}) \defn f_{0}$, and $\later(\lift{f}) \defn \lift{f}$.
\end{definition}

\begin{theorem}\label{th:category}
Monoidal streams over a \productive{} symmetric monoidal category $(\catC, \tensor, \monoidalunit)$ form a \symmetricMonoidalCategory{} $\STREAM$ with a symmetric monoidal identity-on-objects functor from $\NcatC$.
\end{theorem}
\begin{proof}
  Appendix, \Cref{th:monoidalstreamscategory}.
\end{proof}

\subsection{Delayed feedback for streams}

Monoidal streams form a \emph{delayed feedback} monoidal category.
Given some stream in $\STREAM(\delay \stream{S} \tensor \stream{X}, \stream{S} \tensor \stream{Y})$,
we can create a new stream in $\STREAM(\stream{X}, \stream{Y})$ that passes the output in $\stream{S}$ as a memory channel that gets used as the input in $\delay\stream{S}$.
As a consequence, the category of monoidal streams has a graphical calculus given by that of \categoriesWithFeedback{}. This graphical calculus is complete for extensional equivalence (as we saw in \Cref{th:ext-stateful-sequences}).

\begin{definition}[Delay functor]\label{def:delay-fun-stream}
  The functor from \Cref{def:delay-fun-cn} can be lifted to a monoidal functor $\mathbf{\delay} \colon \STREAM \to \STREAM$ that acts on objects in the same way.
  It acts on morphisms by sending a stream $f \in \STREAM(\stream{X},\stream{Y})$ to the stream given by $M(\delay f) = \monoidalunit$, $\now(\delay f) = \im_{\monoidalunit}$ and $\later(\delay f) = f$.
\end{definition}

\begin{definition}[Feedback operation]
  Given any morphism of the form $\fm \in \STREAM(\act{N}{\delay \stream{S} \tensor \stream{X}}, \stream{S} \tensor \stream{Y})$,
  we define $\fbk(f^{N}) \in \STREAM(\act{N}{\stream{X}}, \stream{Y})$ as
  \begin{itemize}
    \item $M(\fbk(f^{N})) = M(f) \tensor S_{0}$,
    \item $\now(\fbk(f^{N})) = \now(f)$ and
    \item $\later(\fbk(f^{N})) = \fbk(\later(f)^{M(f) \tensor S_{0}})$.
  \end{itemize}
  We write \(\fbk(f) \in \STREAM(\stream{X}, \stream{Y})\) for \(\fbk(f^{\monoidalunit})\), the feedback of $f \in \STREAM(\delay \stream{S} \tensor \stream{X}, \stream{S} \tensor \stream{Y})$
\end{definition}

\begin{theorem}\label{th:monoidalstreamsfeedback}
  \MonoidalStreams{} over a \symmetricMonoidalCategory{}
  $(\catC, \tensor, \monoidalunit)$ form a \(\delay\)-feedback monoidal category
  $(\STREAM, \fbk)$.
\end{theorem}
\begin{proof}
  Appendix, \Cref{th:appendix:monoidalstreamsfeedback}.
\end{proof}

\begin{corollary}\label{prop:functor-sequences-streams}
  There is a ``semantics'' identity-on-objects feedback monoidal functor
  \(\Semantics{} \colon \St_{\delay}{\NcatC} \to \STREAM\) from the free $\delay$-feedback monoidal category to the category of monoidal streams.
  Every \extensionalStatefulSequence{} $\extseq{f_{n} \colon M_{n-1} \tensor X_{n} \to Y_{n} \tensor M_{n}}$ gives a monoidal stream $\Semantics{}(f)$, which is defined by $M(\Semantics{}(f)) = M_{0}$,
  $$\now(\Semantics{}(f)) = f_{0},\mbox{ and }\later(\Semantics{}(f)) = \Semantics{}(\tail{f}),$$
  and this is well-defined.
  Moreover, this functor is full when \(\catC\) is \productive{}; it is not generally faithful.
\end{corollary}
\begin{proof}
  We construct $\Semantics{}$ from \Cref{th:monoidalstreamsfeedback,th:free-feedback}.
  Moreover, when \(\catC\) is \productive{}, by~\Cref{theorem:observationalfinalcoalgebra}, \monoidalStreams{} are \extensionalSequences{} quotiented by \observationalEquality{}, giving the fullness of the functor.
\end{proof}

\section{Cartesian Streams}
\label{section:classicalstreams}
Dataflow languages such as \Lucid{} or \Lustre{} \cite{wadge1985lucid,halbwachs1991lustre} can be thought
of as using an underlying cartesian monoidal structure: we can copy and discard variables and resources without affecting the order of operations.
These abilities correspond exactly to cartesianity thanks to Fox's theorem (\Cref{th:fox}, see \cite{fox76}).

\subsection{Causal stream functions}

In the cartesian case, there is available literature on the categorical semantics of dataflow programming languages~\cite{benveniste93,uustalu2008comonadic,cousot19,delpeuch19,oliveira84}.
Uustalu and Vene~\cite{uustalu05} provide elegant comonadic semantics to a \Lucid-like programming language using the non-empty list comonad.
In their framework, \emph{streams} with types \(\stream{X} = \streamExpr{X}\) are families of elements $\mathbf{1} \to X_{n}$.
\emph{Causal stream functions} from \(\stream{X} = \streamExpr{X}\) to \(\stream{Y} = \streamExpr{Y}\) are families of functions $f_{n} \colon X_{0} \times \dots \times X_{n} \to Y_{n}$.
Equivalently, they are, respectively, the states $(\stream{1} \to \stream{X})$ and morphisms $(\stream{X} \to \stream{Y})$ of the cokleisli category of the comonad $\neList \colon \NSet \to \NSet$ defined by
\[(\neList(\stream{X}))_{n} \coloneqq \prod^{n}_{i=0} X_{i}.\]

This comonad can be extended to other base categories, $\neList \colon \NcatC \to \NcatC$ \emph{only} as long as $\catC$ is cartesian.
Indeed, we can prove that the mere existence of such a comonad implies cartesianity of the base category.
For this, we introduce a refined version of Fox's theorem (\Cref{th:refinedfoxappendix}).

\begin{theorem}\label{th:nelist}
  Let $(\catC,\otimes,I)$ be a \symmetricMonoidalCategory{}.
  Let $\defining{linknelist}{\ensuremath{\neList}} \colon \NcatC \to \NcatC$ be the functor defined by
  \[\neList(X)_{n} \coloneqq \bigotimes^{n}_{i=0} X_{i}.\]
  This functor is monoidal, with oplaxators $\psi^{+}_{0} \colon \neList(I) \to I$ and $\psi_{X,Y} \colon \neList(X \otimes Y) \to \neList(X) \otimes \neList(Y)$ given by symmetries, associators and unitors.

  The monoidal functor $\neList \colon \NcatC \to \NcatC$ has a monoidal comonad structure if and only if its base monoidal category $(\catC, \otimes, I)$ is cartesian monoidal.
\end{theorem}
\begin{proof}[Proof sketch]
  The cartesian structure can be shown to make $\neList$ an opmonoidal comonad.
  Conversely, the opmonoidal comonad structure implies that every object should have a natural and uniform counital comagma structure.
  By our refined statement of Fox's theorem (\Cref{th:refinedfoxappendix}), this implies cartesianity. See Appendix, \Cref{ax:th:nelist}.
\end{proof}

This means that we cannot directly extend Uustalu and Vene's approach to the monoidal case.
However, we prove in the next section that our definition of \monoidalStreams{} particularizes to their \emph{causal stream functions}~\cite{uustalu05,katsumata19}.

\subsection{Cartesian monoidal streams}

The main claim of this section is that, in a \cartesianMonoidalCategory{}, \monoidalStreams{} instantiate to \emph{causal stream functions} (\Cref{th:cartesianstreams}). Let us fix such a category, $(\catC,\times,1)$.

The first observation is that the universal property of the cartesian product simplifies the fixpoint equation that defines monoidal streams.
This is a consequence of the following chain of isomorphisms, where we apply a \hyperlink{linkcoyoneda}{Yoneda reduction} to simplify the coend.
\[\begin{aligned}\label{eq:cartesianStreams}
  &\streamProf(\stream{X}, \stream{Y}) \cong \\
  &\textstyle{\coend{M} \idProf(X_{0}, M \times Y_{0}) \times \streamProf(\act{M}{\tail{\stream{X}}}, \stream{Y})} \cong \\
  &\textstyle{\coend{M} \idProf(X_{0}, M) \times \idProf(X_{0}, Y_{0}) \times  \streamProf(\act{M}{\tail{\stream{X}}}, \stream{Y})} \cong \\
  &\idProf(X_{0}, Y_{0}) \times  \streamProf(\act{X_{0}}{\tail{\stream{X}}}, \tail{\stream{Y}}).
\end{aligned}\]
Explicitly, the Yoneda reduction works as follows:
the first action of a stream $f \in \STREAM(\stream{X},\stream{Y})$ can be uniquely split as
$\now(f) = (f_{1},f_{2})$ for some $f_{1} \colon X_{0} \to Y_{0}$ and $f_{2} \colon X_{0} \to M(f)$.
Under the \emph{dinaturality} equivalence relation, $(\sim)$, we can always find a unique representative with $M = X_{0}$ and $f_{2} = \im_{X_{0}}$.

The definition of monoidal streams in the cartesian case is thus simplified (\Cref{prop:fixpoint-cartesian-streams}).
From there, the explicit construction of cartesian monoidal streams is straightforward.

\begin{definition}[Cartesian monoidal streams]\label{prop:fixpoint-cartesian-streams}
The set of \emph{cartesian monoidal streams}, given inputs $\stream{X}$ and outputs $\stream{Y}$, is the terminal fixpoint of the equation
  \[\streamProf(\stream{X}, \stream{Y}) \cong \idProf(X_{0}, Y_{0}) \times  \streamProf(\act{X_{0}}{\tail{\stream{X}}}, \tail{\stream{Y}}).\]
  In other words, a cartesian monoidal stream $f \in \STREAM(\stream{X},\stream{Y})$ is a pair consisting of
  \begin{itemize}
    \item $\fst(f) \in \hom{}(X_{0}, Y_{0})$, the \emph{first action}, and
    \item $\snd(f) \in \STREAM(\act{X_{0}}{\stream{\tail{X}}}, \stream{\tail{Y}})$, the \emph{rest of the action}.
  \end{itemize}
\end{definition}

\begin{theorem}\label{th:cartesianstreams}
  In the cartesian case, the final fixpoint of the equation in \Cref{eq:observationalstreamshort} is given by
  the set of causal functions,
  \begin{equation*}
    \streamProf(\stream{X}, \stream{Y}) = \prod_{n \in \naturals}^{\infty} \idProf(X_{0} \times \cdots \times X_{n}, Y_{n}).
  \end{equation*}
  That is, the category $\STREAM$ of \monoidalStreams{} coincides with the cokleisli monoidal category of the non-empty list monoidal comonad \(\neList \colon \NcatC \to \NcatC\).

\end{theorem}
\begin{proof}
  By Adamek's theorem~(\Cref{th:adamek}).
\end{proof}

\begin{corollary}
  Let $(\catC, \times, \mathbf{1})$ be a cartesian monoidal category. The category $\STREAM$ is cartesian monoidal.
\end{corollary}

\subsection{Example: the Fibonacci sequence}\label{example:fibonacci}
Consider $(\Set,\times,\mathbf{1})$, the \cartesianMonoidalCategory{} of small sets and functions.
And let us go back to the morphism $\mathsf{fib} \in \STREAM(\mathbf{1},\mathbb{N})$ that we presented in the Introduction (\Cref{figure:fibonacci}).
By \Cref{th:cartesianstreams}, a morphism of this type is, equivalently, a sequence of natural numbers. Using the previous definitions in \Cref{sec:monoidal-streams,section:classicalstreams}, \hyperlink{linkexamplefibonacci}{we can explicitly compute} this sequence to be $\mathsf{fib} = [0,1,1,2,3,5,8,\dots]$ (see the Appendix, \Cref{computation:fibonacci}).
 
\section{Stochastic Streams}
\label{section:stochastic-streams}

Monoidal categories are well suited for reasoning about probabilistic processes.
Several different categories of probabilistic channels have been proposed in the literature \cite{panangaden1999, baez2016, cho2019}.
They were largely unified by Fritz \cite{fritz2020} under the name of \emph{\MarkovCategories{}}.
For simplicity, we work in the discrete stochastic setting, i.e. in the Kleisli category of the finite distribution monad, $\Stoch$, but we will be careful to isolate the relevant structure of Markov categories that we use.

The main result of this section is that \emph{\controlledStochasticProcesses{}} \cite{fleming1975,ross1996stochastic} are precisely monoidal streams over $\Stoch$.
That is, controlled stochastic processes are the canonical solution over $\Stoch$ of the fixpoint equation in \Cref{eq:observationalstreamshort}.

\subsection{Stochastic processes}

We start by recalling the notion of \emph{stochastic process} from probability theory and its ``controlled'' version.
The latter is used in the context of \emph{stochastic control} \cite{fleming1975,ross1996stochastic}, where the user has access to the parameters or optimization variables of a probabilistic model.

A \emph{discrete stochastic process} is defined as a collection of random variables $Y_1, \dots Y_n$ indexed by discrete time.
At any time step $n$, these random variables are distributed according to some $p_n \in \distr(Y_1 \times \dots \times Y_n)$.
Since the future cannot influence the past, the marginal of $p_{n+1}$ over $Y_{n+1}$ must equal $p_{n}$.
When this occurs, we say that the family of distributions $(p_n)_{n \in \naturals}$ is \emph{causal}.

More generally, there may be some additional variables $X_1, \dots, X_n$ which we have control over.
In this case, a \emph{\controlledStochasticProcess{}} is defined as a collection of \emph{controlled random variables}
distributing according to $f_n \colon  X_1 \times \dots \times X_n \to \distr(Y_1, \dots, Y_n)$.
Causality ensures that the marginal of $f_{n+1}$ over $Y_{n+1}$ must equal $f_{n}$.

\begin{definition}\label{def:stoch-rep-process}\defining{linkstochasticprocess}{}
  Let \(\stream{X}\) and \(\stream{Y}\) be sequences of sets.
  A \emph{controlled stochastic process} \(f \colon \stream{X} \to \stream{Y}\) is a sequence of functions \[f_{n} \colon X_{n} \times \dots \times X_{1} \to \distr(Y_{n} \times \dots \times Y_{1})\] satisfying \emph{causality} (the \emph{marginalisation property}).
That is, such that \(f_{n}\) coincides with the marginal distribution of \(f_{n+1}\) on the first \(n\) variables,
making the diagram in \Cref{figure:commute} commute.
\begin{figure}[h]
\begin{tikzcd}
    X_0 \times \dots \times X_{n+1} \rar{f_{n+1}} \dar[swap]{\pi_{0,\dots,n}} & D(Y_0 \times \dots \times Y_{n+1}) \dar{D\pi_{0,\dots,n}} \\
    X_0 \times \dots \times X_{n}  \rar{f_{n}} & D(Y_0 \times \dots \times Y_{n})
  \end{tikzcd}
  \caption{Marginalisation for stochastic processes.}
  \label{figure:commute}
\end{figure}

Controlled stochastic processes with componentwise composition, identities and tensoring, are the morphisms of a symmetric monoidal category \(\Marg\).
\end{definition}

Stochastic monoidal streams and stochastic processes not only are the same thing but they compose in the same way: they are \emph{isomorphic} as categories.

\begin{theorem}\label{th:stochasticprocesses}
The category of stochastic processes $\StochProc$  is monoidally isomorphic to the category $\STREAM$ over $\Stoch$.
\end{theorem}
\begin{proof}[Proof sketch]
  Appendix, \Cref{theorem:obstochiso}. The proof of this result is non-trivial and relies on a crucial property concerning ranges in $\Stoch$.
  The proof is moreover written in the language of Markov categories where the property of ranges can be formulated in full abstraction.
\end{proof}

We expect that the theorem above can be generalised to interesting categories
of probabilistic channels over measurable spaces (such as the ones covered
in \cite{panangaden1999, cho2019, fritz2020}).

\begin{corollary}
  \(\StochProc\) is a \feedbackMonoidalCategory{}.
\end{corollary}

\subsection{Examples}

We have characterized in two equivalent ways the notion of controlled stochastic process.
This yields a categorical semantics for probabilistic dataflow programming:
we may use the syntax of feedback monoidal categories to specify simple stochastic programs and evaluate their semantics in $\StochProc$.

\begin{example}[Random Walk] \label{example:walk}

Recall the morphism $\mathsf{walk} \in \STREAM(\mathbf{1},\mathbb{Z})$ that we
depicted back in \Cref{string:walk}.

Here, $\mathsf{unif} \in \STREAM(\mathbf{1},\{-1,1\})$, is a uniform random generator
that, at each step, outputs either $1$ or $(-1)$. The output of this
uniform random generator is then added to the current position, and we declare
the starting position to be $0$.
Our implementation of this morphism, following the definitions from \Cref{sec:monoidal-streams}
(\Cref{implementation:walk}) is, by \Cref{th:stochasticprocesses}, a discrete stochastic process, and it produces samples like the following ones.
\[\begin{aligned}
& [0,1,0,-1,-2,-1,-2,-3,-2,-3,\dots] \\
& [0,1,2,1,2,1,2,3,4,5,\dots] \\
& [0,-1,-2,-1,-2,-1,0,-1,0,-1,\dots]
\end{aligned}\]
 \end{example}

\begin{example}[Ehrenfest model]\label{example:ehrenfest}
The Ehrenfest model \cite[\S 1.4]{kelly11reversibility} is a simplified model of particle diffusion.

\begin{figure}[h!]
\begin{minipage}{0.48\linewidth}
\tikzset{every picture/.style={line width=0.85pt}} %
\begin{tikzpicture}[x=0.75pt,y=0.75pt,yscale=-1,xscale=1]
\draw    (240,130) .. controls (240.4,111.8) and (300.67,118.5) .. (300,100) ;
\draw   (210,130) -- (250,130) -- (250,150) -- (210,150) -- cycle ;
\draw   (270,130) -- (310,130) -- (310,150) -- (270,150) -- cycle ;
\draw    (230.08,62.47) .. controls (231.97,102.56) and (260.33,130.96) .. (260,170) ;
\draw [shift={(230,60)}, rotate = 88.89] [color={rgb, 255:red, 0; green, 0; blue, 0 }  ][line width=0.75]    (10.93,-3.29) .. controls (6.95,-1.4) and (3.31,-0.3) .. (0,0) .. controls (3.31,0.3) and (6.95,1.4) .. (10.93,3.29)   ;
\draw    (230,150) -- (230,158.87) ;
\draw    (220,170) .. controls (219.8,155.8) and (240.2,155.8) .. (240,170) ;
\draw  [fill={rgb, 255:red, 0; green, 0; blue, 0 }  ,fill opacity=1 ] (227.1,158.87) .. controls (227.1,157.27) and (228.4,155.97) .. (230,155.97) .. controls (231.6,155.97) and (232.9,157.27) .. (232.9,158.87) .. controls (232.9,160.47) and (231.6,161.77) .. (230,161.77) .. controls (228.4,161.77) and (227.1,160.47) .. (227.1,158.87) -- cycle ;
\draw    (300,170) -- (300,201.13) ;
\draw    (290,150) -- (290,160) ;
\draw  [fill={rgb, 255:red, 0; green, 0; blue, 0 }  ,fill opacity=1 ] (287.1,158.87) .. controls (287.1,157.27) and (288.4,155.97) .. (290,155.97) .. controls (291.6,155.97) and (292.9,157.27) .. (292.9,158.87) .. controls (292.9,160.47) and (291.6,161.77) .. (290,161.77) .. controls (288.4,161.77) and (287.1,160.47) .. (287.1,158.87) -- cycle ;
\draw    (280,170) .. controls (279.8,155.8) and (300.2,155.8) .. (300,170) ;
\draw    (260,170) .. controls (260.67,184.17) and (279.67,185.5) .. (280,170) ;
\draw    (230,60) .. controls (229.8,45.8) and (250.2,45.8) .. (250,60) ;
\draw    (260,80) .. controls (260.14,92.14) and (280,117.83) .. (280,130) ;
\draw   (290,60) -- (330,60) -- (330,80) -- (290,80) -- cycle ;
\draw    (310,80) -- (310,90) ;
\draw  [fill={rgb, 255:red, 0; green, 0; blue, 0 }  ,fill opacity=1 ] (307.1,90) .. controls (307.1,88.4) and (308.4,87.1) .. (310,87.1) .. controls (311.6,87.1) and (312.9,88.4) .. (312.9,90) .. controls (312.9,91.6) and (311.6,92.9) .. (310,92.9) .. controls (308.4,92.9) and (307.1,91.6) .. (307.1,90) -- cycle ;
\draw    (300,100) .. controls (299.8,85.8) and (320.2,85.8) .. (320,100) ;
\draw    (200,170) .. controls (200.67,184.17) and (219.67,185.5) .. (220,170) ;
\draw    (240,170) -- (240,200) ;
\draw    (170,60) .. controls (169.8,45.8) and (190.2,45.8) .. (190,60) ;
\draw    (300,130) .. controls (300.33,111.17) and (320.67,126.83) .. (320,100) ;
\draw    (170.08,62.47) .. controls (171.96,102.57) and (200,131.29) .. (200,170) ;
\draw [shift={(170,60)}, rotate = 88.89] [color={rgb, 255:red, 0; green, 0; blue, 0 }  ][line width=0.75]    (10.93,-3.29) .. controls (6.95,-1.4) and (3.31,-0.3) .. (0,0) .. controls (3.31,0.3) and (6.95,1.4) .. (10.93,3.29)   ;
\draw   (180,60) -- (220,60) -- (220,80) -- (180,80) -- cycle ;
\draw   (190,20) -- (230,20) -- (230,40) -- (190,40) -- cycle ;
\draw  [fill={rgb, 255:red, 255; green, 255; blue, 255 }  ,fill opacity=1 ] (240,60) -- (280,60) -- (280,80) -- (240,80) -- cycle ;
\draw  [fill={rgb, 255:red, 255; green, 255; blue, 255 }  ,fill opacity=1 ] (250,20) -- (290,20) -- (290,40) -- (250,40) -- cycle ;
\draw    (210,40) -- (210,60) ;
\draw    (270,40) -- (270,60) ;
\draw    (200,80) .. controls (200.14,92.14) and (220,117.83) .. (220,130) ;
\draw (230,140) node  [font=\normalsize]  {$move$};
\draw (290,140) node  [font=\normalsize]  {$move$};
\draw (310,70) node  [font=\normalsize]  {$unif$};
\draw (200,70) node  [font=\normalsize]  {$fby$};
\draw (210,30) node  [font=\footnotesize]  {$( 1..4)$};
\draw (260,70) node  [font=\normalsize]  {$fby$};
\draw (270,30) node  [font=\footnotesize]  {$()$};
\end{tikzpicture}
\end{minipage}
\begin{minipage}{0.5\linewidth}
  \[\begin{gathered}
     \mathsf{ehr} \defn \\
     (1...4) \tensor () ;
     \fbk( \sigma ; \\
    \fby \tensor  \fby \tensor \mathsf{unif}; \\
    \im \tensor \im \tensor \COPY {;} \sigma ; \\
     \mathsf{move} \tensor \mathsf{move} {;} \\
     \COPY)
  \end{gathered}\]
\end{minipage}
\caption{Ehrenfest model: sig. flow graph and morphism.}
\label{example:ehrenfest}
\end{figure}
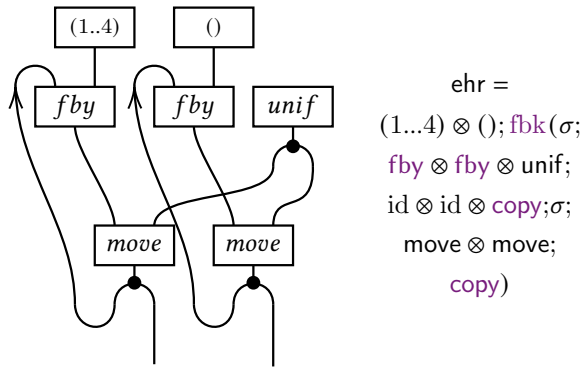

Assume we have two urns with 4 balls, labelled from 1 to 4.
Initially, the balls are all in the first urn.
We randomly (and uniformly) pick an integer from 1 to 4, and the ball labelled by that
number is removed from its box and placed in the other box.
We iterate the procedure, with independent uniform selections each time.

Our implementation of this morphism, following the definitions from \Cref{sec:monoidal-streams} (\Cref{implementation:ehrenfest}) yields samples such as the following.

$\begin{aligned}
& [([2,3,4],[1]),\ & ([3,4],[1,2]),\ && ([1,3,4],[2]), \\
& ([1,4],[2,3]),\ & ([1],[2,3,4]),\ && ([],[1,2,3,4]), \\
& ([2],[1,3,4]),\ & \dots]
\end{aligned}$
 \end{example}

\section{A dataflow programming language}
\label{section:typetheory}

In this section, we introduce the syntax for two \Lucid{}-like dataflow programming languages and their semantics in \monoidalStreams{}.
The first one is \emph{deterministic} and it takes semantics in the \feedbackMonoidalCategory{} of set-based monoidal streams.
The second one is \emph{stochastic} and it takes semantics in the \feedbackMonoidalCategory{} of stochastic processes, or stochastic monoidal streams.

We do so by presenting a type theory for \emph{feedback} monoidal categories (similar to~\cite{shulman21type,hasegawa97}). Terms of the type theory represent programs in our language.

\subsection{Type theory for monoidal categories}

We start by considering a \emph{type theory for \symmetricMonoidalCategories{}} over some generators forming a multigraph $\mathcal{G}$.
Instead of presenting a type theory from scratch, we extend the basic type theory for symmetric monoidal categories described by Shulman~\cite{shulman2016categorical}.
Details are in the Appendix (\Cref{appendix:section:typetheory}). Here, we only illustrate it with an example.

\begin{figure}[H]
\begin{mathpar}
  \infer[\textsc{Gen}]
    {f \in \mathcal{G}(A_1,\dots,A_n;B) \and
    \Gamma_1 \entails x_1 : A_1 \dots \Gamma_n \entails x_n : A_n}
    {\Shuf(\Gamma_1,\dots,\Gamma_n) \entails f(x_1,\dots,x_n) : B} \\
  \infer[\textsc{Pair}]
  {\Gamma_1 \entails x_1 : A_1\ \dots\ \Gamma_n \entails x_n : A_n}
  {\Shuf(\Gamma_1,\mydots,\Gamma_n) \entails [x_1,\mydots,x_n] : A_1 \otimes \mydots \otimes A_n}  \and
  \infer[\textsc{Var}]{ }{x : A\entails x : A} \\
  \infer[\textsc{Split}]{\Delta \entails m : A_1 \otimes \dots \otimes A_n \and \Gamma, x_1 : A_1, \dots , x_n : A_n \entails z : C}{ \Shuf(\Gamma,\Delta) \entails \textsc{Split}\ m \to [x_1,\dots,x_n]\ \textsc{in}\ z : C} \\
\end{mathpar}
\caption{Type theory of symm. monoidal categories~\cite{shulman2016categorical}.}
\end{figure}

The type theory for symmetric monoidal categories is \emph{linear}~\cite{girard87,seely87,lincoln92} in the sense that
any introduced variable must be used exactly once.
This is for a good reason: monoidal categories represent linear theories of processes, where \emph{copying} and \emph{discarding} may not be allowed in general.

\begin{example}\label{ex:monoidaltypes}
  In a monoidal category,
  let $f \colon X \tensor U \to Z$, $g \colon I \to U \tensor V \tensor W$ and $h \colon V \tensor Y \to I$.
  The following is a string diagram together with its term in the type theory.
  
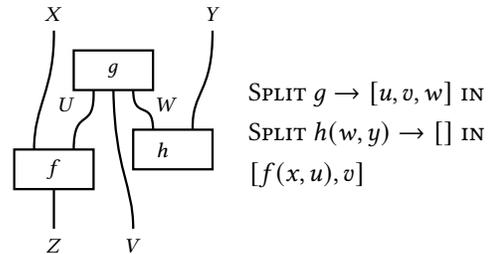
\begin{figure}[h!]
    \tikzset{every picture/.style={line width=0.85pt}} %
\begin{minipage}{0.3\linewidth}
\begin{tikzpicture}[x=0.75pt,y=0.75pt,yscale=-1,xscale=1]
\draw   (60,60) -- (100,60) -- (100,80) -- (60,80) -- cycle ;
\draw   (30,110) -- (70,110) -- (70,130) -- (30,130) -- cycle ;
\draw    (50,130) -- (50,150) ;
\draw   (90,100) -- (130,100) -- (130,120) -- (90,120) -- cycle ;
\draw    (50,50) .. controls (50.4,80.2) and (40,87.4) .. (40,110) ;
\draw    (70,80) .. controls (70.8,101.8) and (60,87.4) .. (60,110) ;
\draw    (90,80) .. controls (90,95.8) and (100,84.2) .. (100,100) ;
\draw    (130,50) .. controls (130.2,81.8) and (119.4,80.2) .. (120,100) ;
\draw    (80,80) .. controls (80.8,121.4) and (90,127.4) .. (90,150) ;
\draw (80,70) node  [font=\footnotesize]  {$g$};
\draw (105,110) node  [font=\footnotesize]  {$h$};
\draw (50,120) node  [font=\footnotesize]  {$f$};
\draw (50,46.6) node [anchor=south] [inner sep=0.75pt]  [font=\footnotesize]  {$X$};
\draw (130,46.6) node [anchor=south] [inner sep=0.75pt]  [font=\footnotesize]  {$Y$};
\draw (50,153.4) node [anchor=north] [inner sep=0.75pt]  [font=\footnotesize]  {$Z$};
\draw (90,153.4) node [anchor=north] [inner sep=0.75pt]  [font=\footnotesize]  {$V$};
\draw (56.5,82.4) node [anchor=north] [inner sep=0.75pt]  [font=\footnotesize]  {$U$};
\draw (107,82.4) node [anchor=north] [inner sep=0.75pt]  [font=\footnotesize]  {$W$};
\end{tikzpicture}
\end{minipage} \begin{minipage}{0.5\linewidth}
  \qquad
  \[\begin{aligned}
      & \SPLIT{g}{u,v,w}{} \\
      & \SPLIT{h(w,y)}{}{} \\
      & [f(x,u),v]
  \end{aligned}\]
\end{minipage}
  \end{figure}
\end{example}

\subsection{Adding feedback}\label{sec:addingfeedback}
We now extend the theory with delay and feedback. We start by considering a $\partial$ operator on types, which extends to contexts inductively as $\partial [] = []$ and $\partial (\Gamma, x {:} A) = \partial \Gamma , (x : \partial A)$. We provide formation rules for introducing delay and feedback.
These need to satisfy equalities making $\Delay$ a functor and $\textsc{Fbk}$ a feedback operator.
  \begin{mathpar}
    \infer[\defining{linkDelay}{\textsc{Delay}}]{\Gamma \entails x : A}{\partial \Gamma \entails x : \partial A}
    \and
    \infer[\defining{linkFbk}{\textsc{Fbk}}]
    {\Gamma , s : \partial S \entails x(s) : S \tensor A}
    {\Gamma \entails \mbox{\textsc{Fbk}}\ s.\ x(s) : A}
  \end{mathpar}
As in \Cref{remark:wait}, we define $\defining{linkWait}{\ensuremath{\mbox{\textsc{Wait}}}}(x) \equiv \Fbk{y}{[x,y]}$.

\subsection{Adding generators}

In both versions of the language (deterministic and stochastic), we include ``copy'' and ``followed by'' operations,
representing the corresponding monoidal streams.
Copying does not need to be natural (in the stochastic case, it will not be) and it does not even need to form a comonoid.
\begin{mathpar}
  \infer[\defining{linkCopy}{\textsc{Copy}}]{\Gamma \entails x : A}{\Gamma \entails \Copy(x) : A \otimes A}
  \and
  \infer[\defining{linkFby}{\defining{linkFby}{\textsc{Fby}}}]{
    \Gamma \entails x : A \and
    \Delta \entails y : \partial(A)}{\Shuf(\Gamma,\Delta) \entails x\ \Fby\ y : A}
\end{mathpar}

In fact, recursive definitions make sense only when we have a \emph{copy} operation, that allows us to rewrite the definition as a feedback that ends with a copy. That is,
\[M = x(M) \quad\mbox{means}\quad M = \Fbk{m}{\Copy(x(m))}.\]
Moreover, in the deterministic version of our language we allow non-linearity: a variable can occur multiple times, implicitly copying it.

\subsection{Examples}

\begin{example}
Recall the example from the introduction (and \Cref{example:fibonacci}).
\[\fib =
   0\ \Fby\ (\fib + (1\ \Fby\ \Wait\ \fib))
\]
Its desugaring, following the previous rules, is below.
\[\begin{aligned}
  \fib =\ &
  \Fbk{f}{} \Copy \\
  & (0\ \Fby \\
  & \quad \SPLIT{\Copy(f)}{f_{1},f_{2}}{} \\
  &  \quad (f_{1} + 1\ \Fby\ \Wait(f_{2}))\ )\\
\end{aligned}\]
\end{example}

\begin{example}[Ehrenfest model]
    The Ehrenfest model described in \Cref{example:ehrenfest} has the following
    specification in the programming language.
    \[\begin{aligned}
        \URNS =\  & [(1,2,3,4),()]\ \Fby\ \\
        & \SPLIT{\URNS}{u_{1},u_{2}}{} \\
        & \SPLIT{\Copy(\uniform)}{n_{1},n_{2}}{} \\
        & [\Move(n_{1},u_{1}), \Move(n_{2},u_{2})]
      \end{aligned}\]
    Sampling twice from the same distribution is different from copying a
    single sample, and \textsc{Split} allows us to express this difference: instead
    of calling the $\uniform$ distribution twice, this program calls it
    once and then copies the result.
\end{example}
 
\section{Conclusions}
Monoidal streams are a common generalization of streams, causal functions and stochastic processes.
In the same way that streams give semantics to dataflow programming~\cite{wadge1985lucid,halbwachs1991lustre} with plain functions, monoidal streams give semantics to dataflow programming with monoidal theories of processes.
Signal flow graphs are a common tool to describe control flow in dataflow programming.
Signal flow graphs are also the natural string diagrams of feedback monoidal categories.
Monoidal streams form a feedback monoidal category, and signal flow graphs are a formal syntax to describe and reason about monoidal streams.
The second syntax we present comes from the type theory of monoidal categories, and it is inspired by the original syntax of dataflow programming.
We have specifically studied stochastic dataflow programming, but the same framework allows for \emph{linear}, \emph{quantum} and \emph{effectful} theories of resources.

The literature on dataflow and feedback is rich enough to provide multiple diverging definitions and approaches.
What we can bring to this discussion are universal constructions.
Universal constructions justify some mathematical object as \emph{the canonical object} satisfying some properties.
In our case, these exact properties are extracted from three, arguably under-appreciated, but standard category-theoretic tools: \emph{dinaturality}, \emph{feedback}, and \emph{coalgebra}.
\emph{Dinaturality}, profunctors and coends, sometimes regarded as highly theoretical developments, are the natural language to describe how processes communicate and compose.
\emph{Feedback}, sometimes eclipsed by trace in the mathematical literature, keeps appearing in multiple variants across computer science.
\emph{Coalgebra} is the established tool to specify and reason about stateful systems.

\subsection{Further work}

\paragraph{Other theories.}
Many interesting examples of theories of processes are not monoidal but just \emph{premonoidal categories}~\cite{jeffrey97,power02}. For instance, the kleisli categories of arbitrary monads, where effects (e.g. reading and writing to a global state) do not need to commute. %
Premonoidal streams can be constructed by restricting dinaturality to their \emph{centres}.
Another important source of theories of processes that we have not covered is that of \emph{linearly distributive} and \emph{*-autonomous categories} \cite{seely87,cockett1997,blute93,blute96}.

Within monoidal categories, we would like to make monoidal streams explicit in the cases of partial maps~\cite{cockett02} for dataflow programming with different clocks \cite{uustalu2008comonadic}, non-deterministic maps~\cite{broy2001algebra,lee09} and quantum processes~\cite{carette21}.
A final question we do not pursue here is \emph{expressivity}: the class of functions a monoidal stream can define.

\paragraph{The 2-categorical view.}

We describe the  morphisms of a category as a final coalgebra.
However, it is also straightforward to describe the 2-endofunctor that should give rise to this category as a final coalgebra itself.

\paragraph{Implementation of the type theory.}

Justifying that the output of monoidal streams is the expected one requires some computations, which we have already implemented separately in the Haskell programming language (Appendix, \Cref{section:implementation}).
Agda has similar foundations and supports the coinductive definitions of this text (\Cref{sec:monoidal-streams}).
It is possible to implement a whole interpreter for a \Lucid{}-like stochastic programming language with a dedicated parser, but that requires some software engineering effort that we postpone for further work.

\bibliography{bibliography}

\newpage
\appendix
\section{Monoidal categories}

\begin{definition}[\cite{maclane78}]
  A \defining{linkmonoidalcategory}{\textbf{monoidal category}},
  \[(\catC, \otimes, I, \alpha, \lambda, \rho),\] is a category $\catC$
  equipped with a functor $\tensor \colon \catC \times \catC \to \catC$,
  a unit $\sI \in \catC$, and three natural isomorphisms: the associator $\alpha_{\sA,\sB,\sC} \colon (\sA \tensor \sB) \tensor \sC \cong \sA \tensor (\sB \tensor \sC)$, the left unitor $\lambda_{\sA} \colon \sI \tensor \sA \cong \sA$ and
  the right unitor $\rho_{\sA} \colon \sA \tensor \sI \cong \sA$;
  such that $\alpha_{\sA,\sI,\sB} ; (\im_{\sA} \tensor \lambda_{\sB}) = \rho_{\sA} \tensor \im_{\sB}$ and
  $(\alpha_{A,B,C} \tensor \im) ; \alpha_{A,B \tensor C, D} ; (\im_{A} \tensor \alpha_{B,C,D}) = \alpha_{A\tensor B,C,D} ; \alpha_{A,B,C \tensor D}$.  A monoidal category is \emph{strict} if $\alpha$, $\lambda$ and $\rho$ are identities.
\end{definition}

\begin{definition}[Monoidal functor, \cite{maclane78}]\defining{linkmonoidalfunctor}{}
  Let \[(\catC,\tensor,\sI,\alpha^{\catC},\lambda^{\catC},\rho^{\catC})\mbox{ and } (\catD,\boxtimes,\sJ,\alpha^{\catD},\lambda^{\catD},\rho^{\catD})\] be \hyperlink{linkmonoidalcategory}{monoidal categories}.
  A \defining{linkmonoidalfunctor}{\emph{monoidal functor}} (sometimes called \emph{strong monoidal functor}) is a triple
  $(F,\varepsilon,\mu)$ consisting of a functor $F \colon \catC \to \catD$ and two natural
  isomorphisms $\varepsilon \colon \sJ \cong F(\sI)$ and $\mu \colon F(\sA \tensor \sB) \cong F(\sA) \boxtimes F(\sB)$;
  such that
  \begin{itemize}
    \item the associators satisfy \[\begin{aligned}
        &\alpha^{\catD}_{FA,FB,FC} ; (\im_{FA} \tensor \mu_{B,C}) ; \mu_{A,B \tensor C} \\
        = &\ (\mu_{A,B} \tensor \im_{FC}) ; \mu_{A \tensor B,C} ; F(\alpha^{\catC}_{A,B,C}),\end{aligned}\]
    \item the left unitor satisfies \[(\varepsilon \tensor \im_{FA}) ; \mu_{I,A} ; F(\lambda^{\catC}_{A}) = \lambda^{\catD}_{FA}\]
    \item the right unitor satisfies \[(\im_{FA} \tensor \varepsilon) ; \mu_{A,I} ; F(\rho^{\catC}_{FA}) = \rho^{\catD}_{FA}.\]
  \end{itemize}
  A monoidal functor is a \emph{monoidal equivalence} if it is moreover an equivalence of categories.  Two monoidal categories are monoidally equivalent if there exists a monoidal equivalence between them.
\end{definition}

During most of the paper, we omit all associators and unitors from monoidal categories, implicitly using the \emph{coherence theorem} for monoidal categories (\Cref{remark:usingcoherence}).

\begin{theorem}[Coherence theorem, \cite{maclane78}]%
  \label{theorem:coherence}
  Every monoidal category is monoidally equivalent to a strict monoidal category.
\end{theorem}

\begin{remark}
  \label{remark:usingcoherence}
Let us comment further on how we use the coherence theorem. Each time we have a
morphism $f \colon A \to B$ in a monoidal category, we have a corresponding
morphism $A \to B$ in its strictification. This morphism can be lifted to the
original category to uniquely produce, say, a morphism
$(\lambda_{A} ; f ; \inverse{\lambda_{B}}) \colon I \otimes A \to I \otimes B$. Each time the source and the target are clearly
determined, we simply write $f$ again for this new morphism.
\end{remark}

\begin{definition}[Symmetric monoidal category, \cite{maclane78}]
  A \defining{linksymmetricmonoidalcategory}{\emph{symmetric monoidal category}}
  $(\catC, \otimes, I, \alpha, \lambda, \rho, \sigma)$ is a monoidal category
  $(\catC, \otimes, I, \alpha, \lambda, \rho)$ equipped with a braiding
  $\sigma_{A,B} \colon A \otimes B \to B \otimes A$, which satisfies the hexagon
  equation
  \[\alpha_{A,B,C} ; \sigma_{A,B \tensor C} ; \alpha_{B,C,A} = (\sigma_{A,B} \tensor \im) ; \alpha_{B,A,C} ; (\im \tensor \sigma_{A,C})\]
  and additionally satisifes $\sigma_{A,B} ; \sigma_{B,A} = \im$.
\end{definition}

\begin{remark}[Notation]
  We omit symmetries when this does not cause confusion.
  We write $\tid{a}$ for the morphism $a$ tensored with some identities when these can be deduced from the context.
  For instance, let $f \colon A \to B$, let $h \colon B \to D$ and let $g \colon B \tensor D \to E$.
  We write $\tid{f;h};g$ for the morphism $(f \tensor \im) ; \sigma ; (\im \tensor h) ; g$,
  which could have been also written as $(f \tensor \im) ; (h \tensor \im) ; \sigma  ; g$
\end{remark}

\begin{definition}[\cite{maclane78}]
  A \defining{linksymmetricmonoidalfunctor}{symmetric monoidal functor}
  between two \hyperlink{linksymmetricmonoidalcategory}{symmetric monoidal categories} $(\catC, \sigma^{\catC})$
  and $(\catD, \sigma^{\catD})$ is a monoidal functor $F \colon \catC \to \catD$ such that $\sigma^{\catD} ; \mu = \mu ; F(\sigma^{\catC})$.
\end{definition}

\begin{definition}
  A \defining{linkcartesianmonoidalcategory}{\emph{cartesian monoidal category}} is
  a monoidal category whose tensor is the categorical product and whose unit is
  a terminal object.
\end{definition}

\begin{definition}\label{def:feedbackfunctor}
  A \defining{linkfeedbackfunctor}{\emph{feedback functor}}
  between two \feedbackMonoidalCategories{}
  $(\catC,\fun{F}^{\catC},\fbk^{\catC})$ and $(\catD,\fun{F}^{\catD},\fbk^{\catD})$ is a symmetric monoidal functor $G \colon \catC \to \catD$ such that $G ; \fun{F}^{\catD} = \fun{F}^{\catC} ; G$ and
  \[ G(\fbk[S]^\catC(\fm)) = \fbk[GS]^{\catD}(\mu_{\fun{F}S,A} ; Gf ; \mu^{-1}_{S,B}),\]
  for each $\fm \colon \fun{F}S \tensor X \to S \tensor Y$, where \(\mu_{A,B} \colon G(\sA) \tensor G(\sB) \to G(\sA \tensor \sB )\) is
  the structure morphism of the monoidal functor \(G\).
\end{definition}

\begin{theorem}[see~\cite{katis02}]\label{ax:th:free-feedback}
  \(\St_{\fun{F}}(\catC)\) is the free category with feedback over \((\catC,\fun{F})\).
\end{theorem}
\begin{proof}[Proof sketch]
  Let $(\catD,\fun{F}^{\catD},\fbk^{\catD})$ be any other symmetric monoidal category with an endofunctor,
  and let $H \colon \catC \to \catD$ be such that $\fun{F} ; H = H ; \fun{F}^{\catD}$.
  We will prove that it can be extended uniquely to a feedback functor $\tilde{H} \colon \St_{\fun{F}}(\catC) \to \catD$.

  It can be proven that any expression involving feedback can be reduced applying the feedback axioms to an expression of the form $\fbk(f)$ for some $f \colon \fun{F}S \tensor X \to S \tensor Y$. After this, the definition of $\tilde{H}$ in this morphism is forced to be $\tilde{H}(\fbk{f}) = \fbk^{\catD}(\catD)$. This reduction is uniquely up to sliding, and the morphisms of the $\St(\bullet)$ construction are precisely morphisms $f \colon \fun{F}S \tensor X \to S \tensor Y$ quotiented by sliding equivalence.  This is the core of the proof in \cite{katis02}.
\end{proof}
\subsection{Markov categories}

\begin{definition}
  \label{def:distributionmonad}\defining{linkfinitedistribution}
  The finite distribution commutative monad $\defining{linkdistr}{\ensuremath{\mathbf{D}}} \colon \Set \to \Set$ associates to each set the set of finite-support probability distributions over it.
  \[\distr(X) \defn
    \left\{ p \colon X \to [0,1] \ \\
      \middle|%
      \sum_{p(x) > 0}^{|\{ x \mid p(x) > 0\}| < \infty}p(x) = 1 \right\}.\]
  We call $\defining{linkstoch}{\ensuremath{\mathbf{Stoch}}}$ to the symmetric monoidal kleisli category of the finite distribution monad, $\kleisli{\distr}$.

  We write $f(y \vert x)$ for the probability $f(x)(y) \in [0, 1]$. Composition, $f ; g$, is defined by
    $$(f \dcomp g)(z \vert x) = \sum_{y \in Y} g(z \vert y) f(y \vert x).$$
  \end{definition}

  The cartesian product $(\times)$ in $\Set$ induces a monoidal (non-cartesian) product on $\kleisli{\distr}$.
That is, $\kleisli{\distr}$ has comonoids $(\blackComonoid)_X: X \to X \times X$ on every object, with $(\blackComonoidUnit)_X \colon X \to 1$ as counit.
However, contrary to what happens in $\Set$, these comultiplications are not natural: \emph{sampling and copying the result} is different from \emph{taking two independent samples}.

\begin{definition}[\defining{linkmarkovcategory}{Markov category}, {{\cite[Definition 2.1]{fritz2020}}}]
  A \emph{Markov category} $\catC$ is a \symmetricMonoidalCategory{} in which
  each object $X \in \catC$ has a cocommutative comonoid structure
  $(X, \varepsilon = \coUnit_{X} \colon X \to I, \delta = \coMult_{X} \colon X \to X \tensor X)$ with
  \begin{itemize}
    \item \emph{uniform} comultiplications, $\blackComonoid_{X \tensor Y} = (\blackComonoid_{X} \tensor \blackComonoid_{Y}) \tid{\sigma_{X,Y}}$;
    \item \emph{uniform} counits, $\coUnit_{X \tensor Y} = \coUnit_{X} \tensor \coUnit_{Y}$; and
    \item \emph{natural} counits, $f; \coUnit_{Y} = \coUnit_{X}$ for each $f \colon X \to Y$.
  \end{itemize}
  Crucially, comultiplications do not need to be natural.
\end{definition}

\begin{remark}[{{\cite[Remark 2.4]{fritz2020}}}]
  Any cartesian category is a Markov category. However, not any Markov category is cartesian, and the most interesting examples are those that fail to be cartesian, such as $\Stoch$.
The failure of comultiplication being natural makes it impossible to apply Fox's theorem (\Cref{th:fox}).
\end{remark}

  The structure of a \MarkovCategory{} is very basic.
  In most cases, we do need extra structure to reason about probabilities:
  this is the role of conditionals and ranges.

\begin{remark}[Notation]
  In a \MarkovCategory{}, given any $f \colon X_{0} \to Y_{0}$ and any $g \colon Y_{0} \tensor X_{0} \tensor X_{1} \to Y_{1}$,
  we write $(f \triangleleft g) \colon X_{0} \tensor X_{1} \to Y_{0} \tensor Y_{1}$ for the morphism defined by
  \[(f \triangleleft g) = \tid{(\coMult_{A}) ; f ; (\coMult_{B})} ; \tid{g},\]
  which is the string diagram in \Cref{figure:notationtriangle}.
  \begin{figure}[h!]
\tikzset{every picture/.style={line width=0.85pt}} %
\begin{tikzpicture}[x=0.75pt,y=0.75pt,yscale=-1,xscale=1]
\draw   (190,90) -- (210,90) -- (210,110) -- (190,110) -- cycle ;
\draw    (210,70) -- (210,78.87) ;
\draw    (190,130) -- (190,170) ;
\draw   (200,140) -- (240,140) -- (240,160) -- (200,160) -- cycle ;
\draw    (220,160) -- (220,170) ;
\draw    (230,70) -- (230,140) ;
\draw    (220,90) -- (220,140) ;
\draw    (210,130) -- (210,140) ;
\draw    (200,90) .. controls (199.8,75.8) and (220.2,75.8) .. (220,90) ;
\draw  [fill={rgb, 255:red, 0; green, 0; blue, 0 }  ,fill opacity=1 ] (207.1,78.87) .. controls (207.1,77.27) and (208.4,75.97) .. (210,75.97) .. controls (211.6,75.97) and (212.9,77.27) .. (212.9,78.87) .. controls (212.9,80.47) and (211.6,81.77) .. (210,81.77) .. controls (208.4,81.77) and (207.1,80.47) .. (207.1,78.87) -- cycle ;
\draw    (200,110) -- (200,118.87) ;
\draw    (190,130) .. controls (189.8,115.8) and (210.2,115.8) .. (210,130) ;
\draw  [fill={rgb, 255:red, 0; green, 0; blue, 0 }  ,fill opacity=1 ] (197.1,118.87) .. controls (197.1,117.27) and (198.4,115.97) .. (200,115.97) .. controls (201.6,115.97) and (202.9,117.27) .. (202.9,118.87) .. controls (202.9,120.47) and (201.6,121.77) .. (200,121.77) .. controls (198.4,121.77) and (197.1,120.47) .. (197.1,118.87) -- cycle ;
\draw (200,100) node  [font=\normalsize]  {$f$};
\draw (220,150) node  [font=\normalsize]  {$g$};
\draw (210,66.6) node [anchor=south] [inner sep=0.75pt]  [font=\small]  {$X_{0}$};
\draw (190,173.4) node [anchor=north] [inner sep=0.75pt]  [font=\small]  {$Y_{0}$};
\draw (220,173.4) node [anchor=north] [inner sep=0.75pt]  [font=\small]  {$Y_{1}$};
\draw (230,66.6) node [anchor=south] [inner sep=0.75pt]  [font=\small]  {$X_{1}$};
\end{tikzpicture}
\caption{The morphism $(f \triangleleft g)$.}
\label{figure:notationtriangle}
  \end{figure}
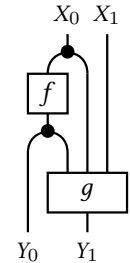
\end{remark}

\begin{proposition}
  Up to symmetries,
  \[(f \triangleleft g) \triangleleft h = f \triangleleft (g \triangleleft h).\]
  We may simply write $(f \triangleleft g \triangleleft h)$ for any of the two, omitting the symmetry.
\end{proposition}
\begin{proof}
  Using string diagrams (\Cref{figure:assoctriangle}). Note that $(\coMult)$ is coassociative and cocommutative.
\end{proof}

\begin{figure}
\tikzset{every picture/.style={line width=0.85pt}} %
\begin{tikzpicture}[x=0.75pt,y=0.75pt,yscale=-1,xscale=1]
\draw   (320,60) -- (340,60) -- (340,80) -- (320,80) -- cycle ;
\draw    (340,40) -- (340,48.87) ;
\draw    (320,100) -- (320,140) ;
\draw   (330,110) -- (370,110) -- (370,130) -- (330,130) -- cycle ;
\draw    (350,130) -- (350,140) ;
\draw    (360,40) -- (360,110) ;
\draw    (350,60) -- (350,110) ;
\draw    (340,100) -- (340,110) ;
\draw    (330,60) .. controls (329.8,45.8) and (350.2,45.8) .. (350,60) ;
\draw  [fill={rgb, 255:red, 0; green, 0; blue, 0 }  ,fill opacity=1 ] (337.1,48.87) .. controls (337.1,47.27) and (338.4,45.97) .. (340,45.97) .. controls (341.6,45.97) and (342.9,47.27) .. (342.9,48.87) .. controls (342.9,50.47) and (341.6,51.77) .. (340,51.77) .. controls (338.4,51.77) and (337.1,50.47) .. (337.1,48.87) -- cycle ;
\draw    (330,80) -- (330,88.87) ;
\draw    (320,100) .. controls (319.8,85.8) and (340.2,85.8) .. (340,100) ;
\draw  [fill={rgb, 255:red, 0; green, 0; blue, 0 }  ,fill opacity=1 ] (327.1,88.87) .. controls (327.1,87.27) and (328.4,85.97) .. (330,85.97) .. controls (331.6,85.97) and (332.9,87.27) .. (332.9,88.87) .. controls (332.9,90.47) and (331.6,91.77) .. (330,91.77) .. controls (328.4,91.77) and (327.1,90.47) .. (327.1,88.87) -- cycle ;
\draw    (360,20) -- (360,28.87) ;
\draw    (340,40) .. controls (339.8,25.8) and (380.2,25.8) .. (380,40) ;
\draw  [fill={rgb, 255:red, 0; green, 0; blue, 0 }  ,fill opacity=1 ] (357.1,28.87) .. controls (357.1,27.27) and (358.4,25.97) .. (360,25.97) .. controls (361.6,25.97) and (362.9,27.27) .. (362.9,28.87) .. controls (362.9,30.47) and (361.6,31.77) .. (360,31.77) .. controls (358.4,31.77) and (357.1,30.47) .. (357.1,28.87) -- cycle ;
\draw    (380,20) -- (380,28.87) ;
\draw    (360,40) .. controls (359.8,25.8) and (400.2,25.8) .. (400,40) ;
\draw  [fill={rgb, 255:red, 0; green, 0; blue, 0 }  ,fill opacity=1 ] (377.1,28.87) .. controls (377.1,27.27) and (378.4,25.97) .. (380,25.97) .. controls (381.6,25.97) and (382.9,27.27) .. (382.9,28.87) .. controls (382.9,30.47) and (381.6,31.77) .. (380,31.77) .. controls (378.4,31.77) and (377.1,30.47) .. (377.1,28.87) -- cycle ;
\draw    (380,40) -- (380,180) ;
\draw    (320,140) -- (320,148.87) ;
\draw    (310,160) .. controls (309.8,145.8) and (330.2,145.8) .. (330,160) ;
\draw  [fill={rgb, 255:red, 0; green, 0; blue, 0 }  ,fill opacity=1 ] (317.1,148.87) .. controls (317.1,147.27) and (318.4,145.97) .. (320,145.97) .. controls (321.6,145.97) and (322.9,147.27) .. (322.9,148.87) .. controls (322.9,150.47) and (321.6,151.77) .. (320,151.77) .. controls (318.4,151.77) and (317.1,150.47) .. (317.1,148.87) -- cycle ;
\draw    (350,140) -- (350,148.87) ;
\draw    (340,160) .. controls (339.8,145.8) and (360.2,145.8) .. (360,160) ;
\draw  [fill={rgb, 255:red, 0; green, 0; blue, 0 }  ,fill opacity=1 ] (347.1,148.87) .. controls (347.1,147.27) and (348.4,145.97) .. (350,145.97) .. controls (351.6,145.97) and (352.9,147.27) .. (352.9,148.87) .. controls (352.9,150.47) and (351.6,151.77) .. (350,151.77) .. controls (348.4,151.77) and (347.1,150.47) .. (347.1,148.87) -- cycle ;
\draw    (400,40) .. controls (401,90.17) and (389.67,121.5) .. (390,180) ;
\draw    (410,20) -- (410,40) ;
\draw    (410,40) .. controls (411,90.17) and (399,119.5) .. (400,180) ;
\draw    (330,160) .. controls (330,179.5) and (360,161.17) .. (360,180) ;
\draw    (360,160) .. controls (360.67,171.5) and (370,169.5) .. (370,180) ;
\draw   (350,180) -- (410,180) -- (410,200) -- (350,200) -- cycle ;
\draw    (310,160) -- (310,220) ;
\draw    (340,160) -- (340,220) ;
\draw    (380,200) -- (380,220) ;
\draw   (470,40) -- (490,40) -- (490,60) -- (470,60) -- cycle ;
\draw    (490,20) -- (490,28.87) ;
\draw    (480,40) .. controls (479.8,25.8) and (500.2,25.8) .. (500,40) ;
\draw  [fill={rgb, 255:red, 0; green, 0; blue, 0 }  ,fill opacity=1 ] (487.1,28.87) .. controls (487.1,27.27) and (488.4,25.97) .. (490,25.97) .. controls (491.6,25.97) and (492.9,27.27) .. (492.9,28.87) .. controls (492.9,30.47) and (491.6,31.77) .. (490,31.77) .. controls (488.4,31.77) and (487.1,30.47) .. (487.1,28.87) -- cycle ;
\draw    (480,60) -- (480,68.87) ;
\draw    (470,80) .. controls (469.8,65.8) and (490.2,65.8) .. (490,80) ;
\draw  [fill={rgb, 255:red, 0; green, 0; blue, 0 }  ,fill opacity=1 ] (477.1,68.87) .. controls (477.1,67.27) and (478.4,65.97) .. (480,65.97) .. controls (481.6,65.97) and (482.9,67.27) .. (482.9,68.87) .. controls (482.9,70.47) and (481.6,71.77) .. (480,71.77) .. controls (478.4,71.77) and (477.1,70.47) .. (477.1,68.87) -- cycle ;
\draw    (515,20) -- (515,85) ;
\draw    (540,20) -- (540,180) ;
\draw    (500,40) -- (500,90) ;
\draw    (500,80) -- (500,100) ;
\draw    (490,110) .. controls (489.8,95.8) and (510.2,95.8) .. (510,110) ;
\draw  [fill={rgb, 255:red, 0; green, 0; blue, 0 }  ,fill opacity=1 ] (497.1,100) .. controls (497.1,98.4) and (498.4,97.1) .. (500,97.1) .. controls (501.6,97.1) and (502.9,98.4) .. (502.9,100) .. controls (502.9,101.6) and (501.6,102.9) .. (500,102.9) .. controls (498.4,102.9) and (497.1,101.6) .. (497.1,100) -- cycle ;
\draw    (515,80) -- (515,100) ;
\draw    (505,110) .. controls (504.8,95.8) and (525.2,95.8) .. (525,110) ;
\draw  [fill={rgb, 255:red, 0; green, 0; blue, 0 }  ,fill opacity=1 ] (512.1,100) .. controls (512.1,98.4) and (513.4,97.1) .. (515,97.1) .. controls (516.6,97.1) and (517.9,98.4) .. (517.9,100) .. controls (517.9,101.6) and (516.6,102.9) .. (515,102.9) .. controls (513.4,102.9) and (512.1,101.6) .. (512.1,100) -- cycle ;
\draw    (475,110) .. controls (474.8,95.8) and (495.2,95.8) .. (495,110) ;
\draw  [fill={rgb, 255:red, 0; green, 0; blue, 0 }  ,fill opacity=1 ] (482.1,100) .. controls (482.1,98.4) and (483.4,97.1) .. (485,97.1) .. controls (486.6,97.1) and (487.9,98.4) .. (487.9,100) .. controls (487.9,101.6) and (486.6,102.9) .. (485,102.9) .. controls (483.4,102.9) and (482.1,101.6) .. (482.1,100) -- cycle ;
\draw    (470,80) .. controls (470.17,102.5) and (459.5,111.17) .. (460,130) ;
\draw    (490,77.1) .. controls (490.67,88.6) and (485,89.5) .. (485,100) ;
\draw   (465,130) -- (505,130) -- (505,150) -- (465,150) -- cycle ;
\draw    (505,110) .. controls (504.83,119.17) and (495.17,119.5) .. (495,130) ;
\draw    (490,110) .. controls (489.83,119.17) and (485.17,119.5) .. (485,130) ;
\draw    (495,110) .. controls (494.83,119.17) and (510.17,119.5) .. (510,130) ;
\draw    (510,110) .. controls (509.83,119.17) and (520.17,119.5) .. (520,130) ;
\draw    (525,110) .. controls (524.83,119.17) and (530.17,119.5) .. (530,130) ;
\draw    (475,110) -- (475,130) ;
\draw    (460,130) -- (460,170) ;
\draw    (520,130) -- (520,180) ;
\draw    (485,150) -- (485,158.87) ;
\draw    (475,170) .. controls (474.8,155.8) and (495.2,155.8) .. (495,170) ;
\draw  [fill={rgb, 255:red, 0; green, 0; blue, 0 }  ,fill opacity=1 ] (482.1,158.87) .. controls (482.1,157.27) and (483.4,155.97) .. (485,155.97) .. controls (486.6,155.97) and (487.9,157.27) .. (487.9,158.87) .. controls (487.9,160.47) and (486.6,161.77) .. (485,161.77) .. controls (483.4,161.77) and (482.1,160.47) .. (482.1,158.87) -- cycle ;
\draw    (530,130) -- (530,180) ;
\draw    (495,170) .. controls (494.5,176.5) and (510.17,173.17) .. (510,180) ;
\draw   (490,180) -- (550,180) -- (550,200) -- (490,200) -- cycle ;
\draw    (460,170) -- (460,215) -- (460,220) ;
\draw    (520,200) -- (520,220) ;
\draw    (475,170) .. controls (474.83,179.17) and (480.17,209.5) .. (480,220) ;
\draw    (510,130) -- (510,170) ;
\draw    (510,170) .. controls (509.5,176.5) and (495.17,173.17) .. (495,180) ;
\draw (330,70) node  [font=\normalsize]  {$f$};
\draw (350,120) node  [font=\normalsize]  {$g$};
\draw (380,190) node  [font=\normalsize]  {$h$};
\draw (480,50) node  [font=\normalsize]  {$f$};
\draw (485,140) node  [font=\normalsize]  {$g$};
\draw (520,190) node  [font=\normalsize]  {$h$};
\draw (360,16.6) node [anchor=south] [inner sep=0.75pt]  [font=\small]  {$X_{0}$};
\draw (380,16.6) node [anchor=south] [inner sep=0.75pt]  [font=\small]  {$X_{1}$};
\draw (410,16.6) node [anchor=south] [inner sep=0.75pt]  [font=\small]  {$X_{2}$};
\draw (310,223.4) node [anchor=north] [inner sep=0.75pt]  [font=\small]  {$Y_{0}$};
\draw (340,223.4) node [anchor=north] [inner sep=0.75pt]  [font=\small]  {$Y_{1}$};
\draw (380,223.4) node [anchor=north] [inner sep=0.75pt]  [font=\small]  {$Y_{2}$};
\draw (490,16.6) node [anchor=south] [inner sep=0.75pt]  [font=\small]  {$X_{0}$};
\draw (515,16.6) node [anchor=south] [inner sep=0.75pt]  [font=\small]  {$X_{1}$};
\draw (540,16.6) node [anchor=south] [inner sep=0.75pt]  [font=\small]  {$X_{2}$};
\draw (460,223.4) node [anchor=north] [inner sep=0.75pt]  [font=\small]  {$Y_{0}$};
\draw (480,223.4) node [anchor=north] [inner sep=0.75pt]  [font=\small]  {$Y_{1}$};
\draw (520,223.4) node [anchor=north] [inner sep=0.75pt]  [font=\small]  {$Y_{2}$};
\draw (425,117.4) node [anchor=north west][inner sep=0.75pt]    {$=$};
\end{tikzpicture}
\caption{Associativity, up to symmetries, of the triangle operation.}
\label{figure:assoctriangle}
\end{figure}
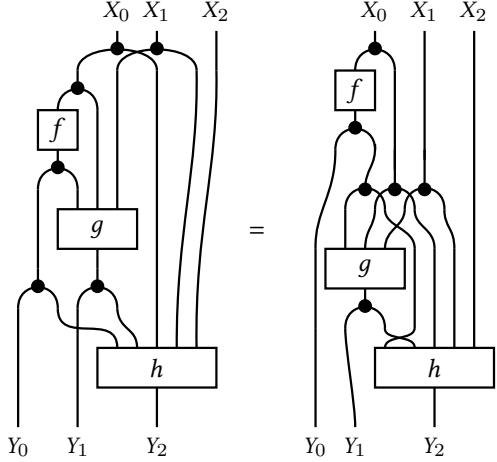

The Markov category $\Stoch$ also has \emph{conditionals}~\cite{fritz2020}, a property which we will use to prove the main result regarding stochastic processes.

\begin{definition}[Conditionals, {{\cite[Definition 11.5]{fritz2020}}}]
  Let $\catC$ be a \MarkovCategory{}. We say that $\catC$ has
  \emph{conditionals} if for every morphism $f \colon A \to X \tensor Y$,
  writing $f_{Y} \colon A \to X$ for its first projection, there
  exists $c_{f} \colon X \tensor A \to Y$ such that
  $f = f_{Y} \triangleleft c_{f}$ (\Cref{figure:conditionalsmarkov}).
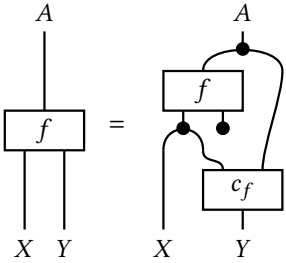
\begin{figure}
\tikzset{every picture/.style={line width=0.85pt}} %
\begin{tikzpicture}[x=0.75pt,y=0.75pt,yscale=-1,xscale=1]
\draw   (100,110) -- (140,110) -- (140,130) -- (100,130) -- cycle ;
\draw    (120,70) -- (120,110) ;
\draw    (110,130) -- (110,170) ;
\draw    (130,130) -- (130,170) ;
\draw    (220,70) -- (220,78.87) ;
\draw    (200,90) .. controls (199.8,75.8) and (240.2,75.8) .. (240,90) ;
\draw  [fill={rgb, 255:red, 0; green, 0; blue, 0 }  ,fill opacity=1 ] (217.1,78.87) .. controls (217.1,77.27) and (218.4,75.97) .. (220,75.97) .. controls (221.6,75.97) and (222.9,77.27) .. (222.9,78.87) .. controls (222.9,80.47) and (221.6,81.77) .. (220,81.77) .. controls (218.4,81.77) and (217.1,80.47) .. (217.1,78.87) -- cycle ;
\draw   (180,90) -- (220,90) -- (220,110) -- (180,110) -- cycle ;
\draw    (210,110) -- (210,118.87) ;
\draw  [fill={rgb, 255:red, 0; green, 0; blue, 0 }  ,fill opacity=1 ] (207.1,118.87) .. controls (207.1,117.27) and (208.4,115.97) .. (210,115.97) .. controls (211.6,115.97) and (212.9,117.27) .. (212.9,118.87) .. controls (212.9,120.47) and (211.6,121.77) .. (210,121.77) .. controls (208.4,121.77) and (207.1,120.47) .. (207.1,118.87) -- cycle ;
\draw    (190,110) -- (190,118.87) ;
\draw    (180,130) .. controls (179.8,115.8) and (200.2,115.8) .. (200,130) ;
\draw  [fill={rgb, 255:red, 0; green, 0; blue, 0 }  ,fill opacity=1 ] (187.1,118.87) .. controls (187.1,117.27) and (188.4,115.97) .. (190,115.97) .. controls (191.6,115.97) and (192.9,117.27) .. (192.9,118.87) .. controls (192.9,120.47) and (191.6,121.77) .. (190,121.77) .. controls (188.4,121.77) and (187.1,120.47) .. (187.1,118.87) -- cycle ;
\draw    (180,130) -- (180,170) ;
\draw   (200,140) -- (240,140) -- (240,160) -- (200,160) -- cycle ;
\draw    (200,130) .. controls (199.8,137.4) and (210.2,133) .. (210,140) ;
\draw    (240,90) .. controls (240.14,102.14) and (230,127.83) .. (230,140) ;
\draw    (220,160) -- (220,170) ;
\draw (120,120) node  []  {$f$};
\draw (200,100) node  []  {$f$};
\draw (220,150) node  []  {$c_{f}$};
\draw (151,113.4) node [anchor=north west][inner sep=0.75pt]    {$=$};
\draw (120,66.6) node [anchor=south] [inner sep=0.75pt]    {$A$};
\draw (220,66.6) node [anchor=south] [inner sep=0.75pt]    {$A$};
\draw (110,173.4) node [anchor=north] [inner sep=0.75pt]    {$X$};
\draw (130,173.4) node [anchor=north] [inner sep=0.75pt]    {$Y$};
\draw (180,173.4) node [anchor=north] [inner sep=0.75pt]    {$X$};
\draw (220,173.4) node [anchor=north] [inner sep=0.75pt]    {$Y$};
\end{tikzpicture}
\caption{Condititionals in a Markov category.}
\label{figure:conditionalsmarkov}
\end{figure}
\end{definition}

\begin{proposition}
  The Markov category $\Stoch$ has conditionals \cite[Example 11.6]{fritz2020}.
\end{proposition}
\begin{proof}
  Let $f \colon A \to X \tensor Y$. If $Y$ is empty, we are automatically done. If not, pick some $y_{0} \in Y$, and define
  \[c_{f}(y|x,a) = \left\{
	\begin{array}{ll}
		f(x,y|a) / \sum_{x \in X} f(x,y|a)  \\ \qquad \mbox{if } f(x,y|a) > 0 \mbox{ for some } x \in X, \\
		(y = y_{0}) \quad \mbox{otherwise}.
	\end{array}
    \right. \]
  It is straightforward to check that this does indeed define a distribution, and that it factors the original $f$ as expected.
\end{proof}

\begin{definition}[Ranges]
  In a Markov category, a \emph{range} for a morphism $f \colon A \to B$ is a morphism $r_{f} \colon A \tensor B \to A \tensor B$
  that
  \begin{enumerate}
    \item does not change its output $f \triangleleft \im_{A \tensor B} = f \triangleleft r_{f}$,
    \item is \emph{deterministic}, meaning $r_{f} ; \coMult_{A \tensor B} = \coMult_{A \tensor B} ; (r_{f} \tensor r_{f})$,
    \item and has the \emph{range} property, $f \triangleleft g = f \triangleleft h$ must imply
    \[(r_{f} \tensor \im) ; g = (r_{f} \tensor \im) ; h\]
    for any suitably typed $g$ and $h$.
  \end{enumerate}
  We say that a Markov category \emph{has ranges} if there exists a range for each morphism of the category.
\end{definition}

\begin{remark}
There already exists a notion of categorical \emph{range} in the literature, due to Cockett, Guo and Hofstra \cite{cockett2012range}.
It arises in parallel to the notion of \emph{support} in \emph{restriction categories} \cite{cockett02}.
The definition better suited for our purposes is different, even if it seems inspired by the same idea.
The main difference is that we are using a \emph{controlled range}; that is, the range of a morphism depends on the input to the original morphism.
We keep the name hoping that it will not cause any confusion, as we do not deal explicitly with restriction categories in this text.
\end{remark}

\begin{proposition}
  The \MarkovCategory{} $\Stoch$ has ranges.
\end{proposition}
\begin{proof}
  Given $f \colon A \to B$, we know that for each $a \in A$ there exists some $b_{a} \in B$ such that $f(b_{a}| a) > 0$.
  We fix such $b_{a} \in B$, and we define $r_{f} \colon A \tensor B \to A \tensor B$ as
  \[r_{f}(a,b) = \left\{
	\begin{array}{ll}
		(a,b)  & \mbox{if } f(b|a) > 0, \\
		(a,b_{a}) & \mbox{if } f(b|a) = 0.
	\end{array}
  \right.\]
  It is straightforward to check that it satisfies all the properties of ranges.
\end{proof}

\begin{theorem}\label{appendix:productivemarkov}
  Any \MarkovCategory{} with conditionals and ranges is \productive{}.
\end{theorem}
\begin{proof}
  Given any $\bra{\alpha} \in \Stage{1}(\stream{X},\stream{Y})$, we can define
  \[\alpha_{0} = \coMult_{A} ; \tid{\alpha ; (\coMult_{Y} \tensor \coUnit_{M})}.\]
  This is indeed well-defined because of naturality of the discarding map $(\coUnit)_{M} \colon M \to I$ in any \MarkovCategory{}.
  Let $c_{\alpha} \colon Y \tensor X \to M$ be a conditional of $\alpha$. This representative can then be factored as $\alpha = \alpha_{0} ; \tid{c_{\alpha}}$ (\Cref{figure:productivemarkov}).
  \begin{figure}
\tikzset{every picture/.style={line width=0.85pt}} %

\begin{tikzpicture}[x=0.75pt,y=0.75pt,yscale=-1,xscale=1]
\draw   (390,50) -- (430,50) -- (430,70) -- (390,70) -- cycle ;
\draw    (430,30) -- (430,38.87) ;
\draw    (410,50) .. controls (409.8,35.8) and (450.2,35.8) .. (450,50) ;
\draw  [fill={rgb, 255:red, 0; green, 0; blue, 0 }  ,fill opacity=1 ] (427.1,38.87) .. controls (427.1,37.27) and (428.4,35.97) .. (430,35.97) .. controls (431.6,35.97) and (432.9,37.27) .. (432.9,38.87) .. controls (432.9,40.47) and (431.6,41.77) .. (430,41.77) .. controls (428.4,41.77) and (427.1,40.47) .. (427.1,38.87) -- cycle ;
\draw    (420,70) -- (420,78.87) ;
\draw  [fill={rgb, 255:red, 0; green, 0; blue, 0 }  ,fill opacity=1 ] (417.1,78.87) .. controls (417.1,77.27) and (418.4,75.97) .. (420,75.97) .. controls (421.6,75.97) and (422.9,77.27) .. (422.9,78.87) .. controls (422.9,80.47) and (421.6,81.77) .. (420,81.77) .. controls (418.4,81.77) and (417.1,80.47) .. (417.1,78.87) -- cycle ;
\draw    (400,70) -- (400,78.87) ;
\draw    (390,90) .. controls (389.8,75.8) and (410.2,75.8) .. (410,90) ;
\draw  [fill={rgb, 255:red, 0; green, 0; blue, 0 }  ,fill opacity=1 ] (397.1,78.87) .. controls (397.1,77.27) and (398.4,75.97) .. (400,75.97) .. controls (401.6,75.97) and (402.9,77.27) .. (402.9,78.87) .. controls (402.9,80.47) and (401.6,81.77) .. (400,81.77) .. controls (398.4,81.77) and (397.1,80.47) .. (397.1,78.87) -- cycle ;
\draw    (390,90) -- (390,130) ;
\draw   (410,100) -- (450,100) -- (450,120) -- (410,120) -- cycle ;
\draw    (410,90) .. controls (409.8,97.4) and (420.2,93) .. (420,100) ;
\draw    (450,50) .. controls (450.14,62.14) and (440,87.83) .. (440,100) ;
\draw    (430,120) -- (430,130) ;
\draw   (300,70) -- (340,70) -- (340,90) -- (300,90) -- cycle ;
\draw    (310,90) -- (310,130) ;
\draw    (330,90) -- (330,130) ;
\draw    (320,30) -- (320,70) ;
\draw   (490,60) -- (550,60) -- (550,80) -- (490,80) -- cycle ;
\draw    (520,30) -- (520,60) ;
\draw    (520,80) -- (520,100) ;
\draw    (500,80) -- (500,97.5) -- (500,130) ;
\draw   (510,100) -- (550,100) -- (550,120) -- (510,120) -- cycle ;
\draw    (530,120) -- (530,130) ;
\draw    (540,80) -- (540,100) ;
\draw (430,26.6) node [anchor=south] [inner sep=0.75pt]    {$X$};
\draw (430,133.4) node [anchor=north] [inner sep=0.75pt]    {$M$};
\draw (410,60) node  [font=\normalsize]  {$\alpha $};
\draw (430,110) node  [font=\normalsize]  {$c_{\alpha }$};
\draw (390,133.4) node [anchor=north] [inner sep=0.75pt]    {$Y$};
\draw (320,80) node  [font=\normalsize]  {$\alpha $};
\draw (320,26.6) node [anchor=south] [inner sep=0.75pt]    {$X$};
\draw (310,133.4) node [anchor=north] [inner sep=0.75pt]    {$Y$};
\draw (330.5,133.4) node [anchor=north] [inner sep=0.75pt]    {$M$};
\draw (367.5,86.6) node [anchor=south] [inner sep=0.75pt]    {$=$};
\draw (520,70) node  [font=\normalsize]  {$\alpha _{0}$};
\draw (530,110) node  [font=\normalsize]  {$c_{\alpha }$};
\draw (472.5,87.6) node [anchor=south] [inner sep=0.75pt]    {$=$};
\draw (520,26.6) node [anchor=south] [inner sep=0.75pt]    {$X$};
\draw (500,133.4) node [anchor=north] [inner sep=0.75pt]    {$Y$};
\draw (530,133.4) node [anchor=north] [inner sep=0.75pt]    {$M$};
\end{tikzpicture}
\caption{Productivity for Markov categories.}
\label{figure:productivemarkov}
\end{figure}
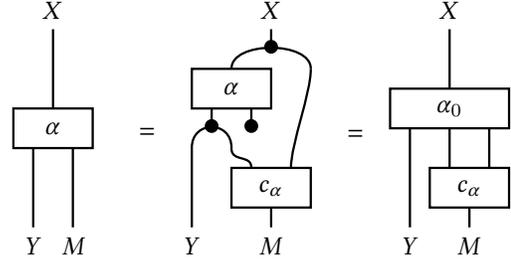

Now assume that for two representatives $\bra{\alpha_{i}} = \bra{\alpha_{j}}$ we have that $\bra{\alpha_{i}; u} = \bra{\alpha_{j} ; v}$.
By naturality of the discarding, $\alpha_{i} ; \tid{\varepsilon} = \alpha_{j} ; \tid{\varepsilon}$, and let $r$ be a range of this map.
Again by naturality of discarding, we have $\tid{\alpha_{i}}; \tid{u} ; \tid{\coUnit} = \tid{\alpha_{j}}; \tid{v} ; \tid{\coUnit}$.
Let then $c_{i}$ and $c_{j}$ be conditionals of $\alpha_{i}$ and $\alpha_{j}$:
we have that $(\alpha_{0}\triangleleft c_{i}); u ; \tid{\coUnit} = (\alpha_{0} \triangleleft c_{j}); v ; \tid{\coUnit}$.
By the properties of ranges (\Cref{figure:propertiesrange}), $(\alpha_{0}\triangleleft r ; c_{i}) ; u ; \coUnit_{M(u)} = (\alpha_{0}\triangleleft r ; c_{j}) ; v ; \coUnit_{M(v)}$,
and thus, $r ; c_{i}; u ; \coUnit_{M(u)} = r ; c_{j}; v ; \coUnit_{M(v)}$.
We pick $s_{i} = r ; c_{i}$ and we have proven that $\bra{\tid{s_{i}} ; u} = \bra{\tid{s_{j}} ; v}$.
\end{proof}

\begin{figure}
\tikzset{every picture/.style={line width=0.85pt}} %
\begin{tikzpicture}[x=0.75pt,y=0.75pt,yscale=-1,xscale=1]
\draw   (160.01,50) -- (200.01,50) -- (200.01,70) -- (160.01,70) -- cycle ;
\draw    (200.01,30) -- (200.01,38.87) ;
\draw    (180.01,50) .. controls (179.81,35.8) and (220.21,35.8) .. (220.01,50) ;
\draw  [fill={rgb, 255:red, 0; green, 0; blue, 0 }  ,fill opacity=1 ] (197.11,38.87) .. controls (197.11,37.27) and (198.41,35.97) .. (200.01,35.97) .. controls (201.61,35.97) and (202.91,37.27) .. (202.91,38.87) .. controls (202.91,40.47) and (201.61,41.77) .. (200.01,41.77) .. controls (198.41,41.77) and (197.11,40.47) .. (197.11,38.87) -- cycle ;
\draw    (190.01,70) -- (190.01,78.87) ;
\draw  [fill={rgb, 255:red, 0; green, 0; blue, 0 }  ,fill opacity=1 ] (187.11,78.87) .. controls (187.11,77.27) and (188.41,75.97) .. (190.01,75.97) .. controls (191.61,75.97) and (192.91,77.27) .. (192.91,78.87) .. controls (192.91,80.47) and (191.61,81.77) .. (190.01,81.77) .. controls (188.41,81.77) and (187.11,80.47) .. (187.11,78.87) -- cycle ;
\draw    (170.01,70) -- (170.01,78.87) ;
\draw    (160.01,90) .. controls (159.81,75.8) and (180.21,75.8) .. (180.01,90) ;
\draw  [fill={rgb, 255:red, 0; green, 0; blue, 0 }  ,fill opacity=1 ] (167.11,78.87) .. controls (167.11,77.27) and (168.41,75.97) .. (170.01,75.97) .. controls (171.61,75.97) and (172.91,77.27) .. (172.91,78.87) .. controls (172.91,80.47) and (171.61,81.77) .. (170.01,81.77) .. controls (168.41,81.77) and (167.11,80.47) .. (167.11,78.87) -- cycle ;
\draw    (160.01,90) -- (160.01,100) ;
\draw   (180.01,130) -- (220.01,130) -- (220.01,150) -- (180.01,150) -- cycle ;
\draw    (180.01,90) .. controls (179.81,97.4) and (190.21,93) .. (190.01,100) ;
\draw    (220.01,50) .. controls (220.15,62.14) and (210.01,87.83) .. (210.01,100) ;
\draw    (200.01,150) -- (200.01,160) ;
\draw   (180.01,100) -- (220.01,100) -- (220.01,120) -- (180.01,120) -- cycle ;
\draw    (190.01,120) -- (190.01,130) ;
\draw    (210.01,120) -- (210.01,130) ;
\draw   (160.01,160) -- (210.01,160) -- (210.01,180) -- (160.01,180) -- cycle ;
\draw    (200,180) -- (200,200) ;
\draw   (270.01,50) -- (310.01,50) -- (310.01,70) -- (270.01,70) -- cycle ;
\draw    (310.01,30) -- (310.01,38.87) ;
\draw    (290.01,50) .. controls (289.81,35.8) and (330.21,35.8) .. (330.01,50) ;
\draw  [fill={rgb, 255:red, 0; green, 0; blue, 0 }  ,fill opacity=1 ] (307.11,38.87) .. controls (307.11,37.27) and (308.41,35.97) .. (310.01,35.97) .. controls (311.61,35.97) and (312.91,37.27) .. (312.91,38.87) .. controls (312.91,40.47) and (311.61,41.77) .. (310.01,41.77) .. controls (308.41,41.77) and (307.11,40.47) .. (307.11,38.87) -- cycle ;
\draw    (300.01,70) -- (300.01,78.87) ;
\draw  [fill={rgb, 255:red, 0; green, 0; blue, 0 }  ,fill opacity=1 ] (297.11,78.87) .. controls (297.11,77.27) and (298.41,75.97) .. (300.01,75.97) .. controls (301.61,75.97) and (302.91,77.27) .. (302.91,78.87) .. controls (302.91,80.47) and (301.61,81.77) .. (300.01,81.77) .. controls (298.41,81.77) and (297.11,80.47) .. (297.11,78.87) -- cycle ;
\draw    (280.01,70) -- (280.01,78.87) ;
\draw    (270.01,90) .. controls (269.81,75.8) and (290.21,75.8) .. (290.01,90) ;
\draw  [fill={rgb, 255:red, 0; green, 0; blue, 0 }  ,fill opacity=1 ] (277.11,78.87) .. controls (277.11,77.27) and (278.41,75.97) .. (280.01,75.97) .. controls (281.61,75.97) and (282.91,77.27) .. (282.91,78.87) .. controls (282.91,80.47) and (281.61,81.77) .. (280.01,81.77) .. controls (278.41,81.77) and (277.11,80.47) .. (277.11,78.87) -- cycle ;
\draw    (270.01,90) -- (270.01,100) ;
\draw   (290.01,130) -- (330.01,130) -- (330.01,150) -- (290.01,150) -- cycle ;
\draw    (290.01,90) .. controls (289.81,97.4) and (300.21,93) .. (300.01,100) ;
\draw    (330.01,50) .. controls (330.15,62.14) and (320.01,87.83) .. (320.01,100) ;
\draw    (310.01,150) -- (310.01,160) ;
\draw   (290.01,100) -- (330.01,100) -- (330.01,120) -- (290.01,120) -- cycle ;
\draw    (300.01,120) -- (300.01,130) ;
\draw    (320.01,120) -- (320.01,130) ;
\draw   (270.01,160) -- (320.01,160) -- (320.01,180) -- (270.01,180) -- cycle ;
\draw    (310,180) -- (310,200) ;
\draw    (270.01,100) .. controls (270.18,119.83) and (260.18,139.5) .. (260.01,160) ;
\draw    (260.01,80) .. controls (259.51,119.17) and (279.51,130.17) .. (280.01,160) ;
\draw    (260.01,30) -- (260.01,80) ;
\draw    (260.01,160) -- (260,200) ;
\draw    (160.02,100) .. controls (160.19,119.83) and (150.19,139.5) .. (150.02,160) ;
\draw    (150.01,80) .. controls (149.51,119.17) and (169.51,130.17) .. (170.01,160) ;
\draw    (150.01,30) -- (150.01,80) ;
\draw    (150.02,160) -- (150,200) ;
\draw    (280,180) -- (280,190) ;
\draw    (170,180) -- (170,190) ;
\draw  [fill={rgb, 255:red, 0; green, 0; blue, 0 }  ,fill opacity=1 ] (167.1,190) .. controls (167.1,188.4) and (168.4,187.1) .. (170,187.1) .. controls (171.6,187.1) and (172.9,188.4) .. (172.9,190) .. controls (172.9,191.6) and (171.6,192.9) .. (170,192.9) .. controls (168.4,192.9) and (167.1,191.6) .. (167.1,190) -- cycle ;
\draw  [fill={rgb, 255:red, 0; green, 0; blue, 0 }  ,fill opacity=1 ] (277.1,190) .. controls (277.1,188.4) and (278.4,187.1) .. (280,187.1) .. controls (281.6,187.1) and (282.9,188.4) .. (282.9,190) .. controls (282.9,191.6) and (281.6,192.9) .. (280,192.9) .. controls (278.4,192.9) and (277.1,191.6) .. (277.1,190) -- cycle ;
\draw (200.01,26.6) node [anchor=south] [inner sep=0.75pt]    {$X$};
\draw (180.01,60) node  [font=\normalsize]  {$\alpha _{i}$};
\draw (200.01,140) node  [font=\normalsize]  {$c_{i}$};
\draw (200.01,110.5) node  [font=\normalsize]  {$r$};
\draw (185.01,170) node  [font=\normalsize]  {$u$};
\draw (310.01,26.6) node [anchor=south] [inner sep=0.75pt]    {$X$};
\draw (290.01,60) node  [font=\normalsize]  {$\alpha _{j}$};
\draw (310.01,140) node  [font=\normalsize]  {$c_{j}$};
\draw (310.01,110.5) node  [font=\normalsize]  {$r$};
\draw (295.01,170) node  [font=\normalsize]  {$v$};
\draw (237.51,107.6) node [anchor=south] [inner sep=0.75pt]    {$=$};
\draw (150.01,26.6) node [anchor=south] [inner sep=0.75pt]    {$Z$};
\draw (150,203.4) node [anchor=north] [inner sep=0.75pt]    {$Y$};
\draw (200,203.4) node [anchor=north] [inner sep=0.75pt]    {$W$};
\draw (260.01,26.6) node [anchor=south] [inner sep=0.75pt]    {$Z$};
\draw (260,203.4) node [anchor=north] [inner sep=0.75pt]    {$Y$};
\draw (310,203.4) node [anchor=north] [inner sep=0.75pt]    {$W$};
\end{tikzpicture}
\caption{Applying the properties of range.}
\label{figure:propertiesrange}
\end{figure}
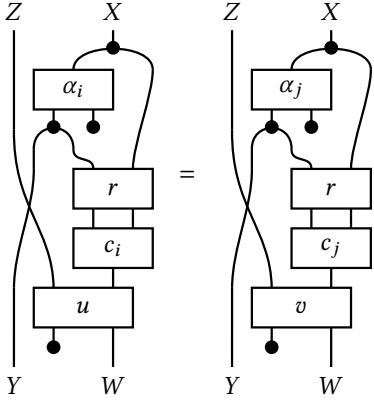

\subsection{Stochastic processes}

\begin{definition}[Controlled stochastic process]\label{appendix:def:stoch-rep-process}
  Let \(\stream{X} = \streamExpr{X}\) and \(\stream{Y} = \streamExpr{Y}\) be infinite sequences of sets.
  A controlled \emph{stochastic process} \(\stream{f} \colon \stream{X} \to \stream{Y}\) is an infinite sequence \(\stream{f} = \streamExpr{f}\) of functions \(f_{n} \colon X_{n} \times \dots \times X_{1} \to \distr(Y_{n} \times \dots \times Y_{1})\) such that \(f_{n}\) coincides with the marginal distribution of \(f_{n+1}\) on the first \(n\) variables.
  In other words, \(f_{n+1} \dcomp D\proj[Y_{0},\dots,Y_{n}] = \proj[X_{0},\dots,X_{n}] \dcomp f_{n}\).
  \[
  \begin{tikzcd}
    X_0 \times \dots \times X_{n+1} \rar{f_{n+1}} \dar[swap]{\pi_{0,\dots,n}} & D(Y_0 \times \dots \times Y_{n+1}) \dar{D\pi_{0,\dots,n}} \\
    X_0 \times \dots \times X_{n}  \rar{f_{n}} & D(Y_0 \times \dots \times Y_{n})
  \end{tikzcd}\]
  Let \(\Marg\) be the category with objects infinite sequences of sets
  \(\stream{X} = \streamExpr{X}\) and morphisms controlled stochastic
  processes \(\stream{f} = \streamExpr{f}\)
  with composition and identities defined componentwise in \(\Stoch\).
\end{definition}

\begin{proposition}[Factoring as conditionals]
  A stochastic process $f \colon \stream{X} \to \stream{Y}$ can be always written as
  \[f_{n} = c_{0} \triangleleft c_{1} \triangleleft \dots \triangleleft c_{n},\]
  for some family of functions
  \[c_{n} \colon Y_{0} \times \dots \times Y_{n-1} \times X_{0} \times \dots \times X_{n} \to Y_{n},\]
  called the \emph{conditionals} of the stochastic process.
\end{proposition}
\begin{proof}
  We proceed by induction, noting first that $c_{0} = f_{0}$.
  In the general case, we apply conditionals to rewrite $f_{n+1} = (f_{n+1} ; (\coUnit)_{Y_{n+1}}) \triangleleft c_{n+1}$.
  Because of the marginalization property, we know that $f_{n+1} ; (\coUnit)_{Y_{n+1}} = f_{n}$.
  So finally, $f_{n+1} = f_{n} \triangleleft c_{n+1}$, which by the induction hypothesis gives the desired result.
\end{proof}

\begin{proposition}\label{prop:condtoproc}
  If two families of conditionals give rise to the same stochastic process,
  \[c_{0} \triangleleft c_{1} \triangleleft \dots \triangleleft c_{n} =
  c_{0}' \triangleleft c_{1}' \triangleleft \dots \triangleleft c_{n}',\]
then, they also give rise to the same n-stage processes in $\Stoch$,
\[\bra{c_{0} \triangleleft \im |c_{1} \triangleleft \im |\dots|c_{n} \triangleleft \im} = \bra{c_{0}' \triangleleft \im |c_{1}' \triangleleft \im|\dots|c_{n}' \triangleleft \im}.\]
\end{proposition}
\begin{proof}
  We start by defining a family of morphisms $r_{n}$ by induction.
  We take $r_{0} = \im$ and $r_{n+1}$ to be a \emph{range} of $r_{n}; \coMult; c_{n}$.

  Let us prove now that for any $n \in \naturals$ and $i \leq n$,
  \[r_{i};\coMult; c_{i} \triangleleft \dots \triangleleft c_{n} = r_{i}; \coMult ; c_{i}' \triangleleft \dots \triangleleft c_{n}'.\]
  We proceed by induction. Observing that $c_{0} = c_{0}'$, we prove it for $n = 0$ and also for the case $i = 0$ for any $n \in \naturals$.
Assume we have it proven for $n$, so in particular we know that $r_{i}; \coMult ; c_{i} = r_{i}; \coMult ; c_{i}'$ for any $i \leq n$. Now, by induction on $i$, we can use the properties of ranges to show that
  \[\begin{gathered}
      r_{i}; \coMult ; c_{i} \triangleleft \dots \triangleleft c_{n} = r_{i}; \coMult ; c_{i}' \triangleleft \dots \triangleleft c_{n}' \\
      (r_{i}; \coMult ; c_{i} \triangleleft \im) ; c_{i+1} \triangleleft  \dots \triangleleft c_{n} =
      (r_{i}; \coMult ; c_{i}' \triangleleft \im) ; c_{i+1}' \triangleleft \dots \triangleleft c_{n}' \\
      (r_{i}; \coMult ; c_{i} \triangleleft r_{i+1}) ; c_{i+1} \triangleleft  \dots \triangleleft c_{n} =
      (r_{i}; \coMult ; c_{i}' \triangleleft r_{i+1}) ; c_{i+1}' \triangleleft \dots \triangleleft c_{n}' \\
      r_{i+1} ; \coMult ;  c_{i+1} \triangleleft  \dots \triangleleft c_{n} =
      r_{i+1} ; \coMult ;  c_{i+1}' \triangleleft \dots \triangleleft c_{n}'. \\
    \end{gathered}\]
  In particular, $r_{n}; \coMult ; c_{n} = r_{n} ; \coMult ; c_{n}'$.

  Now, we claim the following for each $n \in \naturals$ and each $i \leq n$,
  \[\begin{gathered}\bra{c_{0} \triangleleft \im | \dots | c_{n} \triangleleft \im} = \\
  \bra{r_{0} (c_{0} \triangleleft \im) | \dots | r_{i} (c_{i} \triangleleft \im) | c_{i+1} | \dots | c_{n} \triangleleft \im}.\end{gathered}\]
  It is clear for $n = 0$ and for $i = 0$. In the inductive case for $i$,
  \[\begin{aligned}
      \bra{r_{0} (c_{0} \triangleleft \im) | \dots | r_{i} (c_{i} \triangleleft \im) | c_{i+1}\triangleleft \im | \dots | c_{n} \triangleleft \im} = \\
      \bra{r_{0} (c_{0} \triangleleft \im) | \dots | r_{i} ; \coMult ; c_{i} \triangleleft \im | c_{i+1}\triangleleft \im | \dots | c_{n} \triangleleft \im} = \\
      \bra{r_{0} (c_{0} \triangleleft \im) | \dots | r_{i} ; \coMult ; c_{i} \triangleleft r_{i+1} | c_{i+1} \triangleleft \im | \dots | c_{n} \triangleleft \im} = \\
      \bra{r_{0} (c_{0} \triangleleft \im) | \dots | r_{i} ; \coMult ; c_{i} \triangleleft \im | r_{i+1} (c_{i+1}\triangleleft \im) | \dots | c_{n} \triangleleft \im} = \\
      \bra{r_{0} (c_{0} \triangleleft \im) | \dots | r_{i} (c_{i} \triangleleft \im) | r_{i+1} (c_{i+1}\triangleleft \im) | \dots | c_{n} \triangleleft \im}.
   \end{aligned}\]
 A particular case of this claim is then that
 \[\begin{aligned}
     \bra{c_{0} \triangleleft \im | \dots |c_{n} \triangleleft \im} = \\
     \bra{r_{0} (c_{0} \triangleleft \im) | \dots| r_{n} (c_{n} \triangleleft \im)} = \\
     \bra{r_{0} (c_{0}' \triangleleft \im) | \dots| r_{n} (c_{n}' \triangleleft \im)} = \\
     \bra{c_{0}' \triangleleft \im | \dots |c_{n}' \triangleleft \im}.
   \end{aligned}\]
 This can be then proven for any $n \in \naturals$.
\end{proof}

\begin{corollary}
  Any stochastic process $f \in \StochProc(\stream{X},\stream{Y})$ with
  a family of conditionals $c_{n}$ gives rise to the observational sequence
  \[\begin{gathered}
      \obs(f) = [\left\langle (c_{n} \triangleleft \im) \colon (X_{0} \times Y_{0} \times \dots \times X_{n-1} \times Y_{n-1}) \times X_{n} \to \right. \\ \left. (X_{0} \times Y_{0} \times \dots \times X_{n} \times Y_{n}) \times Y_{n}  \right\rangle ]_{\approx},
    \end{gathered}\]
  which is independent of the chosen family of conditionals.
\end{corollary}
\begin{proof}
  Any two families of conditionals for $f$ give rise to the same n-stage processes in $\Stoch{}$ (by \Cref{prop:condtoproc}).
  Being a \productive{} category, observational sequences are determined by their n-stage procesess.
\end{proof}

\begin{proposition}
  An observational sequence in $\Stoch$,
  \[[\braket{g_{n} \colon M_{n-1} \tensor X_{n} \to M_{n} \tensor Y_{n}}]_{\approx} \in \oSeq(\stream{X},\stream{Y})\]
  gives rise to a stochastic process $\mathrm{proc}(g) \in \StochProc(\stream{X},\stream{Y})$
  defined by
  $\mathrm{proc}(g)_{n} = \tid{g_{0}} ; \tid{g_{1}} ; \dots ; \tid{g_{n}} ; \tid{\varepsilon_{M_{n}}}$.
\end{proposition}
\begin{proof}
  The symmetric monoidal category $\Stoch$ is \productive{}: by \Cref{lemma:observationalstatefulisterminalsequence}, observational sequences are determined by their n-stage truncations
  \[\bra{g_{0}|\dots|g_{n}} \in \Stage{n}(\stream{X},\stream{Y}).\]
  Each n-stage truncation gives rise to the n-th component of the stochastic process,
  $\mathrm{proc}(g)_{n} = \tid{g_{0}} ; \tid{g_{1}} ; \dots ; \tid{g_{n}} ; \tid{\varepsilon_{M_{n}}}$,
  and this is well-defined: composing the morphisms is invariant to \emph{sliding equivalence}, and the last discarding map is natural.

  It only remains to show that they satisfy the marginalisation property. Indeed,
  \[\begin{aligned}
      \mathrm{proc}(g)_{n+1} ; \tid{\varepsilon_{n+1}}
      &= \tid{g_{0}} ; \tid{g_{1}} ; \dots ; \tid{g_{n+1}} ; \tid{\varepsilon_{M_{n+1}}} ; \tid{\varepsilon_{Y_{n+1}}} \\
      &= \tid{g_{0}} ; \tid{g_{1}} ; \dots ; \tid{g_{n}} ; \tid{\varepsilon_{M_{n}}} \\
      &= \mathrm{proc}(g)_n.
    \end{aligned}\]
  Thus, $\mathrm{proc}(g)$ is a stochastic process in $\StochProc(\stream{X},\stream{Y})$.
\end{proof}

\begin{proposition}\label{prop:procobs}
  Let $f \in \StochProc(\stream{X},\stream{Y})$, we have that $\proc(\obs(f)) = f$.
\end{proposition}
\begin{proof}
  Indeed, for $c_{n}$ some family of conditionals,
  \[f_{n} = c_{0} \triangleleft \dots \triangleleft c_{n}
    = (c_{0} \triangleleft \im) ; \dots ; (c_{n} \triangleleft \im) ; \varepsilon_{M_{n}}.\qedhere\]
\end{proof}

\begin{theorem}\label{theorem:obstoch}
  Observational sequences in $\Stoch$ are in bijection with stochastic processes.
\end{theorem}
\begin{proof}
  The function $\obs$ is injective by \Cref{prop:procobs}.
  We only need to show it is also surjective.

  We will prove that any n-stage process $\bra{g_{0}|\dots|g_{n}}$
  can be equivalently written in the form $\bra{(c_{0} \triangleleft \im)|\dots|(c_{n} \triangleleft \im)}$.
  We proceed by induction.
  Given any $\bra{g_{0}}$ we use conditionals and dinaturality to rewrite it as
  \[\bra{g_{0}} = \bra{c_{0} \triangleleft c_{M}} = \bra{c_{0} \triangleleft \im}.\]
  Given any $\bra{c_{0} \triangleleft \im|\dots|c_{n} \triangleleft \im |g_{n+1}}$,
  we use again conditionals and dinaturality to rewrite it as
  \[\begin{aligned}
    \bra{c_{0} \triangleleft \im|\dots|c_{n} \triangleleft \im|g_{n+1}} = \\
    \bra{c_{0} \triangleleft \im|\dots|c_{n} \triangleleft \im|c_{n+1}\triangleleft c_{M}} = \\
    \bra{c_{0} \triangleleft \im|\dots|c_{n} \triangleleft \im|c_{n+1}\triangleleft \im}.
  \end{aligned}\]
   We have shown that $\obs$ is both injective and surjective.
\end{proof}

\begin{theorem}[From \Cref{th:stochasticprocesses}]\label{theorem:obstochiso}
  The category $\Stoch$ of stochastic processes is monoidally isomorphic to the category $\STREAM$ over $\Stoch$.
\end{theorem}
\begin{proof}
  We have shown in \Cref{theorem:obstoch} that $\proc$ is a bijection.
  Let us show that it preserves compositions. Indeed,
  \[\begin{aligned}
      \proc(g ; h)_{n}
      & = \tid{g_0} ; \tid{h_{0}} ; \dots ; \tid{g_{n}}; \tid{h_{n}}; \tid{\varepsilon_{M_{n} \tensor N_{n}}} \\
      & = \tid{g_0} ; \dots ; \tid{g_{n}}; \tid{h_{0}} ; \dots ; \tid{h_{n}}; \tid{(\varepsilon_{M_{n}} \tensor \varepsilon_{N_{n}})} \\
      & = \tid{g_0} \dots \tid{g_{n}} ;\tid{\varepsilon_{M_{n}}} ; \tid{h_{0}} \dots \tid{h_{n}} ; \tid{\varepsilon_{N_{n}}} \\
      & = \proc(g)_{n} ; \proc(h)_{n}.
  \end{aligned}\]
It also trivially preserves the identity. It induces thus an identity-on-objects functor which is moreover
an equivalence of categories. Let us finally show that it preserves tensoring of morphisms.
  \[\begin{aligned}
      \proc(g \tensor h)_{n}
      & = \tid{(g_0 \tensor h_{0})} ; \dots  ;\tid{(g_{n} \tensor h_{n})} ; \tid{\varepsilon_{M_{n} \tensor N_{n}}} \\
      & = \tid{(g_0 \tensor h_{0})} ; \dots \tid{((g_{n}  \varepsilon_{M_{n}}) \tensor (h_{n} \varepsilon_{N_{n}}))}  \\
      & = \tid{(g_0 \dots g_{n}  \varepsilon_{M_{n}})} \tensor \tid{(h_{0} \dots h_{n}  \varepsilon_{N_{n}})} \\
      & = \proc(g)_{n} \tensor \proc(h)_{n}.
  \end{aligned}\]
It is thus also a monoidal equivalence.
\end{proof}

\section{Coend Calculus and Profunctors}

\emph{Coend calculus} is the name given to the a branch of category theory that
describes the behaviour of certain colimits called \emph{coends}.
MacLane \cite{maclane78} and Loregian \cite{loregian2021} give complete presentations of coend calculus.

\begin{definition}
  \defining{linkcoend}{\emph{Coends}} are the coequalizers of the action of morphisms on both arguments of a profunctor.
\[\mathrm{coend}(P) \coloneqq \mathrm{coeq}\left( \begin{tikzcd}[column sep=small]
\coprod_{f \colon B \to A} P(A,B) \rar[yshift=-0.5ex, swap] \rar[yshift=0.5ex] &
\coprod_{X \in \catC} P(X,X)
\end{tikzcd}\right).\]
Coends are usually denoted with a superscripted integral, drawing on an
analogy with classical calculus.
\[\coend{X \in \catC} P(X,X) \defn \mathrm{coend}(P).\]
\end{definition}

\begin{proposition}
  [\defining{linkcoyoneda}{Yoneda} reduction] Let $\catC$ be any category and
  let $F \colon \catC \to \Set$ be a functor; the following isomorphism holds
  for any given object $A \in \catC$.
\label{prop:yonedareduction}
\[\coend{X \in \catC} \idProf(X,A) \times FX \cong FA. \]
Following the analogy with classical analysis, the $\idProf$ profunctor works as
a Dirac's delta.
\end{proposition}

\begin{proposition}
  [\defining{linkfubini}{Fubini rule}] Coends commute between them; that is,
  there exists a natural isomorphism
  \label{prop:fubinirule}
  \[\begin{aligned}
    & \coend{X_{1} \in \catC} \coend{X_{2} \in \catC} P(X_{1},X_{2},X_{1},X_{2}) \\
    \cong & \\
    & \coend{X_{2} \in \catC} \coend{X_{1} \in \catC}  P(X_{1},X_{2},X_{1},X_{2}).
  \end{aligned}\]
  In fact, they are both isomorphic to the coend over the product category,
  \[ \coend{(X_{1}, X_{2}) \in \catC \times \catC} P(X_{1},X_{2},X_{1},X_{2}). \]
  Following the analogy with classical analysis, coends follow the Fubini rule
  for integrals.
\end{proposition}

A \defining{linkprofunctor}{profunctor} from a category $\catA$ to a category $\catB$ is a functor $P \colon \catA^{op} \times \catB \to \Set$.
They can be seen as a categorification of the concept of \emph{relations}, functions $A \times B \to 2$.
Under this analogy, existential quantifiers correspond to \emph{coends}.
The canonical example of a profunctor is, $\idProf \colon \catA^{op} \times \catA \to \Set$, the profunctor that returns the set of morphisms between two objects.
Many operations relating families of processes are more easily defined in terms of profunctors: for instance, sequential composition connects the outputs of a family of processes to the outputs of another family.

\begin{definition}[Sequential composition]
  Two profunctors $P \colon \catA^{op} \times \catB \to \Set$ and
  $Q \colon \catB^{op} \times \catC \to \Set$ compose \defining{linksequential}{sequentially} into
  a profunctor $P \diamond Q \colon \catA^{op} \times \catC \to \Set$ defined by
  \[(P \diamond Q)(A,C) \defn \coend{B \in \catB} P(A,B) \times Q(B,C).\] The
  \defining{linkidprof}{hom-profunctor}
  $\idProf \colon \catA^{op} \times \catA \to \Set$ that returns the
  set of morphisms between two objects is the unit for sequential composition.
  Sequential composition is associative up to isomorphism.
\end{definition}

\begin{definition}[Parallel composition]
  Two profunctors  $P \colon \catA_{1}^{op} \times \catB_{1} \to \Set$ and
  $Q \colon \catA_{2}^{op} \times \catB_{2} \to \Set$ compose \defining{linkparallel}{\emph{in parallel}}
  into a profunctor $P \times Q \colon \catA^{op}_{1} \times \catA^{op}_{2} \times \catB_{1} \times \catB_{2} \to \Set$
  defined by
  \[(P \times Q)(A,A',B,B') \defn P(A,B) \times Q(A',B').\]
\end{definition}

\begin{definition}[Intensional communicating composition]
  \defining{linkintensionalcommunicatingcomposition}{}
  Let $\catA, \catB, \catC$ be categories and let $\catB$ have a monoidal
  structure. Let
  $P \colon \catA^{op} \times \catB \to \Set$ and
  $Q \colon \catB^{op} \times \catC^{\naturals} \to \Set$ be a
  pair of profunctors. Their \emph{intensional communicating composition} is the
  profunctor
  $\intensionalBox{P}{Q} \colon \catA^{op} \times \catB^{op} \times \catB \times \catC \to \Set$
  defined as
  \[(\intensionalBox{P}{Q})(A,B;B',C) \defn \sum_{M \in \catB} P(A,B \otimes M) \times Q(M \otimes B',C).\]
\end{definition}

\begin{remark}
  Let $\catC$ be a monoidal category and let $P \colon \catC^{op} \times \catC \to \Set$ and
  $Q \colon \NcatC^{op} \times \NcatC \to \Set$ be a pair of profunctors.
  Note that $\NcatC \cong \catC \times \NcatC$, and so the second profunctor can be interpreted as having type $Q \colon \catC^{op} \times (\NcatC^{op} \times \NcatC) \to \Set$.
In this case, their intensional communicating composition is defined by
  \[(\intensionalBox{P}{Q})(\stream{X};\stream{Y}) \coloneqq \sum_{M \in \catC} P(X_{0}, M \tensor Y_{0}) \times Q(M \cdot \tail{\stream{X}}, \tail{\stream{Y}}).\]
  This is the composition we use when we describe the endofunctor
  $(\intensionalBox{\hom{}}{\ \bullet}) \colon [\NcatC^{op} \times \NcatC, \Set] \to [\NcatC^{op} \times \NcatC, \Set]$.
  \[(\intensionalBox{\hom{}}{\ Q})(\stream{X};\stream{Y}) \coloneqq \sum_{M \in \catC} \hom{}(X_{0}, M \tensor Y_{0}) \times Q(M \cdot \tail{\stream{X}}, \tail{\stream{Y}}).\]
\end{remark}

\begin{definition}[Communicating profunctor composition]
  Let $\catA, \catB, \catC$ be categories and let $\catB$ have a monoidal
  structure. Two profunctors $P \colon \catA^{op} \times \catB \to \Set$ and
  $Q \colon \catB^{op} \times \catC \to \Set$ \defining{linkocomm}{compose communicating along}
  $\catB$ into the profunctor
  $(P \andThen Q) \colon \catA^{op} \times \catB \times \catB^{op} \times \catC \to \Set$
  defined by
  \[(P \andThen Q)(A,B;B',C) \defn \coend{M} P(A,B \otimes M) \times Q(M \otimes B',C).\]
  The profunctors $\idProf(I,\bullet) \colon \catB \to \Set$ and
  $\idProf(\bullet,I) \colon \catB^{op} \to \Set$ are left and right units with
  respect to communicating composition. The communicating composition of three
  profunctors $P \colon \catA^{op} \times \catB \to \Set$,
  $Q \colon \catB^{op} \times \catC \to \Set$ and
  $R \colon \catC^{op} \times \catD \to \Set$ is associative up to isomorphism and a
  representative can be written simply by
  $(P \andThen Q \andThen R) \colon \catA^{op} \times \catB \times \catB^{op} \times \catC \times \catC^{op} \times \catD \to \Set$,
  where both $\catB$ and $\catC$ are assumed to have a monoidal structure.
\end{definition}

\begin{remark}
  This is the composition we use when we describe the endofunctor
  $(\hom{} \andThen\ \bullet) \colon [\NcatC^{op} \times \NcatC, \Set] \to [\NcatC^{op} \times \NcatC, \Set]$.
  \[(\hom{} \andThen\ Q)(\stream{X};\stream{Y}) \coloneqq \coend{M \in \catC} \hom{}(X_{0}, M \tensor Y_{0}) \times Q(M \cdot \tail{\stream{X}}, \tail{\stream{Y}}).\]
\end{remark}
\subsection{Initial algebras, final coalgebras}

\begin{definition}[Algebras and coalgebras]
  Let $\catC$ be a category and let $F \colon \catC \to \catC$ be an endofunctor. An
  \emph{algebra} $(X, \alpha)$ is an object $X \in \catC$, together with a
  morphism $\alpha \colon FX \to X$. A \emph{coalgebra} $(Y,\beta)$ is an
  object $Y \in \catC$, together with a morphism $\beta \colon Y \to FY$.

  An \emph{algebra morphism} $f \colon (X,\alpha) \to (X',\alpha')$ is a morphism
  $f \colon X \to X'$ such that the diagram on the left commutes. A \emph{coalgebra
  morphism} $g \colon (Y,\beta) \to (Y',\beta')$ is a morphism
  $f \colon Y \to Y'$ such that the diagram on the right commutes.
  \[\begin{tikzcd}
      FX \dar[swap]{\alpha} \rar{Ff} & FX' \dar{\alpha'} & Y \rar{g}\dar[swap]{\beta} & Y' \dar{\beta'} \\
    X \rar{f} & X' & FY \rar{Fg} & FY'
  \end{tikzcd}\]
  Algebras for an endofunctor form a category with algebra morphisms between them.
  The initial algebra is the initial object in this category.
  Coalgebras for an endofunctor form a category with coalgebra morphisms between them.
  The final coalgebra is the terminal object in this category.
\end{definition}

\begin{definition}[Fixpoints of an endofunctor]
  Let $\catC$ be a category and let $F \colon \catC \to \catC$ be an endofunctor.
  A \emph{fixpoint} is an algebra $(X,\alpha)$ such that $\alpha \colon FX \to X$
  is an isomorphism. Equivalently, a fixpoint is a coalgebra $(Y,\beta)$ such that
  $\beta \colon Y \to FY$ is an isomorphism.

  Fixpoints form a category with algebra morphisms (or, equivalently, coalgebra morphisms) between
  them.
\end{definition}

\begin{theorem}[Lambek, \cite{lambek68}]
  \label{ax:th:lambektheorem}
  The final coalgebra of a functor is a fixpoint.
  As a consequence, when it exists, it is the final fixpoint.
\end{theorem}

\begin{theorem}[Adamek, \cite{adamek74}]
  \label{ax:th:adamek}
  Let $\catD$ be a category with a final object $1$ and $\omega$-shaped
  limits. Let $F \colon \catD \to \catD$ be an endofunctor. We write $L \defn \lim\nolimits_{n} F^{n}1$ for
  the limit of the following $\omega$-chain, which is called the \emph{terminal sequence}.
  \[1 \overset{!}\longleftarrow F1 \overset{F!}\longleftarrow FF1 \overset{FF!}\longleftarrow FFF1 \overset{FFF!}\longleftarrow \dots \]
  Assume that $F$ preserves this limit, meaning that the canonical morphism $FL \to L$ is an isomorphism.
  Then, $L$ is the final $F$-coalgebra.
\end{theorem}
\subsection{Size concerns, limits and colimits}

\begin{remark}
We call $\Set$ to the category of sets and functions below a certain Grothendieck universe.
We do take colimits (and coends) over this category without creating size issues: we can be sure of their existence in our metatheoretic category of sets.
\end{remark}

\begin{proposition}
  Terminal coalgebras exist in $\Set$.
  More generally, the category of sets below a certain regular uncountable cardinal is algebraically complete and cocomplete; meaning that every $\Set$-endofunctor has a terminal coalgebra and an initial algebra.
  See \cite[Theorem 13]{adamek19}.
\end{proposition}

\begin{theorem}[Coproducts commute with connected limits]
  \defining{linkconnectedlimits}
  Let $I$ be a set, understood as a discrete category, and let $\catA$ be a
  connected category with $F \colon I \times \catA \to \Set$ a functor. The
  canonical morphism
  \[ \sum_{i \in I} \lim_{a \in A} F(i,a) \to \lim_{a \in A}\sum_{i \in I} F(i,a)\]
  is an isomorphism.

  In particular, let $F_{n} \colon I \to \Set$ be a family of functors indexed
  by the natural numbers with a family of natural transformations
  $\alpha_{n} \colon F_{n+1} \to F_{n}$. The canonical morphism
  \[\sum_{i \in I} \lim_{n \in \naturals} F_{n}(i) \to \lim_{n \in \naturals}\sum_{i \in I} F_{n}(i) \]
  is an isomorphism.
\end{theorem}
\begin{proof}
  Note that there are no morphisms between any two indices $i, j \in I$.
  Once some $i \in I$ is chosen in any factor of the connected limit, it
  forces any other factor to also choose $i \in I$. This makes the local
  choice of $i \in I$ be equivalent to the global choice of $i \in I$.
\end{proof}
\section{The $\mathsf{List}^{+}$ opmonoidal comonad and Fox's theorem}
\begin{proposition}\label{prop:cartesianproductive}
  Cartesian monoidal categories are \productive{}.
\end{proposition}
\begin{proof}
  Let $\bra{\alpha} \in \Stage{1}(\stream{X},\stream{Y})$.
For some given representative $\alpha \colon X_{0} \to M \tensor Y_{0}$, we define the two projections $\alpha_{Y} = \alpha ; \tid{\coUnit_{M}} \colon X_{0} \to Y_{0}$ and $\alpha_{M} = \alpha ; \tid{\coUnit_{Y}}$.
The second projection $\alpha_{M}$ depends on the specific representative $\alpha$ we have chosen; however, the first projection $\alpha_{Y}$ is defined independently of the specific representative $\alpha$, as a consequence of naturality of the discarding map (see Fox's theorem for cartesian monoidal categories \Cref{th:fox}).
We define $\alpha_{0} = \delta_{X_{0}} ; \tid{\alpha_{Y}}$.
Then, we can factor any representative as $\alpha = \alpha_{0} ; \tid{\alpha_{M}}$ (see \Cref{figure:productivity:cartesian}).
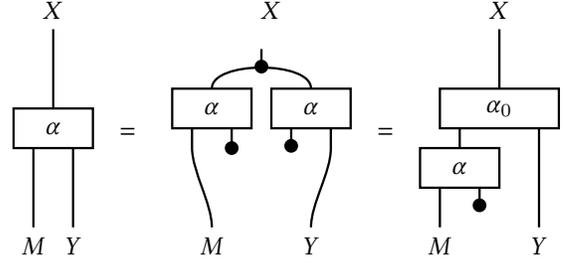
\begin{figure}
\tikzset{every picture/.style={line width=0.85pt}} %
\begin{tikzpicture}[x=0.75pt,y=0.75pt,yscale=-1,xscale=1]
\draw   (380,60) -- (420,60) -- (420,80) -- (380,80) -- cycle ;
\draw    (425,40) -- (425,48.87) ;
\draw    (400,60) .. controls (399.8,45.8) and (450.2,45.8) .. (450,60) ;
\draw  [fill={rgb, 255:red, 0; green, 0; blue, 0 }  ,fill opacity=1 ] (422.1,48.87) .. controls (422.1,47.27) and (423.4,45.97) .. (425,45.97) .. controls (426.6,45.97) and (427.9,47.27) .. (427.9,48.87) .. controls (427.9,50.47) and (426.6,51.77) .. (425,51.77) .. controls (423.4,51.77) and (422.1,50.47) .. (422.1,48.87) -- cycle ;
\draw    (410,80) -- (410,90) ;
\draw    (390,80) -- (390,90) ;
\draw    (460,90) .. controls (460.14,102.14) and (450,117.83) .. (450,130) ;
\draw   (300,70) -- (340,70) -- (340,90) -- (300,90) -- cycle ;
\draw    (310,90) -- (310,130) ;
\draw    (330,90) -- (330,130) ;
\draw    (320,30) -- (320,70) ;
\draw   (515,60) -- (575,60) -- (575,80) -- (515,80) -- cycle ;
\draw    (545,30) -- (545,60) ;
\draw    (565,80) -- (565,97.5) -- (565,130) ;
\draw    (525,80) -- (525,90) ;
\draw   (430,60) -- (470,60) -- (470,80) -- (430,80) -- cycle ;
\draw    (460,80) -- (460,90) ;
\draw    (440,80) -- (440,88.87) ;
\draw  [fill={rgb, 255:red, 0; green, 0; blue, 0 }  ,fill opacity=1 ] (437.1,88.87) .. controls (437.1,87.27) and (438.4,85.97) .. (440,85.97) .. controls (441.6,85.97) and (442.9,87.27) .. (442.9,88.87) .. controls (442.9,90.47) and (441.6,91.77) .. (440,91.77) .. controls (438.4,91.77) and (437.1,90.47) .. (437.1,88.87) -- cycle ;
\draw    (390,90) .. controls (390.14,102.14) and (400,117.83) .. (400,130) ;
\draw  [fill={rgb, 255:red, 0; green, 0; blue, 0 }  ,fill opacity=1 ] (407.1,90) .. controls (407.1,88.4) and (408.4,87.1) .. (410,87.1) .. controls (411.6,87.1) and (412.9,88.4) .. (412.9,90) .. controls (412.9,91.6) and (411.6,92.9) .. (410,92.9) .. controls (408.4,92.9) and (407.1,91.6) .. (407.1,90) -- cycle ;
\draw   (505,90) -- (545,90) -- (545,110) -- (505,110) -- cycle ;
\draw    (515,110) -- (515,130) ;
\draw    (535,110) -- (535,118.87) ;
\draw  [fill={rgb, 255:red, 0; green, 0; blue, 0 }  ,fill opacity=1 ] (532.1,118.87) .. controls (532.1,117.27) and (533.4,115.97) .. (535,115.97) .. controls (536.6,115.97) and (537.9,117.27) .. (537.9,118.87) .. controls (537.9,120.47) and (536.6,121.77) .. (535,121.77) .. controls (533.4,121.77) and (532.1,120.47) .. (532.1,118.87) -- cycle ;

\draw (430,26.6) node [anchor=south] [inner sep=0.75pt]    {$X$};
\draw (400,70) node  [font=\normalsize]  {$\alpha $};
\draw (450,133.4) node [anchor=north] [inner sep=0.75pt]    {$Y$};
\draw (320,80) node  [font=\normalsize]  {$\alpha $};
\draw (320,26.6) node [anchor=south] [inner sep=0.75pt]    {$X$};
\draw (330,133.4) node [anchor=north] [inner sep=0.75pt]    {$Y$};
\draw (310,133.4) node [anchor=north] [inner sep=0.75pt]    {$M$};
\draw (357.5,86.6) node [anchor=south] [inner sep=0.75pt]    {$=$};
\draw (545,70) node  [font=\normalsize]  {$\alpha _{0}$};
\draw (545,26.6) node [anchor=south] [inner sep=0.75pt]    {$X$};
\draw (565,133.4) node [anchor=north] [inner sep=0.75pt]    {$Y$};
\draw (400,133.4) node [anchor=north] [inner sep=0.75pt]    {$M$};
\draw (450,70) node  [font=\normalsize]  {$\alpha $};
\draw (487.5,86.6) node [anchor=south] [inner sep=0.75pt]    {$=$};
\draw (525,100) node  [font=\normalsize]  {$\alpha $};
\draw (515,133.4) node [anchor=north] [inner sep=0.75pt]    {$M$};
\end{tikzpicture}
\caption{{Productivity for cartesian categories.}}
\label{figure:productivity:cartesian}
\end{figure}
Now, assume that we have two representatives $\bra{\alpha_{i}} = \bra{\alpha_{j}}$
for which $\bra{\alpha_{i} ; u} = \bra{\alpha_{j} ; v}$.
By naturality of the discarding map, $\alpha_{i} ; \tid{\coUnit} = \alpha_{j} ; \tid{\coUnit}$,
and we call this map $\alpha_{Y}$.
Again by naturality of the discarding map, $\alpha_{i} ; u ; \tid{\coUnit} = \alpha_{j} ; v ; \tid{\coUnit}$,
and discarding the output in $Y$, we get that $\alpha_{M,i} ; u ;\tid{\coUnit} = \alpha_{M,j} ; v ; \tid{\coUnit}$,
which implies $\bra{\alpha_{M,i} ; u} = \bra{\alpha_{M,j} ; v}$.
\end{proof}

\begin{definition}[Opmonoidal comonad] In a monoidal category $(\catC,\otimes,I)$,
  a comonad $(R,\varepsilon,\delta)$
  is an \emph{opmonoidal comonad} when the endofunctor
  $R \colon \catC \to \catC$ is oplax monoidal with laxators
  $\psi_{X,Y} \colon R(X \otimes Y) \to RX \otimes RY$ and $\psi^{I} \colon RI \to I$,
  and both the counit $\varepsilon_{X} \colon RX \to X$ and the comultiplication
  $\delta_{X} \colon RX \to RRX$ are monoidal natural transformations.

  Explicitly, $\varepsilon_{I} = \psi_{0}$,
  $\varepsilon_{X \otimes Y} = \psi_{X,Y} ; (\varepsilon_X \otimes \varepsilon_Y)$,
  $\delta_{I} ; \psi_{0} ; \psi_{0} = \psi_{0}$ and
  $\delta_{X \otimes Y} ; \psi_{X,Y} ; \psi_{RX,RY} = \psi_{X,Y} ; (\delta_{X} \otimes \delta_{Y})$.

  Alternatively, an \emph{opmonoidal comonad} is a comonoid in the bicategory
  $\mathbf{MonOplax}$ of oplax monoidal functors with composition and monoidal
  natural transformations between them.
\end{definition}

\begin{definition}
  Let $(\catC,\otimes,I)$ be a symetric monoidal category. There is a functor
  $\defining{linknelist}{\ensuremath{\fun{List}^{+}}} \colon \catC \to \catC$ defined on objects by
  \[\neList(X)_{n} \coloneqq \bigotimes^{n}_{i=0} X_{i}.\]
  This functor is monoidal, with oplaxators $\psi^{+}_{0} \colon \neList(I) \to I$
  and $\psi_{X,Y} \colon \neList(X \otimes Y) \to \neList(X) \otimes \neList(Y)$ given by symmetries,
  associators and unitors.
\end{definition}

\begin{theorem}[From \Cref{th:nelist}]
  \label{ax:th:nelist}
  The opmonoidal functor $\neList$ has an opmonoidal comonad structure if and only
  if its base monoidal category $(\catC, \otimes, I)$ is cartesian monoidal.
\end{theorem}
\begin{proof}
  When $\catC$ is cartesian, we can construct the comonad structure using projections
  $\prod\nolimits_{i=0}^{n} X_{i} \to X_{n}$
  and copying together with braidings
  $\prod\nolimits_{i=0}^{n} X_{i} \to \prod\nolimits_{i=0}^{n}\prod\nolimits_{k=0}^{i} X_{k}$.
  These are monoidal natural transformations making $\neList$ a monoidal comonad.

  Suppose $(L,\varepsilon,\delta)$ is an opmonoidal comonad structure. This means it has
  families of natural transformations
  \[\delta_{n} \colon \bigotimes^{n}_{i=0} X_{i} \to \bigotimes^{n}_{i=0} \bigotimes^{i}_{k=0} X_{k}
    \mbox{ and } \varepsilon_{n} \colon \bigotimes^{n}_{i=0} X_{i} \to X_{n}.\]
  We will use these to construct a uniform counital comagma structure on every
  object of the category. By a refined version of Fox's theorem (\Cref{th:refinedfoxappendix}),
  this will imply that $\catC$ is cartesian monoidal.

  Let $X \in \catC$ be any object.
  Choosing $n=2$, $X_{0} = X$ and $X_{1} = I$; and using coherence maps,
  we get $\delta_{2} \colon X \to X \tensor X$ and $\varepsilon_{2} \colon X \to I$.
  These are coassociative, counital, natural and uniform because the corresponding transformations
$\delta$ and $\varepsilon$ are themselves coassociative, counital, natural and monoidal.
This induces a uniform comagma structure in every object $(X,\delta_{2},\varepsilon_{2})$;
with this structure,
every morphism of the category is a comagma homomorphism
because $\delta_{2}$ and $\varepsilon_{2}$ are natural.
\end{proof}

\begin{theorem}[Fox's theorem~\cite{fox76}]
  \label{th:fox}
  A \symmetricMonoidalCategory{} $(\catC,\otimes,I)$ is cartesian monoidal if and
  only if every object $X \in \catC$ has a cocommutative comonoid structure
  $(X,\varepsilon_{X},\delta_{X})$, every morphism of the category
  $f \colon X \to Y$ is a comonoid homomorphism, and this structure is uniform
  across the monoidal category: meaning that
  $\varepsilon_{X \otimes Y} = \varepsilon_{X} \otimes \varepsilon_{Y}$, that
  $\varepsilon_{I} = \mathrm{id}$, that $\delta_{I} = \mathrm{id}$ and that
  $\delta_{X \otimes Y} = (\delta_{X} \otimes \delta_{Y}) ; (\mathrm{id} \otimes \sigma_{X,Y} \otimes \mathrm{id})$.
\end{theorem}

\begin{remark}
  Most sources ask the comonoid structure in Fox's theorem
  (\Cref{th:fox}) to be cocommutative~\cite{fox76,fong2019supplying}.
  However, cocommutativity and coassociativity of the comonoid structure are implied by the fact that the structure is uniform and natural.
  We present an original refined version of Fox's theorem.
\end{remark}

\begin{theorem}[Refined Fox's theorem]
  \label{th:refinedfoxappendix}
  A symmetric monoidal category $(\catC,\otimes,I)$ is cartesian monoidal if and
  only if every object $X \in \catC$ has a counital comagma structure
  $(X,\varepsilon_{X},\delta_{X})$, or $(X,\blackComonoidUnit_{X},\blackComonoid_{X})$, every morphism of the category
  $f \colon X \to Y$ is a comagma homomorphism, and this structure is uniform
  across the monoidal category: meaning that
  $\varepsilon_{X \otimes Y} = \varepsilon_{X} \otimes \varepsilon_{Y}$,
  $\varepsilon_{I} = \mathrm{id}$, $\delta_{I} = \mathrm{id}$ and
  $\delta_{X \otimes Y} = (\delta_{X} \otimes \delta_{Y}) ; (\im \tensor \sigma_{X,Y} \tensor \im)$.
\end{theorem}
\begin{proof}
  We prove that such a comagma structure is necessarily coassociative and cocommutative.
  Note that any comagma homomorphism $f \colon A \to B$ must satisfy $\delta_{A} ; (f \tensor f) = f ; \delta_{B}$.
  In particular, $\delta_{X} \colon X \to X \tensor X$ must itself be a comagma homomorphism
  (see \Cref{figure:comultiplication}), meaning that
  \begin{equation}\label{eq:comultiplicationhomomorphism}
    \delta_{X} ; (\delta_X \tensor \delta_{X}) =
    \delta_{X} ; \delta_{X \tensor X} =
    \delta_{X} ; (\delta_X \tensor \delta_{X}) ; (\im \tensor \sigma_{X,Y} \tensor \im),
  \end{equation}
  where the second equality follows by uniformity.
\begin{figure}[h]
\tikzset{every picture/.style={line width=0.85pt}}
\begin{tikzpicture}[x=0.75pt,y=0.75pt,yscale=-1,xscale=1]
\draw    (25,65) .. controls (24.8,50.8) and (44.2,51.6) .. (45,65) ;
\draw  [fill={rgb, 255:red, 0; green, 0; blue, 0 }  ,fill opacity=1 ] (32.1,55) .. controls (32.1,53.4) and (33.4,52.1) .. (35,52.1) .. controls (36.6,52.1) and (37.9,53.4) .. (37.9,55) .. controls (37.9,56.6) and (36.6,57.9) .. (35,57.9) .. controls (33.4,57.9) and (32.1,56.6) .. (32.1,55) -- cycle ;
\draw    (35,45) -- (35,55) ;
\draw    (35,45) .. controls (34.8,30.8) and (65.2,30.8) .. (65,45) ;
\draw  [fill={rgb, 255:red, 0; green, 0; blue, 0 }  ,fill opacity=1 ] (47.1,35) .. controls (47.1,33.4) and (48.4,32.1) .. (50,32.1) .. controls (51.6,32.1) and (52.9,33.4) .. (52.9,35) .. controls (52.9,36.6) and (51.6,37.9) .. (50,37.9) .. controls (48.4,37.9) and (47.1,36.6) .. (47.1,35) -- cycle ;
\draw    (25,65) -- (25,80) ;
\draw    (45,65) -- (45,80) ;
\draw    (50,15) -- (50,35) ;
\draw    (55,65) .. controls (54.8,50.8) and (74.2,51.6) .. (75,65) ;
\draw  [fill={rgb, 255:red, 0; green, 0; blue, 0 }  ,fill opacity=1 ] (62.1,55) .. controls (62.1,53.4) and (63.4,52.1) .. (65,52.1) .. controls (66.6,52.1) and (67.9,53.4) .. (67.9,55) .. controls (67.9,56.6) and (66.6,57.9) .. (65,57.9) .. controls (63.4,57.9) and (62.1,56.6) .. (62.1,55) -- cycle ;
\draw    (65,45) -- (65,55) ;
\draw    (55,65) -- (55,80) ;
\draw    (75,65) -- (75,80) ;
\draw    (115,65) .. controls (114.8,50.8) and (135.4,53.8) .. (135,60) ;
\draw  [fill={rgb, 255:red, 0; green, 0; blue, 0 }  ,fill opacity=1 ] (122.1,55) .. controls (122.1,53.4) and (123.4,52.1) .. (125,52.1) .. controls (126.6,52.1) and (127.9,53.4) .. (127.9,55) .. controls (127.9,56.6) and (126.6,57.9) .. (125,57.9) .. controls (123.4,57.9) and (122.1,56.6) .. (122.1,55) -- cycle ;
\draw    (125,45) -- (125,55) ;
\draw    (125,45) .. controls (124.8,30.8) and (155.2,30.8) .. (155,45) ;
\draw  [fill={rgb, 255:red, 0; green, 0; blue, 0 }  ,fill opacity=1 ] (137.1,35) .. controls (137.1,33.4) and (138.4,32.1) .. (140,32.1) .. controls (141.6,32.1) and (142.9,33.4) .. (142.9,35) .. controls (142.9,36.6) and (141.6,37.9) .. (140,37.9) .. controls (138.4,37.9) and (137.1,36.6) .. (137.1,35) -- cycle ;
\draw    (115,70) -- (115,80) ;
\draw    (140,15) -- (140,35) ;
\draw    (145,60) .. controls (144.6,50.2) and (164.2,51.6) .. (165,65) ;
\draw  [fill={rgb, 255:red, 0; green, 0; blue, 0 }  ,fill opacity=1 ] (152.1,55) .. controls (152.1,53.4) and (153.4,52.1) .. (155,52.1) .. controls (156.6,52.1) and (157.9,53.4) .. (157.9,55) .. controls (157.9,56.6) and (156.6,57.9) .. (155,57.9) .. controls (153.4,57.9) and (152.1,56.6) .. (152.1,55) -- cycle ;
\draw    (155,45) -- (155,55) ;
\draw    (165,65) -- (165,75) ;
\draw    (115,65) -- (115,75) ;
\draw    (135,60) .. controls (135.4,74.2) and (144.6,65.8) .. (145,80) ;
\draw    (165,70) -- (165,80) ;
\draw    (145,60) .. controls (145.4,74.2) and (134.6,65.8) .. (135,80) ;
\draw (86,33.4) node [anchor=north west][inner sep=0.75pt]    {$=$};
\end{tikzpicture}
\caption{Comultiplication is a comagma homomorphism.}
\label{figure:comultiplication}
\end{figure}
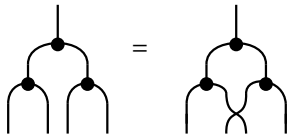

Now, we prove cocommutativity (\Cref{string:cocommutative}): composing both sides of \Cref{eq:comultiplicationhomomorphism} with $(\epsilon_{X} \tensor \im \tensor \im \tensor \epsilon_{X})$ discards the two external outputs and gives $\delta_{X} = \delta_{X} ; \sigma_{X}$.
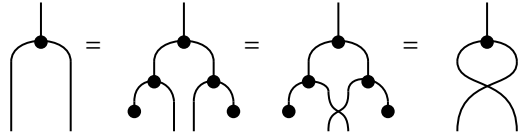
\begin{figure}[H]
\tikzset{every picture/.style={line width=0.75pt}} %
\begin{tikzpicture}[x=0.75pt,y=0.75pt,yscale=-1,xscale=1]
\draw    (72.1,65) .. controls (71.9,50.8) and (91.3,51.6) .. (92.1,65) ;
\draw  [fill={rgb, 255:red, 0; green, 0; blue, 0 }  ,fill opacity=1 ] (79.2,55) .. controls (79.2,53.4) and (80.5,52.1) .. (82.1,52.1) .. controls (83.7,52.1) and (85,53.4) .. (85,55) .. controls (85,56.6) and (83.7,57.9) .. (82.1,57.9) .. controls (80.5,57.9) and (79.2,56.6) .. (79.2,55) -- cycle ;
\draw    (82.1,45) -- (82.1,55) ;
\draw    (82.1,45) .. controls (81.9,30.8) and (112.3,30.8) .. (112.1,45) ;
\draw  [fill={rgb, 255:red, 0; green, 0; blue, 0 }  ,fill opacity=1 ] (94.2,35) .. controls (94.2,33.4) and (95.5,32.1) .. (97.1,32.1) .. controls (98.7,32.1) and (100,33.4) .. (100,35) .. controls (100,36.6) and (98.7,37.9) .. (97.1,37.9) .. controls (95.5,37.9) and (94.2,36.6) .. (94.2,35) -- cycle ;
\draw    (92.1,65) -- (92.1,80) ;
\draw    (97.1,15) -- (97.1,35) ;
\draw    (102.1,65) .. controls (101.9,50.8) and (121.3,51.6) .. (122.1,65) ;
\draw  [fill={rgb, 255:red, 0; green, 0; blue, 0 }  ,fill opacity=1 ] (109.2,55) .. controls (109.2,53.4) and (110.5,52.1) .. (112.1,52.1) .. controls (113.7,52.1) and (115,53.4) .. (115,55) .. controls (115,56.6) and (113.7,57.9) .. (112.1,57.9) .. controls (110.5,57.9) and (109.2,56.6) .. (109.2,55) -- cycle ;
\draw    (112.1,45) -- (112.1,55) ;
\draw    (102.1,65) -- (102.1,80) ;
\draw    (122.1,65) -- (122.1,70) ;
\draw    (150,65) .. controls (149.8,50.8) and (170.4,53.8) .. (170,60) ;
\draw  [fill={rgb, 255:red, 0; green, 0; blue, 0 }  ,fill opacity=1 ] (157.1,55) .. controls (157.1,53.4) and (158.4,52.1) .. (160,52.1) .. controls (161.6,52.1) and (162.9,53.4) .. (162.9,55) .. controls (162.9,56.6) and (161.6,57.9) .. (160,57.9) .. controls (158.4,57.9) and (157.1,56.6) .. (157.1,55) -- cycle ;
\draw    (160,45) -- (160,55) ;
\draw    (160,45) .. controls (159.8,30.8) and (190.2,30.8) .. (190,45) ;
\draw  [fill={rgb, 255:red, 0; green, 0; blue, 0 }  ,fill opacity=1 ] (172.1,35) .. controls (172.1,33.4) and (173.4,32.1) .. (175,32.1) .. controls (176.6,32.1) and (177.9,33.4) .. (177.9,35) .. controls (177.9,36.6) and (176.6,37.9) .. (175,37.9) .. controls (173.4,37.9) and (172.1,36.6) .. (172.1,35) -- cycle ;
\draw    (175,15) -- (175,35) ;
\draw    (180,60) .. controls (179.6,50.2) and (199.2,51.6) .. (200,65) ;
\draw  [fill={rgb, 255:red, 0; green, 0; blue, 0 }  ,fill opacity=1 ] (187.1,55) .. controls (187.1,53.4) and (188.4,52.1) .. (190,52.1) .. controls (191.6,52.1) and (192.9,53.4) .. (192.9,55) .. controls (192.9,56.6) and (191.6,57.9) .. (190,57.9) .. controls (188.4,57.9) and (187.1,56.6) .. (187.1,55) -- cycle ;
\draw    (190,45) -- (190,55) ;
\draw    (200,65) -- (200,70) ;
\draw    (150,65) -- (150,70) ;
\draw    (170,60) .. controls (170.4,74.2) and (179.6,65.8) .. (180,80) ;
\draw    (180,60) .. controls (180.4,74.2) and (169.6,65.8) .. (170,80) ;
\draw  [fill={rgb, 255:red, 0; green, 0; blue, 0 }  ,fill opacity=1 ] (119.2,70) .. controls (119.2,68.4) and (120.5,67.1) .. (122.1,67.1) .. controls (123.7,67.1) and (125,68.4) .. (125,70) .. controls (125,71.6) and (123.7,72.9) .. (122.1,72.9) .. controls (120.5,72.9) and (119.2,71.6) .. (119.2,70) -- cycle ;
\draw    (72.1,65) -- (72.1,70) ;
\draw  [fill={rgb, 255:red, 0; green, 0; blue, 0 }  ,fill opacity=1 ] (69.2,70) .. controls (69.2,68.4) and (70.5,67.1) .. (72.1,67.1) .. controls (73.7,67.1) and (75,68.4) .. (75,70) .. controls (75,71.6) and (73.7,72.9) .. (72.1,72.9) .. controls (70.5,72.9) and (69.2,71.6) .. (69.2,70) -- cycle ;
\draw    (10,45) -- (10,80) ;
\draw    (10,45) .. controls (9.8,30.8) and (40.2,30.8) .. (40,45) ;
\draw  [fill={rgb, 255:red, 0; green, 0; blue, 0 }  ,fill opacity=1 ] (22.1,35) .. controls (22.1,33.4) and (23.4,32.1) .. (25,32.1) .. controls (26.6,32.1) and (27.9,33.4) .. (27.9,35) .. controls (27.9,36.6) and (26.6,37.9) .. (25,37.9) .. controls (23.4,37.9) and (22.1,36.6) .. (22.1,35) -- cycle ;
\draw    (25,15) -- (25,35) ;
\draw    (40,45) -- (40,80) ;
\draw  [fill={rgb, 255:red, 0; green, 0; blue, 0 }  ,fill opacity=1 ] (147.1,70) .. controls (147.1,68.4) and (148.4,67.1) .. (150,67.1) .. controls (151.6,67.1) and (152.9,68.4) .. (152.9,70) .. controls (152.9,71.6) and (151.6,72.9) .. (150,72.9) .. controls (148.4,72.9) and (147.1,71.6) .. (147.1,70) -- cycle ;
\draw  [fill={rgb, 255:red, 0; green, 0; blue, 0 }  ,fill opacity=1 ] (197.1,70) .. controls (197.1,68.4) and (198.4,67.1) .. (200,67.1) .. controls (201.6,67.1) and (202.9,68.4) .. (202.9,70) .. controls (202.9,71.6) and (201.6,72.9) .. (200,72.9) .. controls (198.4,72.9) and (197.1,71.6) .. (197.1,70) -- cycle ;
\draw    (235,45) .. controls (234.8,30.8) and (265.2,30.8) .. (265,45) ;
\draw  [fill={rgb, 255:red, 0; green, 0; blue, 0 }  ,fill opacity=1 ] (247.1,35) .. controls (247.1,33.4) and (248.4,32.1) .. (250,32.1) .. controls (251.6,32.1) and (252.9,33.4) .. (252.9,35) .. controls (252.9,36.6) and (251.6,37.9) .. (250,37.9) .. controls (248.4,37.9) and (247.1,36.6) .. (247.1,35) -- cycle ;
\draw    (250,15) -- (250,35) ;
\draw    (265,45) .. controls (265.4,59.2) and (235,51.75) .. (235,80) ;
\draw    (235,45) .. controls (235.4,59.2) and (265,51.75) .. (265,80) ;
\draw (126,33.4) node [anchor=north west][inner sep=0.75pt]    {$=$};
\draw (46,33.4) node [anchor=north west][inner sep=0.75pt]    {$=$};
\draw (206,33.4) node [anchor=north west][inner sep=0.75pt]    {$=$};
\end{tikzpicture}
\caption{Cocommutativity}
\label{string:cocommutative}
\end{figure}

Now, we prove coassociativity (\Cref{string:coassociativity}): composing both sides of \Cref{eq:comultiplicationhomomorphism} with $(\im \tensor \epsilon_{X} \tensor \im \tensor \im)$ discards one of the middle outputs and gives $\delta_{X} ; (\im \tensor \delta_{X}) = \delta_{X} ; (\delta_{X} \tensor \im)$.
\begin{figure}
\tikzset{every picture/.style={line width=0.75pt}} %

\begin{tikzpicture}[x=0.75pt,y=0.75pt,yscale=-1,xscale=1]
\draw    (80,65) .. controls (79.8,50.8) and (99.2,51.6) .. (100,65) ;
\draw  [fill={rgb, 255:red, 0; green, 0; blue, 0 }  ,fill opacity=1 ] (87.1,55) .. controls (87.1,53.4) and (88.4,52.1) .. (90,52.1) .. controls (91.6,52.1) and (92.9,53.4) .. (92.9,55) .. controls (92.9,56.6) and (91.6,57.9) .. (90,57.9) .. controls (88.4,57.9) and (87.1,56.6) .. (87.1,55) -- cycle ;
\draw    (90,45) -- (90,55) ;
\draw    (90,45) .. controls (89.8,30.8) and (120.2,30.8) .. (120,45) ;
\draw  [fill={rgb, 255:red, 0; green, 0; blue, 0 }  ,fill opacity=1 ] (102.1,35) .. controls (102.1,33.4) and (103.4,32.1) .. (105,32.1) .. controls (106.6,32.1) and (107.9,33.4) .. (107.9,35) .. controls (107.9,36.6) and (106.6,37.9) .. (105,37.9) .. controls (103.4,37.9) and (102.1,36.6) .. (102.1,35) -- cycle ;
\draw    (80,65) -- (80,80) ;
\draw    (100,65) -- (100,80) ;
\draw    (105,15) -- (105,35) ;
\draw    (110,65) .. controls (109.8,50.8) and (129.2,51.6) .. (130,65) ;
\draw  [fill={rgb, 255:red, 0; green, 0; blue, 0 }  ,fill opacity=1 ] (117.1,55) .. controls (117.1,53.4) and (118.4,52.1) .. (120,52.1) .. controls (121.6,52.1) and (122.9,53.4) .. (122.9,55) .. controls (122.9,56.6) and (121.6,57.9) .. (120,57.9) .. controls (118.4,57.9) and (117.1,56.6) .. (117.1,55) -- cycle ;
\draw    (120,45) -- (120,55) ;
\draw    (130,65) -- (130,80) ;
\draw    (110,65) -- (110,70) ;
\draw    (160,65) .. controls (159.8,50.8) and (180.4,53.8) .. (180,60) ;
\draw  [fill={rgb, 255:red, 0; green, 0; blue, 0 }  ,fill opacity=1 ] (167.1,55) .. controls (167.1,53.4) and (168.4,52.1) .. (170,52.1) .. controls (171.6,52.1) and (172.9,53.4) .. (172.9,55) .. controls (172.9,56.6) and (171.6,57.9) .. (170,57.9) .. controls (168.4,57.9) and (167.1,56.6) .. (167.1,55) -- cycle ;
\draw    (170,45) -- (170,55) ;
\draw    (170,45) .. controls (169.8,30.8) and (200.2,30.8) .. (200,45) ;
\draw  [fill={rgb, 255:red, 0; green, 0; blue, 0 }  ,fill opacity=1 ] (182.1,35) .. controls (182.1,33.4) and (183.4,32.1) .. (185,32.1) .. controls (186.6,32.1) and (187.9,33.4) .. (187.9,35) .. controls (187.9,36.6) and (186.6,37.9) .. (185,37.9) .. controls (183.4,37.9) and (182.1,36.6) .. (182.1,35) -- cycle ;
\draw    (160,70) -- (160,80) ;
\draw    (185,15) -- (185,35) ;
\draw    (190,60) .. controls (189.6,50.2) and (209.2,51.6) .. (210,65) ;
\draw  [fill={rgb, 255:red, 0; green, 0; blue, 0 }  ,fill opacity=1 ] (197.1,55) .. controls (197.1,53.4) and (198.4,52.1) .. (200,52.1) .. controls (201.6,52.1) and (202.9,53.4) .. (202.9,55) .. controls (202.9,56.6) and (201.6,57.9) .. (200,57.9) .. controls (198.4,57.9) and (197.1,56.6) .. (197.1,55) -- cycle ;
\draw    (200,45) -- (200,55) ;
\draw    (210,65) -- (210,80) ;
\draw    (160,65) -- (160,75) ;
\draw    (180,60) .. controls (180.4,74.2) and (190,66) .. (190,75) ;
\draw    (190,60) .. controls (190.4,74.2) and (179.6,65.8) .. (180,80) ;
\draw    (15,65) .. controls (14.8,50.8) and (34.2,51.6) .. (35,65) ;
\draw  [fill={rgb, 255:red, 0; green, 0; blue, 0 }  ,fill opacity=1 ] (22.1,55) .. controls (22.1,53.4) and (23.4,52.1) .. (25,52.1) .. controls (26.6,52.1) and (27.9,53.4) .. (27.9,55) .. controls (27.9,56.6) and (26.6,57.9) .. (25,57.9) .. controls (23.4,57.9) and (22.1,56.6) .. (22.1,55) -- cycle ;
\draw    (25,45) -- (25,55) ;
\draw    (25,45) .. controls (24.8,30.8) and (55.2,30.8) .. (55,45) ;
\draw  [fill={rgb, 255:red, 0; green, 0; blue, 0 }  ,fill opacity=1 ] (37.1,35) .. controls (37.1,33.4) and (38.4,32.1) .. (40,32.1) .. controls (41.6,32.1) and (42.9,33.4) .. (42.9,35) .. controls (42.9,36.6) and (41.6,37.9) .. (40,37.9) .. controls (38.4,37.9) and (37.1,36.6) .. (37.1,35) -- cycle ;
\draw    (15,65) -- (15,80) ;
\draw    (35,65) -- (35,80) ;
\draw    (40,15) -- (40,35) ;
\draw    (55,45) -- (55,80) ;
\draw  [fill={rgb, 255:red, 0; green, 0; blue, 0 }  ,fill opacity=1 ] (107.1,70) .. controls (107.1,68.4) and (108.4,67.1) .. (110,67.1) .. controls (111.6,67.1) and (112.9,68.4) .. (112.9,70) .. controls (112.9,71.6) and (111.6,72.9) .. (110,72.9) .. controls (108.4,72.9) and (107.1,71.6) .. (107.1,70) -- cycle ;
\draw  [fill={rgb, 255:red, 0; green, 0; blue, 0 }  ,fill opacity=1 ] (187.1,75) .. controls (187.1,73.4) and (188.4,72.1) .. (190,72.1) .. controls (191.6,72.1) and (192.9,73.4) .. (192.9,75) .. controls (192.9,76.6) and (191.6,77.9) .. (190,77.9) .. controls (188.4,77.9) and (187.1,76.6) .. (187.1,75) -- cycle ;
\draw    (240,45) -- (240,80) ;
\draw    (240,45) .. controls (239.8,30.8) and (270.2,30.8) .. (270,45) ;
\draw  [fill={rgb, 255:red, 0; green, 0; blue, 0 }  ,fill opacity=1 ] (252.1,35) .. controls (252.1,33.4) and (253.4,32.1) .. (255,32.1) .. controls (256.6,32.1) and (257.9,33.4) .. (257.9,35) .. controls (257.9,36.6) and (256.6,37.9) .. (255,37.9) .. controls (253.4,37.9) and (252.1,36.6) .. (252.1,35) -- cycle ;
\draw    (255,15) -- (255,35) ;
\draw    (260,65) .. controls (259.8,50.8) and (279.2,51.6) .. (280,65) ;
\draw  [fill={rgb, 255:red, 0; green, 0; blue, 0 }  ,fill opacity=1 ] (267.1,55) .. controls (267.1,53.4) and (268.4,52.1) .. (270,52.1) .. controls (271.6,52.1) and (272.9,53.4) .. (272.9,55) .. controls (272.9,56.6) and (271.6,57.9) .. (270,57.9) .. controls (268.4,57.9) and (267.1,56.6) .. (267.1,55) -- cycle ;
\draw    (270,45) -- (270,55) ;
\draw    (260,65) -- (260,80) ;
\draw    (280,65) -- (280,80) ;
\draw (140,33.4) node [anchor=north west][inner sep=0.75pt]    {$=$};
\draw (60,33.4) node [anchor=north west][inner sep=0.75pt]    {$=$};
\draw (211,33.4) node [anchor=north west][inner sep=0.75pt]    {$=$};
\end{tikzpicture}
\caption{Coassociativity}
\label{string:coassociativity}
\end{figure}
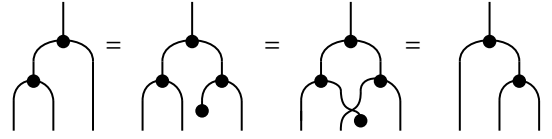

A coassociative and cocommutative comagma is a cocommutative comonoid.
We can then apply the classical form of Fox's theorem (\Cref{th:fox}).
\end{proof}

\paragraph{Distributive laws}
One could hope to add effects such as probability or non-determinism to set-based streams via the bikleisli category arising from a monad-comonad distributive law $\neList \circ \stream{T} \Rightarrow \stream{T} \circ \neList$ \cite{power99,beck69}, as proposed by Uustalu and Vene~\cite{uustalu05}.
This would correspond to a lifting of the $\neList$ comonad to the kleisli category of some commutative monad $T$; the arrows $\stream{X} \to \stream{Y}$ of such a category would look as follows,
\[f_{n} \colon X_{1} \times \dots \times X_{n} \to TY_{n}.\]
However, we have already shown that this will not result in a monoidal comonad whenever $\kleisli{T}$ is not cartesian.
To see explicitly what fails, we use the string diagrams to show how composition should work for the case $n = 2$ (\Cref{string:klcomposition}). This composition is not associative or unital whenever the kleisli category does not have natural comultiplications or counits, respectively (\Cref{string:klunitality,string:klassoc}).
\begin{figure}[h]
\tikzset{every picture/.style={line width=0.85pt}} %
\begin{tikzpicture}[x=0.75pt,y=0.75pt,yscale=-1,xscale=1]
\draw    (325,70) -- (325,78.87) ;
\draw    (310,90) .. controls (309.8,75.8) and (340.2,75.8) .. (340,90) ;
\draw  [fill={rgb, 255:red, 0; green, 0; blue, 0 }  ,fill opacity=1 ] (322.1,78.87) .. controls (322.1,77.27) and (323.4,75.97) .. (325,75.97) .. controls (326.6,75.97) and (327.9,77.27) .. (327.9,78.87) .. controls (327.9,80.47) and (326.6,81.77) .. (325,81.77) .. controls (323.4,81.77) and (322.1,80.47) .. (322.1,78.87) -- cycle ;
\draw   (340,120) -- (380,120) -- (380,140) -- (340,140) -- cycle ;
\draw   (315,160) -- (355,160) -- (355,180) -- (315,180) -- cycle ;
\draw    (360,70) .. controls (360.14,82.14) and (370,107.83) .. (370,120) ;
\draw   (300,90) -- (320,90) -- (320,110) -- (300,110) -- cycle ;
\draw    (340,90) .. controls (340.14,102.14) and (350,107.83) .. (350,120) ;
\draw    (335,180) -- (335,195) ;
\draw    (360,140) .. controls (360.14,152.14) and (345,147.83) .. (345,160) ;
\draw    (310,110) .. controls (310.2,135) and (324.6,139.8) .. (325,160) ;
\draw (325,66.6) node [anchor=south] [inner sep=0.75pt]    {$X_{0}$};
\draw (360,130) node  [font=\footnotesize]  {$f_{1}$};
\draw (310,100) node  [font=\footnotesize]  {$f_{0}$};
\draw (335,170) node  [font=\footnotesize]  {$g_{1}$};
\draw (360,66.6) node [anchor=south] [inner sep=0.75pt]    {$X_{1}$};
\draw (335,198.4) node [anchor=north] [inner sep=0.75pt]    {$Y_{1}$};
\end{tikzpicture}
\caption{Composition in the case $n=2$.}
\label{string:klcomposition}
\end{figure}
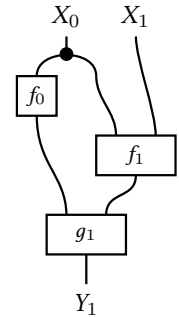
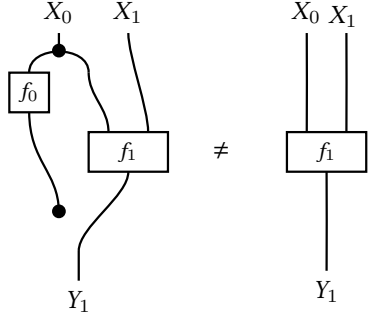
\begin{figure}[h]
\tikzset{every picture/.style={line width=0.85pt}} %
\begin{tikzpicture}[x=0.75pt,y=0.75pt,yscale=-1,xscale=1]
\draw    (325,70) -- (325,78.87) ;
\draw    (310,90) .. controls (309.8,75.8) and (340.2,75.8) .. (340,90) ;
\draw  [fill={rgb, 255:red, 0; green, 0; blue, 0 }  ,fill opacity=1 ] (322.1,78.87) .. controls (322.1,77.27) and (323.4,75.97) .. (325,75.97) .. controls (326.6,75.97) and (327.9,77.27) .. (327.9,78.87) .. controls (327.9,80.47) and (326.6,81.77) .. (325,81.77) .. controls (323.4,81.77) and (322.1,80.47) .. (322.1,78.87) -- cycle ;
\draw   (340,120) -- (380,120) -- (380,140) -- (340,140) -- cycle ;
\draw    (360,70) .. controls (360.14,82.14) and (370,107.83) .. (370,120) ;
\draw   (300,90) -- (320,90) -- (320,110) -- (300,110) -- cycle ;
\draw    (340,90) .. controls (340.14,102.14) and (350,107.83) .. (350,120) ;
\draw    (335,180) -- (335,195) ;
\draw    (360,140) .. controls (360.14,152.14) and (335,167.83) .. (335,180) ;
\draw    (310,110) .. controls (310.2,135) and (324.6,139.8) .. (325,160) ;
\draw  [fill={rgb, 255:red, 0; green, 0; blue, 0 }  ,fill opacity=1 ] (322.1,160) .. controls (322.1,158.4) and (323.4,157.1) .. (325,157.1) .. controls (326.6,157.1) and (327.9,158.4) .. (327.9,160) .. controls (327.9,161.6) and (326.6,162.9) .. (325,162.9) .. controls (323.4,162.9) and (322.1,161.6) .. (322.1,160) -- cycle ;
\draw   (440,120) -- (480,120) -- (480,140) -- (440,140) -- cycle ;
\draw    (470,70) .. controls (470.14,82.14) and (470,107.83) .. (470,120) ;
\draw    (450,70) .. controls (450.14,82.14) and (450,107.83) .. (450,120) ;
\draw    (460,140) .. controls (460.14,152.14) and (460,177.83) .. (460,190) ;
\draw (325,66.6) node [anchor=south] [inner sep=0.75pt]    {$X_{0}$};
\draw (360,130) node  [font=\footnotesize]  {$f_{1}$};
\draw (310,100) node  [font=\footnotesize]  {$f_{0}$};
\draw (360,66.6) node [anchor=south] [inner sep=0.75pt]    {$X_{1}$};
\draw (335,198.4) node [anchor=north] [inner sep=0.75pt]    {$Y_{1}$};
\draw (460,130) node  [font=\footnotesize]  {$f_{1}$};
\draw (401,123.4) node [anchor=north west][inner sep=0.75pt]    {$\neq $};
\draw (450,66.6) node [anchor=south] [inner sep=0.75pt]    {$X_{0}$};
\draw (468.5,68.6) node [anchor=south] [inner sep=0.75pt]    {$X_{1}$};
\draw (460,193.4) node [anchor=north] [inner sep=0.75pt]    {$Y_{1}$};
\end{tikzpicture}
\caption{Failure of unitality if discarding is not natural, as it happens, for instance, with partial functions.}
\label{string:klunitality}
\end{figure}

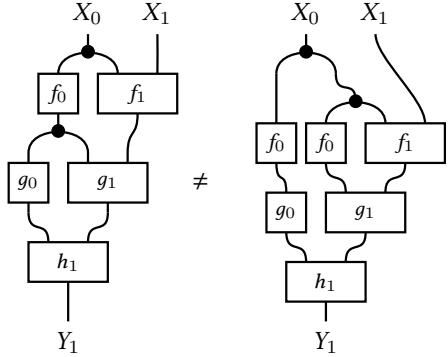
\begin{figure}
\tikzset{every picture/.style={line width=0.85pt}} %
\begin{tikzpicture}[x=0.75pt,y=0.75pt,yscale=-1,xscale=1]
\draw   (330,90) -- (370,90) -- (370,110) -- (330,110) -- cycle ;
\draw   (315,135) -- (355,135) -- (355,155) -- (315,155) -- cycle ;
\draw    (360,70) .. controls (360.14,82.14) and (360,77.83) .. (360,90) ;
\draw   (300,90) -- (320,90) -- (320,110) -- (300,110) -- cycle ;
\draw    (315,195) -- (315,215) ;
\draw    (350,110) .. controls (350.14,122.14) and (345,117.83) .. (345,130) ;
\draw    (325,70) -- (325,78.87) ;
\draw    (310,90) .. controls (309.8,75.8) and (340.2,75.8) .. (340,90) ;
\draw  [fill={rgb, 255:red, 0; green, 0; blue, 0 }  ,fill opacity=1 ] (322.1,78.87) .. controls (322.1,77.27) and (323.4,75.97) .. (325,75.97) .. controls (326.6,75.97) and (327.9,77.27) .. (327.9,78.87) .. controls (327.9,80.47) and (326.6,81.77) .. (325,81.77) .. controls (323.4,81.77) and (322.1,80.47) .. (322.1,78.87) -- cycle ;
\draw    (310,110) -- (310,118.87) ;
\draw    (295,130) .. controls (294.8,115.8) and (325.2,115.8) .. (325,130) ;
\draw  [fill={rgb, 255:red, 0; green, 0; blue, 0 }  ,fill opacity=1 ] (307.1,118.87) .. controls (307.1,117.27) and (308.4,115.97) .. (310,115.97) .. controls (311.6,115.97) and (312.9,117.27) .. (312.9,118.87) .. controls (312.9,120.47) and (311.6,121.77) .. (310,121.77) .. controls (308.4,121.77) and (307.1,120.47) .. (307.1,118.87) -- cycle ;
\draw   (285,135) -- (305,135) -- (305,155) -- (285,155) -- cycle ;
\draw   (295,175) -- (335,175) -- (335,195) -- (295,195) -- cycle ;
\draw    (295,160) .. controls (295.14,172.14) and (305,162.83) .. (305,175) ;
\draw    (335,160) .. controls (335.14,172.14) and (325,162.83) .. (325,175) ;
\draw    (435,70) -- (435,78.87) ;
\draw    (420,90) .. controls (419.8,75.8) and (450.2,75.8) .. (450,90) ;
\draw  [fill={rgb, 255:red, 0; green, 0; blue, 0 }  ,fill opacity=1 ] (432.1,78.87) .. controls (432.1,77.27) and (433.4,75.97) .. (435,75.97) .. controls (436.6,75.97) and (437.9,77.27) .. (437.9,78.87) .. controls (437.9,80.47) and (436.6,81.77) .. (435,81.77) .. controls (433.4,81.77) and (432.1,80.47) .. (432.1,78.87) -- cycle ;
\draw    (445,115) .. controls (444.8,100.8) and (475.2,100.8) .. (475,115) ;
\draw  [fill={rgb, 255:red, 0; green, 0; blue, 0 }  ,fill opacity=1 ] (457.1,103.87) .. controls (457.1,102.27) and (458.4,100.97) .. (460,100.97) .. controls (461.6,100.97) and (462.9,102.27) .. (462.9,103.87) .. controls (462.9,105.47) and (461.6,106.77) .. (460,106.77) .. controls (458.4,106.77) and (457.1,105.47) .. (457.1,103.87) -- cycle ;
\draw   (410,115) -- (430,115) -- (430,135) -- (410,135) -- cycle ;
\draw   (415,150) -- (435,150) -- (435,170) -- (415,170) -- cycle ;
\draw   (435,115) -- (455,115) -- (455,135) -- (435,135) -- cycle ;
\draw    (450,90) .. controls (450.14,102.14) and (460,91.7) .. (460,103.87) ;
\draw    (470,70) .. controls (470.14,82.14) and (495,102.83) .. (495,115) ;
\draw   (465,115) -- (505,115) -- (505,135) -- (465,135) -- cycle ;
\draw   (445,150) -- (485,150) -- (485,170) -- (445,170) -- cycle ;
\draw    (445,135) .. controls (445.14,147.14) and (455,137.83) .. (455,150) ;
\draw    (485,135) .. controls (485.14,147.14) and (475,137.83) .. (475,150) ;
\draw    (420,90) -- (420,115) ;
\draw    (420,135) .. controls (420.14,147.14) and (425,137.83) .. (425,150) ;
\draw   (425,185) -- (465,185) -- (465,205) -- (425,205) -- cycle ;
\draw    (425,170) .. controls (425.14,182.14) and (435,172.83) .. (435,185) ;
\draw    (465,170) .. controls (465.14,182.14) and (455,172.83) .. (455,185) ;
\draw    (445,205) -- (445,215) ;
\draw    (295,130) -- (295,135) ;
\draw    (325,130) -- (325,135) ;
\draw    (345,130) -- (345,135) ;
\draw    (335,155) -- (335,160) ;
\draw    (295,155) -- (295,160) ;

\draw (325,66.6) node [anchor=south] [inner sep=0.75pt]    {$X_{0}$};
\draw (350,100) node  [font=\footnotesize]  {$f_{1}$};
\draw (310,100) node  [font=\footnotesize]  {$f_{0}$};
\draw (335,145) node  [font=\footnotesize]  {$g_{1}$};
\draw (360,66.6) node [anchor=south] [inner sep=0.75pt]    {$X_{1}$};
\draw (315,218.4) node [anchor=north] [inner sep=0.75pt]    {$Y_{1}$};
\draw (295,145) node  [font=\footnotesize]  {$g_{0}$};
\draw (315,185) node  [font=\footnotesize]  {$h_{1}$};
\draw (420,125) node  [font=\footnotesize]  {$f_{0}$};
\draw (425,160) node  [font=\footnotesize]  {$g_{0}$};
\draw (445,125) node  [font=\footnotesize]  {$f_{0}$};
\draw (485,125) node  [font=\footnotesize]  {$f_{1}$};
\draw (465,160) node  [font=\footnotesize]  {$g_{1}$};
\draw (445,195) node  [font=\footnotesize]  {$h_{1}$};
\draw (445,218.4) node [anchor=north] [inner sep=0.75pt]    {$Y_{1}$};
\draw (435,66.6) node [anchor=south] [inner sep=0.75pt]    {$X_{0}$};
\draw (470,66.6) node [anchor=south] [inner sep=0.75pt]    {$X_{1}$};
\draw (376,138.4) node [anchor=north west][inner sep=0.75pt]    {$\neq $};

\end{tikzpicture}
\caption{Failure of associativity if copying is not natural, as it happens, for instance, with stochastic functions.}
\label{string:klassoc}
\end{figure}

\section{Productive categories}

\begin{definition}[\defining{linktruncatingcoherently}{Truncating coherently}]\label{def:truncatingcoherently}
  Let $f_{n}^{k} \colon X_{n} \tensor M_{n-1} \to Y_{n} \tensor M_{n}$ be a family
  of families of morphisms of increasing length, indexed by $k \in \naturals$ and $n \leq k$.
  We say that this family \emph{truncates coherently} if
  $\bra{f_{0}^{p}|\dots|f_{n}^{p}} = \bra{f_{0}^{q}|\dots|f_{n}^{q}}$
  for each $p, q \in \naturals$ and each $n \leq \min\{p,q\}$.
\end{definition}

\begin{lemma}[Factoring a family of processes]\label{lemma:familyfactoring}
  In a \productive{} category, let
  $(\bra{f_{0}^{k}|\dots|f_{k}^{k}})_{k \in \naturals}$ be a sequence
  of sequences that \truncatesCoherently{}. Then, there exists a sequence
  $h_{i}$ with $s^{k}_{i-1}f^{k}_{i} = h_{i}s^{k}_{i}$ such that, for each $k \in \naturals$ and each $n \leq k$,
  $\bra{f_{0}^{k}|\dots|f_{n}^{k}} = \bra{h_{0}|\dots|h_{n}}$.
  Moreover, this family $h_{i}$ is such that
  $\bra{h_{0}\dots h_{n}s^{p}_{n}u} = \bra{h_{0}\dots h_{n}s^{q}_{n}v}$
  implies $\bra{s^{p}_{n}u} = \bra{s^{q}_{n}v}$.
\end{lemma}
\begin{proof}
  We construct the family by induction. In the case $n=0$, we use that
  the family \truncatesCoherently{} to have that $\bra{f^{p}_{0}} = \bra{f^{q}_{0}}$
  and thus, by \productivity{}, create an $h_{0}$ with $f^{k}_{0} = h_{0}s^{k}_{0}$
  such that $\bra{h_{0}s^{p}_{0}u} = \bra{h_{0}s^{q}_{0}v}$ implies $\bra{s^{p}_{0}u} = \bra{s^{q}_{0}v}$.

  In the general case, assume we already have constructed $h_{0},\dots,h_{n-1}$
  with $s^{k}_{i-1}f^{k}_{i} = h_{i}s^{k}_{i}$ such that, for each $k \in \naturals$ and
  $\bra{f_{0}^{k}|\dots|f_{n-1}^{k}} = \bra{h_{0}|\dots|h_{n-1}}$. Moreover,
  $\bra{h_{0}\dots h_{n-1}s^{p}_{n-1}u} = \bra{h_{0}\dots h_{n-1}s^{q}_{n-1}v}$
  implies $\bra{s^{p}_{n-1}u} = \bra{s^{q}_{n-1}v}$.

  In this case, we use the fact that composition ``along a bar'' is dinatural:
  $\bra{f_{0}^p|\dots|f_{n}^{p}} = \bra{f_{0}^{q}|\dots|f_{n}^{q}}$ implies that
  $\bra{\tid{f_{0}^p\dots f_{n}^{p}}} = \bra{\tid{f_{0}^q\dots f_{n}^{q}}}$. This can be then
  rewritten as
  $$\bra{h_{0}\dots h_{n-1}s^{p}_{n-1}f_{n}^{p}} = \bra{h_{0}\dots h_{n-1}s^{q}_{n-1}f_{n}^{q}},$$
  which in turn implies
  $\bra{s^{p}_{n-1}f_{n}^{p}} = \bra{s^{q}_{n-1}f_{n}^{q}}$. By \productivity{},
  there exists $h_{n}$ with $s_{n-1}^{k}f^{k}_{n} = h_{n}s^{k}_{n}$ such that
  $\bra{f_{0}^{k}|\dots|f_{n}^{k}} = \bra{h_{0}|\dots|h_{n}}$.

  Finally, assume that
  $\bra{h_{0}\dots h_{n-1}h_{n}s_{n}^{p}u} = \bra{h_{0} \dots  h_{n-1}h_{n}s_{n}^{q}v}$. Thus, we
  have $\bra{h_{0}\dots  h_{n-1}s_{n-1}^{p}f_{n}^{p}u} = \bra{h_{0} \dots  h_{n-1}s_{n-1}^{q}f_{n}^{q}v}$
  and $\bra{s_{n-1}^{p}f_{n}^{p}u} = \bra{s_{n-1}^{q}f_{n}^{q}v}$. This can be rewritten as
  $\bra{h_{n}s_{n}^{p}u} = \bra{h_{n}s_{n}^{q}v}$, which in turn implies $\bra{s_{n}^{p}u} = \bra{s_{n}^{q}v}$.
  We have shown that the $h_{n}$ that we constructed satisfies the desired property.
\end{proof}

\begin{lemma}
  \label{lemma:observationalstatefulisterminalsequence}
  In a \productive{} category, the set of observational sequencesis isomorphic to the limit of the terminal
  sequence of the endofunctor $(\hom{} \odot\, \bullet)$ via the canonical
  map between them.
\end{lemma}
\begin{proof}
  We start by noting that observational equivalence of sequences is, by
  definition, the same thing as being equal under the canonical map to the
  limit of the terminal sequence
  \[\lim_{n} \int^{M_{0},\dots,M_{n}} \prod_{i=0}^{n} \hom{}(X_{i} \tensor M_{i-1}, Y_{i} \tensor M_{i}).\]
  We will show that this canonical map is surjective. That means that the domain quotiented by
  equality under the map is isomorphic to the codomain, q.e.d.

  Indeed, given any family $f_{n}^{k}$ that \truncatesCoherently{}, we can
  apply \Cref{lemma:familyfactoring} to find a sequence $h_{i}$ such
  that $\bra{f_{0}^{k}|\dots|f_{n}^{k}} = \bra{h_{0}|\dots|h_{n}}$. This means
  that it is the image of the stateful sequence $h_{i}$.
\end{proof}

\begin{lemma}[Factoring two processes]\label{lemma:multiplefactoring}
  In a \productive{} category, let $\bra{f_{0}} = \bra{g_{0}}$. Then there
  exists $h_{0}$ with $f_{0} = h_{0}s_{0}$ and $g_{0} = h_{0}t_{0}$ such that
  \[\bra{f_{0}|\dots|f_{n}} = \bra{g_{0}|\dots|g_{n}}\]
  implies the existence of a family $h_{i}$ together with $s_{i}$ and $t_{i}$ such that
  $s_{i-1}f_{i} = h_{i}s_{i}$ and $t_{i-1}g_{i} = h_{i}t_{i}$; and moreover, such that
  \[\bra{h_{0}\dots h_{n}s_{n}u} = \bra{h_{0}\dots h_{n}t_{n}v} \mbox{ implies } \bra{s_{n}u} = \bra{t_{n}v}.\]
\end{lemma}
\begin{proof}
  By \productivity{}, we can find such a factorization $f_{0} = h_{0}s_{0}$ and $g_{0} = h_{0}t_{0}$.

  Assume now that we have a family of morphisms such that $\bra{f_{0}|\dots|f_{n}} = \bra{g_{0}|\dots|g_{n}}$.
  We proceed by induction on $n$, the size of the family. The case $n = 0$ follows from
  the definition of \productive{} category.

  In the general case, we will construct the relevant $h_{n}$. The assumption $\bra{f_{0}|\dots|f_{n}} = \bra{g_{0}|\dots|g_{n}}$
  implies, in particular, that
  $\bra{f_{0}|\dots|f_{n-1}} = \bra{g_{0}|\dots|g_{n-1}}$.
  Thus, by induction
  hypothesis, there exist $h_{1},\dots,h_{n-1}$ together with
  $s_{i-1}f_{i} = h_{i}s_{i}$ and $t_{i-1}g_{i} = h_{i}t_{i}$, such that
  \[\bra{h_{0}\dots h_{n}s_{n-1}u} = \bra{h_{0}\dots h_{n}t_{n-1}v} \mbox{ implies } \bra{s_{n-1}u} = \bra{t_{n-1}v}.\]
  We know that $\bra{f_{0}\dots f_{n}} = \bra{g_{0} \dots g_{n}}$ and thus,
  \[\bra{h_{0}\dots h_{n-1}s_{n-1}f_{n}} = \bra{h_{0} \dots h_{n-1}t_{n-1}g_{n}},\]
  which, by induction hypothesis, implies $\bra{s_{n-1}f_{n}} = \bra{t_{n-1}g_{n}}$.
  By \productivity{}, there exists $h_{n}$ with $s_{n-1}f_{n} = h_{n}s_{n}$ and
  $t_{n-1}g_{n} = h_{n}t_{n}$ such that
  $\bra{h_{n}s_{n}u} = \bra{h_{n}t_{n}v} \mbox{ implies } \bra{s_{n}u} = \bra{t_{n}v}$.

  Finally, assume that
  $\bra{h_{0}\dots h_{n-1}h_{n}s_{n}u} = \bra{h_{0} \dots  h_{n-1}h_{n}t_{n}v}$. Thus, we
  have $\bra{h_{0}\dots  h_{n-1}s_{n-1}f_{n}u} = \bra{h_{0} \dots  h_{n-1}t_{n-1}g_{n}v}$
  and $\bra{s_{n-1}f_{n}u} = \bra{t_{n-1}g_{n}v}$. This can be rewritten as
  $\bra{h_{n}s_{n}u} = \bra{h_{n}t_{n}v}$, which in turn implies $\bra{s_{n}u} = \bra{t_{n}v}$.
  We have shown that the $h_{n}$ that we constructed satisfies the desired property.
\end{proof}

\begin{lemma}[Removing the first step]\label{lemma:removestep}
  In a \productive{} category, let $\bra{f_{0}} = \bra{g_{0}}$. Then there
  exists $h$ with $f_{0} = hs$ and $g_{0} = ht$ such that
  $\bra{f_{0}|\dots|f_{n}} = \bra{g_{0}|\dots|g_{n}}$ implies
  $\bra{sf_{1}|\dots|f_{n}} = \bra{tg_{1}|\dots|g_{n}}$.
\end{lemma}
\begin{proof}
  By \Cref{lemma:multiplefactoring}, we obtain a factorization $f_{0} = hs$ and
  $g_{0} = ht$. Moreover, each time that we have
  $\bra{f_{0}|\dots|f_{n}} = \bra{g_{0}|\dots|g_{n}}$, we can obtain a family
  $h_{i}$ together with $s_{i}$ and $t_{i}$ such that
  $s_{i-1}f_{i} = h_{i}s_{i}$ and $t_{i-1}f_{i} = h_{i}t_{i}$. Using the fact
  that $\bra{h_{n}s_{n}} = \bra{h_{n}t_{n}}$, we have that
  $\bra{h_{1}|\dots|h_{n}s_{n}} = \bra{h_{1}|\dots|h_{n}t_{n}}$, which can be
  rewritten using dinaturality as
  $\bra{s_{0}f_{1}|\dots|f_{n}} = \bra{t_{0}g_{1}|\dots|g_{n}}$.
\end{proof}

\begin{lemma}
  \label{lemma:coalgebraexists}
  In a \productive{} category, the final coalgebra of the endofunctor
  $(\hom{} \odot\, \bullet)$ does exist and it is given by the limit of the
  terminal sequence
  \[L \coloneqq \lim_{n} \int^{M_{0},\dots,M_{n}} \prod_{i=0}^{n} \hom{}(X_{i} \tensor M_{i-1}, Y_{i} \tensor M_{i}).\]
\end{lemma}
\begin{proof}
  We will apply \Cref{th:adamek}. The endofunctor
  $(\hom{} \odot\, \bullet)$
  acts on the category $[(\NcatC)^{op} \times \NcatC, \Set]$, which, being a presheaf
  category, has all small limits. We will show that there is an isomorphism
  $\hom{} \odot\ L \cong L$ given by the canonical morphism between them.

  First, note that the set $L(\stream{X};\stream{Y})$ is, explicitly,
  \[\lim_{n} \int^{M_{1},\dots,M_{n}} \prod_{i=1}^{n} \hom{}(X_{i} \tensor M_{i-1}, Y_{i} \tensor M_{i}).\]
  A generic element from this set is a \emph{sequence of sequences of increasing length}.
  Moreover, the sequences must \emph{truncate coherently} (\Cref{def:truncatingcoherently}).

  Secondly, note that the set
  $(\hom{} \odot\ L)(\stream{X};\stream{Y})$ is, explicitly,
  \begin{align*}\int^{M_{0}}\hom{}(&X_{0}, Y_{0} \tensor M_{0}) \times \\ & \lim_{n} \int^{M_{1},\dots,M_{n}} \prod_{i=1}^{n} \hom{}(X_{i} \tensor M_{i-1}, Y_{i} \tensor M_{i}).\end{align*}
  A generic element from this set is of the form
  \[\bra{f|(\bra{f_{1}^{k}|\dots|f_{k}^{k}})_{k \in \naturals}},\]
  that is, a pair consisting on a first morphism
  $f \colon X_{0} \to Y_{0} \tensor M_{0}$ and a family of sequences
  $(\bra{f_{1}^{k}|\dots|f_{k}^{k}})$, quotiented by dinaturality of $M_{0}$ and
  truncating coherently. The canonical map to $L(\stream{X};\stream{Y})$ maps
  this generic element to the family of sequences
  $(\bra{f_{0}|f_{1}^{k}|\dots|f_{k}^{k}})_{k \in \naturals}$, which truncates
  coherently because the previous family did and we are precomposing with $f_{0}$,
  which is dinatural.

  Thirdly, this map is injective. Imagine a pair of elements
  $\bra{f_{0}|(\bra{f_{1}^{k}|\dots|f_{k}^{k}})_{k \in \naturals}}$ and
  $\bra{g_{0}|(\bra{g_{1}^{k}|\dots|g_{k}^{k}})_{k \in \naturals}}$ that have the
  same image, meaning that, for each $k \in \naturals$,
  \[\bra{f_{0}|f_{1}^{k}|\dots|f_{k}^{k}} = \bra{g_{0}|g_{1}^{k}|\dots|g_{k}^{k}}.\]
  By \Cref{lemma:removestep}, we can find $h$ with $f_{0} = hs$ and $g_{0} = ht$
  such that, for each $k \in \naturals$,
  $\bra{sf_{1}^{k}|\dots|f_{k}^{k}} = \bra{tg_{1}^{k}|\dots|g_{k}^{k}}$. Thus,
  \[\begin{aligned}
    \bra{f_{0}|(\bra{f_{1}^{k}|\dots|f_{k}^{k}})_{k \in \naturals}} =
    \bra{h|(\bra{sf_{1}^{k}|\dots|f_{k}^{k}})_{k \in \naturals}} = \\
    \bra{h|(\bra{tg_{1}^{k}|\dots|g_{k}^{k}})_{k \in \naturals}} =
    \bra{g_{0}|(\bra{tg_{1}^{k}|\dots|g_{k}^{k}})_{k \in \naturals}}.
  \end{aligned}\]

Finally, this map is also surjective. From \Cref{lemma:familyfactoring}, it follows
that any family that \truncatesCoherently{} can be equivalently written as $\bra{h_{0}|\dots|h_{n}}_{n \in \naturals}$,
which is the image of the element $\bra{h_{0}|\bra{h_{1}|\dots|h_{n}}_{n \in \naturals}}$.
\end{proof}

\begin{theorem}[From \Cref{theorem:observationalfinalcoalgebra}]\label{appendix:theorem:observationalfinalcoalgebra}
  In a \productive{} category, the final coalgebra of the endofunctor
  $(\hom{} \odot\, \bullet)$ exists and it is given by the set of stateful
  sequences quotiented by observational equivalence.
  \[\left(\int^{M \in [\naturals,\catC]} \prod^{\infty}_{i=0} \hom{}(X_{i} \tensor M_{i-1}, Y_{i} \tensor M_{i})\right)
  \bigg/\approx\]
\end{theorem}
\begin{proof}
  By \Cref{lemma:coalgebraexists}, we know that the final coalgebra exists
  and is given by the limit of the terminal sequence.
  By \Cref{lemma:observationalstatefulisterminalsequence}, we know that it
  is isomorphic to the set of stateful sequences quotiented by observational
  equivalence.
\end{proof}
\section{Monoidal streams}

\begin{lemma}
  \label{lemma:sequentialcompositionassociativememories}
  Sequential composition of streams with memories (\Cref{def:sequentialstream}) is associative.
  Given three streams
  \begin{itemize}
    \item $f \in \STREAM(\sA \cdot \stream{X},\stream{Y})$,
    \item $g \in \STREAM(\sB \cdot \stream{Y},\stream{Z})$,
    \item and $h \in \STREAM(\sB \cdot \stream{Z},\stream{W})$;
  \end{itemize}
  we can compose them in two different ways,
  \begin{itemize}
    \item $(f^{\sA} ; g^{\sB}) ; h^{\sC} \in \STREAM((\sA \tensor \sB) \tensor \sC \cdot \stream{X},\stream{W})$, or
    \item $f^{\sA} ; (g^{\sB} ; h^{\sC}) \in \STREAM(\sA \tensor (\sB \tensor \sC) \cdot \stream{X},\stream{W})$.
  \end{itemize}
  We claim that
  \[((f^{\sA} ; g^{\sB}) ; h^{\sC}) = \alpha_{\sA,\sB,\sC} \cdot (f^{\sA} ; (g^{\sB} ; h^{\sC})).\]
\end{lemma}
\begin{proof}
  First, we note that both sides of the equation represent streams with different memories.
  \begin{itemize}
    \item $M((f^{\sA} ; g^{\sB}) ; h^{\sC}) = (M(f) \tensor M(g)) \tensor M(h)$,
    \item $M(f^{\sA} ; (g^{\sB} ; h^{\sC})) = M(f) \tensor (M(g) \tensor M(h))$.
  \end{itemize}
We will prove they are related by dinaturality over the associator $\alpha$.
We know that $\now((f^{\sA} ; g^{\sB}) ; h^{\sC}) = \now(f^{\sA} ; (g^{\sB} ; h^{\sC}))$ by string diagrams (see~\Cref{strings:assoc}).
Then, by coinduction, we know that
\[\begin{gathered}(\later(f)^{M(f)} ; \later(g)^{M(g)}) ; \later(h)^{M(h)} = \\
    \act{\alpha}{(\later(f)^{M(f)} ; (\later(g)^{M(g)} ; \later(h)^{M(h)}))},\end{gathered}\]
that is, $\later((f^{\sA} ; g^{\sB}) ; h^{\sC}) = \alpha \cdot \later(f^{\sA} ; (g^{\sB} ; h^{\sC}))$.
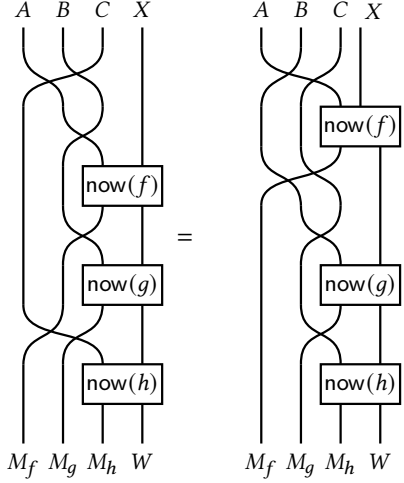
\begin{figure}[h]
\tikzset{every picture/.style={line width=0.85pt}}
\begin{tikzpicture}[x=0.75pt,y=0.75pt,yscale=-1,xscale=1]
\draw    (90,50) .. controls (89.83,69.77) and (49.5,61.43) .. (50,80) ;
\draw   (80,110) -- (120,110) -- (120,130) -- (80,130) -- cycle ;
\draw    (110,40) -- (110,110) ;
\draw    (90,80) .. controls (90,99.43) and (69.67,90.1) .. (70,110) ;
\draw    (70,80) .. controls (70,98.77) and (90.33,90.77) .. (90,110) ;
\draw    (50,50) .. controls (49.83,69.77) and (70.5,62.1) .. (70,80) ;
\draw    (70,50) .. controls (69.5,69.1) and (89.5,60.77) .. (90,80) ;
\draw    (70,40) -- (70,50) ;
\draw    (90,40) -- (90,50) ;
\draw    (50,40) -- (50,50) ;
\draw    (70,110) -- (70,130) ;
\draw    (50,80) -- (50,180) ;
\draw   (80,160) -- (120,160) -- (120,180) -- (80,180) -- cycle ;
\draw    (90,130) .. controls (89.67,150.43) and (69.67,140.43) .. (70,160) ;
\draw    (70,130) .. controls (70,150.43) and (90.33,140.43) .. (90,160) ;
\draw    (70,160) -- (70,180) ;
\draw    (110,130) -- (110,160) ;
\draw    (110,180) -- (110,210) ;
\draw   (80,210) -- (120,210) -- (120,230) -- (80,230) -- cycle ;
\draw    (90,230) -- (90,250) ;
\draw    (110,230) -- (110,250) ;
\draw    (50,210) -- (50,250) ;
\draw    (70,210) -- (70,250) ;
\draw    (50,180) .. controls (49.83,199.77) and (90.33,190.43) .. (90,210) ;
\draw    (70,180) .. controls (69.83,199.77) and (50.5,192.1) .. (50,210) ;
\draw    (90,180) .. controls (89.5,199.1) and (69.5,190.77) .. (70,210) ;
\draw    (210.01,100) .. controls (209.84,119.77) and (169.5,111.43) .. (170,130) ;
\draw   (200,80) -- (240,80) -- (240,100) -- (200,100) -- cycle ;
\draw    (210,180) .. controls (210,199.43) and (189.67,190.1) .. (190,210) ;
\draw    (190,180) .. controls (190,198.77) and (210.34,190.77) .. (210,210) ;
\draw    (190,100) .. controls (189.83,119.77) and (210.49,112.1) .. (209.99,130) ;
\draw    (170,100) .. controls (169.5,119.1) and (189.5,110.77) .. (190,130) ;
\draw    (190,40) -- (190,50) ;
\draw    (210,40) -- (210,50) ;
\draw    (170,40) -- (170,50) ;
\draw    (170,80) -- (170,100) ;
\draw   (200,160) -- (240,160) -- (240,180) -- (200,180) -- cycle ;
\draw    (209.99,130) .. controls (209.66,150.43) and (189.66,140.43) .. (189.99,160) ;
\draw    (190,130) .. controls (190,150.43) and (210.33,140.43) .. (210,160) ;
\draw    (190,160) -- (190,168.85) -- (190,180) ;
\draw    (230.01,100) -- (230.01,160) ;
\draw    (220,40) -- (220,80) ;
\draw   (200,210) -- (240,210) -- (240,230) -- (200,230) -- cycle ;
\draw    (230,230) -- (230,250) ;
\draw    (210,230) -- (210,250) ;
\draw    (190,210) -- (190,250) ;
\draw    (170,130) -- (169.99,250) ;
\draw    (170,50) .. controls (169.83,69.77) and (210.33,60.43) .. (210,80) ;
\draw    (190,50) .. controls (189.83,69.77) and (170.5,62.1) .. (170,80) ;
\draw    (210,50) .. controls (209.5,69.1) and (189.5,60.77) .. (190,80) ;
\draw    (190,80) -- (190,100) ;
\draw    (230,180) -- (230,210) ;
\draw (126,142.4) node [anchor=north west][inner sep=0.75pt]    {$=$};
\draw (50,253.4) node [anchor=north] [inner sep=0.75pt]  [font=\small]  {$M_{f}$};
\draw (50,36.6) node [anchor=south] [inner sep=0.75pt]  [font=\small]  {$A$};
\draw (70,36.6) node [anchor=south] [inner sep=0.75pt]  [font=\small]  {$B$};
\draw (90,36.6) node [anchor=south] [inner sep=0.75pt]  [font=\small]  {$C$};
\draw (110,36.6) node [anchor=south] [inner sep=0.75pt]  [font=\small]  {$X$};
\draw (70,253.4) node [anchor=north] [inner sep=0.75pt]  [font=\small]  {$M_{g}$};
\draw (90,253.4) node [anchor=north] [inner sep=0.75pt]  [font=\small]  {$M_{h}$};
\draw (110,253.4) node [anchor=north] [inner sep=0.75pt]  [font=\small]  {$W$};
\draw (170,36.6) node [anchor=south] [inner sep=0.75pt]  [font=\small]  {$A$};
\draw (190,36.6) node [anchor=south] [inner sep=0.75pt]  [font=\small]  {$B$};
\draw (210,36.6) node [anchor=south] [inner sep=0.75pt]  [font=\small]  {$C$};
\draw (227,37.6) node [anchor=south] [inner sep=0.75pt]  [font=\small]  {$X$};
\draw (169.99,253.4) node [anchor=north] [inner sep=0.75pt]  [font=\small]  {$M_{f}$};
\draw (190,253.4) node [anchor=north] [inner sep=0.75pt]  [font=\small]  {$M_{g}$};
\draw (210,253.4) node [anchor=north] [inner sep=0.75pt]  [font=\small]  {$M_{h}$};
\draw (230,253.4) node [anchor=north] [inner sep=0.75pt]  [font=\small]  {$W$};
\draw (100,120) node  [font=\small]  {${\textstyle \mathsf{now}( f)}$};
\draw (100,170) node  [font=\small]  {${\textstyle \mathsf{now}( g)}$};
\draw (100,220) node  [font=\small]  {${\textstyle \mathsf{now}( h)}$};
\draw (220,90) node  [font=\small]  {${\textstyle \mathsf{now}( f)}$};
\draw (220,170) node  [font=\small]  {${\textstyle \mathsf{now}( g)}$};
\draw (220,220) node  [font=\small]  {${\textstyle \mathsf{now}( h)}$};
\end{tikzpicture}
\caption{Associativity for sequential composition}
\label{strings:assoc}
\end{figure}
\end{proof}

\begin{lemma}
  \label{lemma:sequentialcompositionassociative}
  Sequential composition of streams (\Cref{def:sequentialstream}) is associative.
  Given three streams
  \begin{itemize}
    \item $f \in \STREAM(\stream{X},\stream{Y})$,
    \item $g \in \STREAM(\stream{Y},\stream{Z})$,
    \item and $h \in \STREAM(\stream{Z},\stream{W})$;
  \end{itemize}
  we claim that $((f ; g) ; h) = (f ; (g ; h))$.
\end{lemma}
\begin{proof}
  Direct consequence of~\Cref{lemma:sequentialcompositionassociativememories},
  after considering the appropriate coherence morphisms.
\end{proof}

\begin{lemma}\label{lemma:parfunassociativememories}
  Parallel composition of streams with memories is functorial with regards to sequential composition of streams with memories.
  Given four streams
  \begin{itemize}
    \item $f \in \STREAM(\sA \cdot \stream{X},\stream{Y})$,
    \item $f' \in \STREAM(\sA' \cdot \stream{X}',\stream{Y}')$,
    \item $g \in \STREAM(\sB \cdot \stream{Y},\stream{Z})$, and
    \item $g' \in \STREAM(\sB' \cdot \stream{Y}',\stream{Z}')$,
  \end{itemize}
  we can compose them in two different ways,
  \begin{itemize}
    \item $\Ncomp{(\Ntensor{f}{A}{f'}{A'})}{(A \tensor A')}{(\Ntensor{g}{B}{g'}{B'})}{(B \tensor B')}$, and
    \item $\Ntensor{(\Ncomp{f}{A}{g}{B})}{(A \tensor B)}{(\Ncomp{f'}{A'}{g'}{B'})}{(A' \tensor B')}$,
  \end{itemize}
  having slightly different types, respectively,
  \begin{itemize}
    \item $\STREAM(\act{(A \tensor A') \tensor (B \tensor B')}{\stream{X} \tensor \stream{X}'},\stream{Z} \tensor \stream{Z}')$, and
    \item $\STREAM(\act{(A \tensor B) \tensor (A' \tensor B')}{\stream{X} \tensor \stream{X}'},\stream{Z} \tensor \stream{Z}')$.
  \end{itemize}
  We claim that
  \[\begin{gathered}
      \Ncomp{(\Ntensor{f}{A}{f'}{A'})}{(A \tensor A')}{(\Ntensor{g}{B}{g'}{B'})}{(B \tensor B')} = \\
      \act{\sigma_{A',B}}{\Ntensor{(\Ncomp{f}{A}{g}{B})}{(A \tensor B)}{(\Ncomp{f'}{A'}{g'}{B'})}{(A' \tensor B')}}.
    \end{gathered}\]
\end{lemma}
\begin{proof}
  First, we note that both sides of the equation (which, from now on, we call $LHS$ and $RHS$, respectively)
represent strams with different memories.
  \begin{itemize}
    \item $M(LHS) = (M(f) \tensor M(f')) \tensor (M(g) \tensor M(g'))$,
    \item $M(RHS) = (M(f) \tensor M(g)) \tensor (M(f') \tensor M(g'))$.
  \end{itemize}
We will prove they are related by dinaturality over the symmetry $\sigma$.
We know that $\now(LHS) ; \sigma = \now(RHS)$ by string diagrams (see \Cref{strings:functorpar}).
Then, by coinduction, we know that $\later(LHS) = \act{\sigma}{\later(RHS)}$.
\end{proof}

\begin{figure}
\tikzset{every picture/.style={line width=0.85pt}}
\begin{tikzpicture}[x=0.75pt,y=0.75pt,yscale=-1,xscale=1]
\draw    (310,40) -- (310,50) ;
\draw    (330,40) -- (330,50) ;
\draw    (290,40) -- (290,50) ;
\draw    (370,40) -- (370,80) ;
\draw    (390,40) -- (390,110) ;
\draw    (350,40) -- (350,50) ;
\draw    (350,50) .. controls (349.83,69.77) and (309.5,61.43) .. (310,80) ;
\draw    (330,50) .. controls (329.83,69.77) and (289.5,61.43) .. (290,80) ;
\draw    (290,50) .. controls (289.83,69.77) and (329.5,61.43) .. (330,80) ;
\draw    (310,50) .. controls (309.83,69.77) and (349.5,61.43) .. (350,80) ;
\draw    (350,80) .. controls (350,100.43) and (370.33,90.43) .. (370,110) ;
\draw    (370,80) .. controls (370,100.43) and (350.33,90.43) .. (350,110) ;
\draw   (320,110) -- (360,110) -- (360,130) -- (320,130) -- cycle ;
\draw   (360,110) -- (400,110) -- (400,130) -- (360,130) -- cycle ;
\draw    (350,130) .. controls (350,150.43) and (370.33,140.43) .. (370,160) ;
\draw    (370,130) .. controls (370,150.43) and (350.33,140.43) .. (350,160) ;
\draw    (330,80) -- (330,110) ;
\draw    (330,130) -- (330,160) ;
\draw    (390,130) -- (390,220) ;
\draw    (350,160) .. controls (349.84,179.77) and (309.5,171.43) .. (310,190) ;
\draw    (330,160) .. controls (329.84,179.77) and (289.5,171.43) .. (290,190) ;
\draw    (290,160) .. controls (289.84,179.77) and (329.5,171.43) .. (330,190) ;
\draw    (310,160) .. controls (309.84,179.77) and (349.5,171.43) .. (350,190) ;
\draw    (310,80) -- (310,160) ;
\draw    (290,80) -- (290,160) ;
\draw   (321.5,220) -- (361.5,220) -- (361.5,240) -- (321.5,240) -- cycle ;
\draw   (361.5,220) -- (401.5,220) -- (401.5,240) -- (361.5,240) -- cycle ;
\draw    (350,190) .. controls (350,210.43) and (370.33,200.43) .. (370,220) ;
\draw    (370,190) .. controls (370,210.43) and (350.33,200.43) .. (350,220) ;
\draw    (350,240) .. controls (350,260.43) and (370.33,250.43) .. (370,270) ;
\draw    (370,240) .. controls (370,260.43) and (350.33,250.43) .. (350,270) ;
\draw    (370,160) -- (370,190) ;
\draw    (390,240) -- (390,270) ;
\draw    (330,240) -- (330,270) ;
\draw    (290,190) -- (290,270) ;
\draw    (310,190) -- (310,270) ;
\draw    (330,190) -- (330,220) ;
\draw    (442.5,40) -- (442.5,100) ;
\draw    (522.5,40) -- (522.5,70) ;
\draw    (542.5,40) -- (542.5,130) ;
\draw    (482.5,100) -- (482.5,130) ;
\draw    (502.5,100) .. controls (502.5,120.43) and (522.83,110.43) .. (522.5,130) ;
\draw    (522.5,100) .. controls (522.5,120.43) and (502.83,110.43) .. (502.5,130) ;
\draw    (462.5,40) .. controls (462.5,60.43) and (482.83,50.43) .. (482.5,70) ;
\draw    (482.5,40) .. controls (482.5,60.43) and (462.83,50.43) .. (462.5,70) ;
\draw    (522.5,70) .. controls (522.33,89.77) and (482,81.43) .. (482.5,100) ;
\draw    (482.5,70) .. controls (482.33,89.77) and (502,81.43) .. (502.5,100) ;
\draw    (502.5,70) .. controls (502.33,89.77) and (522,81.43) .. (522.5,100) ;
\draw   (512.5,140) -- (552.5,140) -- (552.5,160) -- (512.5,160) -- cycle ;
\draw    (462.5,70) -- (462.5,100) ;
\draw    (442.5,100) .. controls (442.5,120.43) and (462.83,110.43) .. (462.5,130) ;
\draw    (462.5,100) .. controls (462.5,120.43) and (442.83,110.43) .. (442.5,130) ;
\draw   (452.5,140) -- (492.5,140) -- (492.5,160) -- (452.5,160) -- cycle ;
\draw    (502.5,170) .. controls (502.5,190.43) and (522.83,180.43) .. (522.5,200) ;
\draw    (522.5,170) .. controls (522.5,190.43) and (502.83,180.43) .. (502.5,200) ;
\draw    (442.5,170) .. controls (442.5,190.43) and (462.83,180.43) .. (462.5,200) ;
\draw    (462.5,170) .. controls (462.5,190.43) and (442.83,180.43) .. (442.5,200) ;
\draw   (512.5,210) -- (552.5,210) -- (552.5,230) -- (512.5,230) -- cycle ;
\draw   (452.5,210) -- (492.5,210) -- (492.5,230) -- (452.5,230) -- cycle ;
\draw    (502.5,40) -- (502.5,70) ;
\draw    (482.5,170) -- (482.5,200) ;
\draw    (442.5,130) -- (442.5,160) ;
\draw    (502.5,140) -- (502.5,160) ;
\draw    (542.5,170) -- (542.5,200) ;
\draw    (482.5,230) .. controls (482.33,249.77) and (522,241.43) .. (522.5,260) ;
\draw    (502.5,230) .. controls (502.33,249.77) and (482,241.43) .. (482.5,260) ;
\draw    (522.5,230) .. controls (522.33,249.77) and (502,241.43) .. (502.5,260) ;
\draw    (502.5,210) -- (502.5,230) ;
\draw    (542.5,230) -- (542.5,270) ;
\draw    (442.5,210) -- (442.5,260) ;
\draw    (462.5,230) -- (462.5,260) ;
\draw    (522.5,260) -- (522.5,270) ;
\draw    (502.5,260) -- (502.5,270) ;
\draw    (482.5,260) -- (482.5,270) ;
\draw    (462.5,260) -- (462.5,270) ;
\draw    (442.5,260) -- (442.5,270) ;
\draw    (542.5,160) -- (542.5,170) ;
\draw    (522.5,160) -- (522.5,170) ;
\draw    (502.5,160) -- (502.5,170) ;
\draw    (482.5,160) -- (482.5,170) ;
\draw    (462.5,160) -- (462.5,170) ;
\draw    (442.5,160) -- (442.5,170) ;
\draw    (442.5,200) -- (442.5,210) ;
\draw    (502.5,200) -- (502.5,210) ;
\draw    (522.5,200) -- (522.5,210) ;
\draw    (542.5,200) -- (542.5,210) ;
\draw    (482.5,200) -- (482.5,210) ;
\draw    (462.5,200) -- (462.5,210) ;
\draw    (502.5,130) -- (502.5,140) ;
\draw    (522.5,130) -- (522.5,140) ;
\draw    (542.5,130) -- (542.5,140) ;
\draw    (482.5,130) -- (482.5,140) ;
\draw    (462.5,130) -- (462.5,140) ;
\draw (330,36.6) node [anchor=south] [inner sep=0.75pt]  [font=\small]  {$B$};
\draw (290,273.4) node [anchor=north] [inner sep=0.75pt]  [font=\small]  {$M_{f}$};
\draw (340,120) node  [font=\footnotesize]  {${\textstyle \mathsf{now}( f)}$};
\draw (380,120) node  [font=\footnotesize]  {${\textstyle \mathsf{now}( f')}$};
\draw (341.5,230) node  [font=\footnotesize]  {${\textstyle \mathsf{now}( g)}$};
\draw (381.5,230) node  [font=\footnotesize]  {${\textstyle \mathsf{now}( g')}$};
\draw (290,36.6) node [anchor=south] [inner sep=0.75pt]  [font=\small]  {$A$};
\draw (310,36.6) node [anchor=south] [inner sep=0.75pt]  [font=\small]  {$A'$};
\draw (350,36.6) node [anchor=south] [inner sep=0.75pt]  [font=\small]  {$B'$};
\draw (370,36.6) node [anchor=south] [inner sep=0.75pt]  [font=\small]  {$X$};
\draw (390,36.6) node [anchor=south] [inner sep=0.75pt]  [font=\small]  {$X'$};
\draw (310,273.4) node [anchor=north] [inner sep=0.75pt]  [font=\small]  {$M'_{f}$};
\draw (330,273.4) node [anchor=north] [inner sep=0.75pt]  [font=\small]  {$M_{g}$};
\draw (350,273.4) node [anchor=north] [inner sep=0.75pt]  [font=\small]  {$M'_{g}$};
\draw (370,273.4) node [anchor=north] [inner sep=0.75pt]  [font=\small]  {$Z$};
\draw (390,273.4) node [anchor=north] [inner sep=0.75pt]  [font=\small]  {$Z'$};
\draw (472.5,150) node  [font=\footnotesize]  {${\textstyle \mathsf{now}( f)}$};
\draw (532.5,150) node  [font=\footnotesize]  {${\textstyle \mathsf{now}( f')}$};
\draw (472.5,220) node  [font=\footnotesize]  {${\textstyle \mathsf{now}( g)}$};
\draw (532.5,220) node  [font=\footnotesize]  {${\textstyle \mathsf{now}( g')}$};
\draw (442.5,36.6) node [anchor=south] [inner sep=0.75pt]  [font=\small]  {$A$};
\draw (462.5,36.6) node [anchor=south] [inner sep=0.75pt]  [font=\small]  {$A'$};
\draw (482.5,36.6) node [anchor=south] [inner sep=0.75pt]  [font=\small]  {$B$};
\draw (502.5,36.6) node [anchor=south] [inner sep=0.75pt]  [font=\small]  {$B'$};
\draw (522.5,36.6) node [anchor=south] [inner sep=0.75pt]  [font=\small]  {$X$};
\draw (542.5,36.6) node [anchor=south] [inner sep=0.75pt]  [font=\small]  {$X'$};
\draw (442.5,273.4) node [anchor=north] [inner sep=0.75pt]  [font=\small]  {$M_{f}$};
\draw (462.5,273.4) node [anchor=north] [inner sep=0.75pt]  [font=\small]  {$M'_{f}$};
\draw (482.5,273.4) node [anchor=north] [inner sep=0.75pt]  [font=\small]  {$M_{g}$};
\draw (502.5,273.4) node [anchor=north] [inner sep=0.75pt]  [font=\small]  {$M'_{g}$};
\draw (522.5,273.4) node [anchor=north] [inner sep=0.75pt]  [font=\small]  {$Z$};
\draw (542.5,273.4) node [anchor=north] [inner sep=0.75pt]  [font=\small]  {$Z'$};
\draw (411,143.4) node [anchor=north west][inner sep=0.75pt]    {$=$};
\end{tikzpicture}
\caption{Functoriality of parallel composition.}
\label{strings:functorpar}
\end{figure}
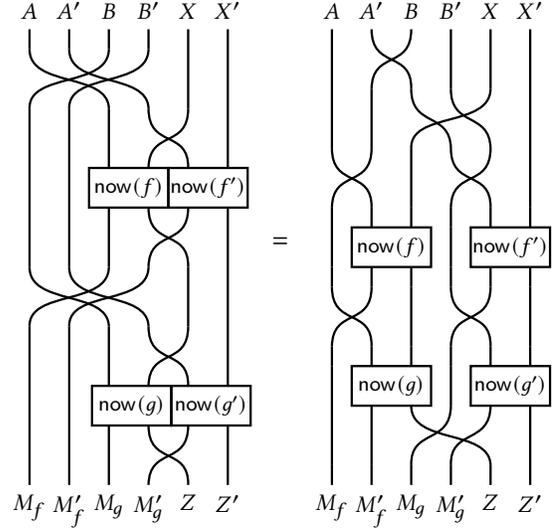

\begin{lemma}\label{lemma:parfun}
  Parallel composition of streams is functorial with respect to sequential composition of streams.
  Given four streams
  \begin{itemize}
    \item $f \in \STREAM(\stream{X},\stream{Y})$,
    \item $f' \in \STREAM(\stream{X}',\stream{Y}')$,
    \item $g \in \STREAM(\stream{Y},\stream{Z})$, and
    \item $g' \in \STREAM(\stream{Y},\stream{Z})$;
  \end{itemize}
  we claim that $(f \tensor f') ; (g \tensor g') = (f ; g) \tensor (f' ; g')$.
\end{lemma}
\begin{proof}
  Direct consequence of~\Cref{lemma:parfunassociativememories},
  after considering the appropriate coherence morphisms.
\end{proof}

\begin{theorem}[see~\cite{roman2020}]\label{th:monoidalstreamscategory}
  Monoidal streams over a \symmetricMonoidalCategory{}
  $(\catC, \tensor, \monoidalunit)$ form a symmetric monoidal category
  $\STREAM$.
\end{theorem}
\begin{proof}
  Sequential composition of streams (\Cref{def:sequentialstream})
  is associative (\Cref{lemma:sequentialcompositionassociative})
  and unital with respect to identities.
  Parallel composition is bifunctorial with respect to sequential composition (\Cref{lemma:parfun});
  this determines a bifunctor, which is the tensor of the monoidal category.
  The coherence morphisms and the symmetry can be included from sets, so they still satisfy the pentagon and triangle equations.
\end{proof}

\begin{lemma}\label{lemma:tighteninglong}
  The structure $(\STREAM, \fbk)$ with memories satisfies the tightening axiom (A1).
  Given three streams
  \begin{itemize}
    \item $u \in \STREAM(\act{A}{\stream{X}'}, \stream{X})$,
    \item $f \in \STREAM(\act{B \tensor T}{\delay\stream{S} \tensor \stream{X}}, \stream{S} \tensor \stream{Y})$, and
    \item $v \in \STREAM(\act{C}{\stream{Y}}, \stream{Y}')$;
  \end{itemize}
  we claim that
  \[\fbk^{S}(u^{A} ; f^{B} ; v^{C}) = \act{\sigma}{u^{A} ; \fbk^{S}(f^{B \tensor T}) ; v^{C}}.\]
\end{lemma}
\begin{proof}
  First, we note that both sides of the equation (which, from now on, we call $LHS$ and $RHS$, respectively)
  represent streams with different memories.
  \begin{itemize}
    \item $M(LHS) = A \tensor B \tensor C \tensor T$,
    \item $M(RHS) = A \tensor B \tensor T \tensor C$.
  \end{itemize}
  We will prove that they are related by dinaturality over the symmetry $\sigma$.
  We know that $\now(LHS) ; \sigma = \now(RHS)$ by string diagrams (see \Cref{strings:tightening}).
  Then, by coinduction, we know that $\later(LHS) = \act{\sigma}{\later(RHS)}$.
\end{proof}

\begin{figure}
\tikzset{every picture/.style={line width=0.85pt}} %
\begin{tikzpicture}[x=0.75pt,y=0.75pt,yscale=-1,xscale=1]
\draw    (140,60) -- (140,100) ;
\draw    (120,60) -- (120,100) ;
\draw    (200,60) -- (200.01,170) ;
\draw    (160,60) .. controls (160,80.43) and (180.33,70.43) .. (180,90) ;
\draw    (180,60) .. controls (180,80.43) and (160.33,70.43) .. (160,90) ;
\draw    (180.01,100) .. controls (179.84,119.77) and (119.51,111.43) .. (120.01,130) ;
\draw    (140.01,100) .. controls (139.84,119.77) and (159.5,111.43) .. (160,130) ;
\draw    (120.01,100) .. controls (119.85,119.77) and (139.51,111.43) .. (140.01,130) ;
\draw    (160.01,100) .. controls (159.84,119.77) and (179.5,111.43) .. (180,130) ;
\draw    (160.01,140) .. controls (159.84,159.77) and (139.51,151.43) .. (140.01,170) ;
\draw    (140.01,140) .. controls (139.85,159.77) and (179.51,151.43) .. (180.01,170) ;
\draw    (180.01,140) .. controls (179.84,159.77) and (159.51,151.43) .. (160.01,170) ;
\draw   (170.01,170) -- (210.01,170) -- (210.01,190) -- (170.01,190) -- cycle ;
\draw    (180.01,190) .. controls (179.84,209.77) and (139.51,201.43) .. (140.01,220) ;
\draw    (140.01,190) .. controls (139.84,209.77) and (159.51,201.43) .. (160.01,220) ;
\draw    (160.01,190) .. controls (159.84,209.77) and (179.51,201.43) .. (180.01,220) ;
\draw   (150.01,220) -- (210.01,220) -- (210.01,240) -- (150.01,240) -- cycle ;
\draw    (200.01,190) -- (200.01,220) ;
\draw    (120.01,240) .. controls (119.84,259.77) and (179.51,251.43) .. (180.01,270) ;
\draw   (170.01,270) -- (210.01,270) -- (210.01,290) -- (170.01,290) -- cycle ;
\draw    (200.01,240) -- (200.01,270) ;
\draw    (160.01,170) -- (160.01,190) ;
\draw    (140.01,170) -- (140.01,190) ;
\draw    (140.01,220) -- (140.01,240) ;
\draw    (120.01,140) -- (120.01,240) ;
\draw    (180.01,240) .. controls (179.84,259.77) and (159.51,251.43) .. (160.01,270) ;
\draw    (160.01,240) .. controls (159.84,259.77) and (139.51,251.43) .. (140.01,270) ;
\draw    (140.01,240) .. controls (139.84,259.77) and (119.51,251.43) .. (120.01,270) ;
\draw    (160.01,270) -- (160.01,310) ;
\draw    (140.01,270) -- (140.01,310) ;
\draw    (120.01,270) -- (120.01,310) ;
\draw    (180.01,290) -- (180.01,310) ;
\draw    (200.01,290) -- (200.01,310) ;
\draw    (320,60) -- (320,110) ;
\draw    (300.01,60) .. controls (300.02,80.43) and (260.34,70.43) .. (260,90) ;
\draw    (260,60) .. controls (260.01,80.43) and (280.34,70.43) .. (280,90) ;
\draw    (340,60) -- (340,130) ;
\draw    (280.01,60) .. controls (280.02,80.43) and (300.34,70.43) .. (300,90) ;
\draw    (300,90) .. controls (300,100.37) and (280,100.37) .. (280,110) ;
\draw    (260,90) -- (260,210) ;
\draw    (280,110) -- (280,170) ;
\draw   (310,130) -- (350,130) -- (350,150) -- (310,150) -- cycle ;
\draw    (300,130) -- (300,150) ;
\draw   (290,190) -- (350,190) -- (350,210) -- (290,210) -- cycle ;
\draw    (320,170) -- (320,190) ;
\draw    (340,150) -- (340,190) ;
\draw    (280,190) -- (280,210) ;
\draw    (300,210) .. controls (299.67,220.37) and (279.67,220.37) .. (280,230) ;
\draw    (320,210) -- (320,230) ;
\draw    (340,210) -- (340,250) ;
\draw    (280,230) -- (280,310) ;
\draw    (260,230) -- (260,310) ;
\draw   (310,250) -- (350,250) -- (350,270) -- (310,270) -- cycle ;
\draw    (300,250) -- (300,270) ;
\draw    (340,270) -- (340,310) ;
\draw    (280,90) .. controls (280,100.37) and (300,100.37) .. (300,110) ;
\draw    (320,110) .. controls (320,120.37) and (300,120.37) .. (300,130) ;
\draw    (300,110) .. controls (300,120.37) and (320,120.37) .. (320,130) ;
\draw    (320,150) .. controls (320,160.37) and (300,160.37) .. (300,170) ;
\draw    (300,150) .. controls (300,160.37) and (320,160.37) .. (320,170) ;
\draw    (300,170) .. controls (300,180.37) and (280,180.37) .. (280,190) ;
\draw    (280,170) .. controls (280,180.37) and (300,180.37) .. (300,190) ;
\draw    (320,230) .. controls (320,240.37) and (300,240.37) .. (300,250) ;
\draw    (300,230) .. controls (300,240.37) and (320,240.37) .. (320,250) ;
\draw    (320,270) .. controls (320,280.37) and (300,280.37) .. (300,290) ;
\draw    (300,270) .. controls (300,280.37) and (320,280.37) .. (320,290) ;
\draw    (320,290) .. controls (320,300.37) and (300,300.37) .. (300,310) ;
\draw    (300,290) .. controls (300,300.37) and (320,300.37) .. (320,310) ;
\draw    (280,210) .. controls (279.67,220.37) and (259.67,220.37) .. (260,230) ;
\draw    (260,210) .. controls (260.34,220.37) and (299.67,220.37) .. (300,230) ;
\draw    (160,90) -- (160.01,100) ;
\draw    (180,90) -- (180.01,100) ;
\draw    (180,130) -- (180.01,140) ;
\draw    (160,130) -- (160.01,140) ;
\draw    (140,130) -- (140.01,140) ;
\draw    (120,130) -- (120.01,140) ;
\draw (161,56.6) node [anchor=south] [inner sep=0.75pt]  [font=\small]  {$C$};
\draw (121,56.6) node [anchor=south] [inner sep=0.75pt]  [font=\small]  {$A$};
\draw (141,56.6) node [anchor=south] [inner sep=0.75pt]  [font=\small]  {$B$};
\draw (181,56.6) node [anchor=south] [inner sep=0.75pt]  [font=\small]  {$T$};
\draw (200,56.6) node [anchor=south] [inner sep=0.75pt]  [font=\small]  {$X'$};
\draw (190.01,180) node  [font=\footnotesize]  {${\textstyle \mathsf{now}( u)}$};
\draw (180.01,230) node  [font=\footnotesize]  {${\textstyle \mathsf{now}( f)}$};
\draw (190.01,280) node  [font=\footnotesize]  {${\textstyle \mathsf{now}( v)}$};
\draw (121.5,312.4) node [anchor=north] [inner sep=0.75pt]  [font=\small]  {$M_{u}$};
\draw (141.5,312.4) node [anchor=north] [inner sep=0.75pt]  [font=\small]  {$M_{f}$};
\draw (181.5,312.4) node [anchor=north] [inner sep=0.75pt]  [font=\small]  {$M_{v}$};
\draw (161.5,312.4) node [anchor=north] [inner sep=0.75pt]  [font=\small]  {$S_{0}$};
\draw (201.5,312.4) node [anchor=north] [inner sep=0.75pt]  [font=\small]  {$Y'$};
\draw (300.01,56.6) node [anchor=south] [inner sep=0.75pt]  [font=\small]  {$C$};
\draw (260,56.6) node [anchor=south] [inner sep=0.75pt]  [font=\small]  {$A$};
\draw (280.01,56.6) node [anchor=south] [inner sep=0.75pt]  [font=\small]  {$B$};
\draw (320,56.6) node [anchor=south] [inner sep=0.75pt]  [font=\small]  {$T$};
\draw (340,56.6) node [anchor=south] [inner sep=0.75pt]  [font=\small]  {$X'$};
\draw (330,140) node  [font=\footnotesize]  {${\textstyle \mathsf{now}( u)}$};
\draw (320,200) node  [font=\footnotesize]  {${\textstyle \mathsf{now}( f)}$};
\draw (330,260) node  [font=\footnotesize]  {${\textstyle \mathsf{now}( v)}$};
\draw (260,313.4) node [anchor=north] [inner sep=0.75pt]  [font=\small]  {$M_{u}$};
\draw (280,313.4) node [anchor=north] [inner sep=0.75pt]  [font=\small]  {$M_{f}$};
\draw (300,313.4) node [anchor=north] [inner sep=0.75pt]  [font=\small]  {$S_{0}$};
\draw (320,313.4) node [anchor=north] [inner sep=0.75pt]  [font=\small]  {$M_{v}$};
\draw (340,313.4) node [anchor=north] [inner sep=0.75pt]  [font=\small]  {$Y'$};
\draw (221,172.4) node [anchor=north west][inner sep=0.75pt]    {$=$};
\end{tikzpicture}
\caption{The tightening axiom (A1).}
\label{strings:tightening}
\end{figure}
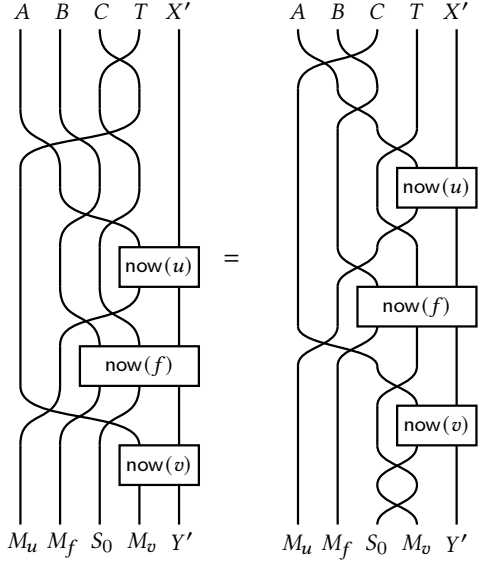

\begin{lemma}\label{lemma:tightening}
  The structure $(\STREAM, \fbk)$ satisfies the tightening axiom (A1).
  Given streams
  \begin{itemize}
     \item $u \in \STREAM(\stream{X}',\stream{X})$,
     \item $f \in \STREAM(\delay\stream{S}\tensor \stream{X},\stream{S}\tensor \stream{Y})$, and
     \item $v \in \STREAM(\stream{Y},\stream{Y}')$;
  \end{itemize}
  we claim that $\fbk^{S}(u ; f ; v) = u ; \fbk^{S}(f) ; v$.
\end{lemma}
\begin{proof}
  Consequence of \Cref{lemma:tighteninglong}, after applying the necessary coherence morphisms.
\end{proof}

\begin{lemma}\label{lemma:vanishinglong}
  The structure $(\STREAM, \fbk)$ with memories satisfies the vanishing axiom (A2).
  Given a stream
  \begin{itemize}
    \item $f \in \STREAM(\act{A}{\delay\stream{S} \tensor \stream{X}}, \stream{S} \tensor \stream{Y})$,
  \end{itemize}
  we claim that $\fbk^{I}(f^{A}) = \act{\rho}{f}$.
\end{lemma}
\begin{proof}
  First, we note that both sides of the equation represent streams with different memories, $M(\fbk^{I}(f^{A})) = M(f) \tensor I$.
  We will prove that they are related by dinaturality over the right unitor $\rho$.
  We know that $\now(\fbk^{I}(f^{A})) = \now(f)$ by definition.
  Then, by coinduction, we konw that $\later(\fbk^{I}(f^{A})) = \act{\rho}{\later(f)}$.
\end{proof}

\begin{lemma}\label{lemma:vanishing}
  The structure $(\STREAM, \fbk)$ satisfies the vanishing axiom (A2).
  Given a stream $f \in \STREAM(\delay\stream{S} \tensor \stream{X}, \stream{S} \tensor \stream{Y})$,
  we claim that $\fbk^{I}(f) = f$.
\end{lemma}
\begin{proof}
  Consequence of \Cref{lemma:vanishinglong}, after applying the necessary coherence morphisms.
\end{proof}

\begin{lemma}\label{lemma:joininglong}
  The structure $(\STREAM, \fbk)$ with memories satisfies the joining axiom (A3).
  Given a stream
  \begin{itemize}
    \item $f \in \STREAM(\act{(A \tensor P \tensor Q)}{\delay \stream{S} \tensor \stream{X}}, \stream{S} \tensor \stream{Y})$,
  \end{itemize}
  we claim that
  \[\fbk^{S \tensor T}(f^{A \tensor (P \tensor Q)}) = \act{\alpha}{\fbk^{T}(\act{\sigma}{\fbk^{S}(\act{\sigma}{f^{(A \tensor Q) \tensor P}})})}.\]
\end{lemma}
\begin{proof}
  First, we note that both sides of the equation (which, from now on, we call $LHS$ and $RHS$, respectively)
  represent strams with different memories.
  \begin{itemize}
    \item $M(LHS) = M(f) \tensor (S_{0} \tensor T_{0})$,
    \item $M(RHS) = (M(f) \tensor S_{0}) \tensor T_{0}$.
  \end{itemize}
  We will prove that they are related by dinaturality over the associator $\alpha$.
  We know that $\now(LHS) ; \alpha = \now(RHS)$ by definition.
  Then, by coinduction, we know that $\later(LHS) = \act{\alpha}{\later(RHS)}$.
\end{proof}
\begin{lemma}\label{lemma:joining}
  The structure $(\STREAM, \fbk)$ satisfies the joining axiom (A3).
  Given a stream
  \begin{itemize}
    \item $f \in \STREAM(\delay \stream{S} \tensor \stream{X}, \stream{S} \tensor \stream{Y})$,
  \end{itemize}
  we claim that $\fbk^{S \tensor T}(f) = \fbk^{T}(\fbk^{S}(f))$.
\end{lemma}
\begin{proof}
  Consequence of \Cref{lemma:joininglong}, after applying the necessary coherence morphisms.
\end{proof}

\begin{lemma}\label{lemma:strengthlong}
  The structure $(\STREAM, \fbk)$ with memories satisfies the strength axiom (A4).
  Given two streams
  \begin{itemize}
    \item $f \in \STREAM(\act{(A \tensor P)}{\delay \stream{S} \tensor \stream{X}}, \stream{S} \tensor \stream{Y})$, and
    \item $g \in \STREAM(\act{B}{\stream{X}'},\stream{Y}')$,
  \end{itemize}
  we claim that
  \[\act{\alpha}{\fbk^{S}(f^{A} \tensor g^{B})} = \fbk^{S}(f^{A \tensor P})^{A \tensor P} \tensor g^{B}.\]
\end{lemma}
\begin{proof}
  First, we note that both sides of the equation (which, from now on, we call $LHS$ and $RHS$, respectively)
  represent streams with different memories.
  \begin{itemize}
     \item $M(LHS) = (M(f) \tensor S_{0}) \tensor M(g)$,
     \item $M(RHS) = M(f) \tensor (S_{0} \tensor M(g))$.
  \end{itemize}
  We will prove that they are related by dinaturality over the symmetry $\alpha$.
  We know that $\now(LHS) = \now(RHS) ; \alpha$ by string diagrams (see \Cref{strings:strength}).
  Then, by coinduction, we know that $\act{\alpha}{\later(LHS)} = \later(RHS)$.
\end{proof}

\begin{figure}
\tikzset{every picture/.style={line width=0.85pt}} %
\begin{tikzpicture}[x=0.75pt,y=0.75pt,yscale=-1,xscale=1]
\draw    (180,60) .. controls (179.83,79.77) and (159.5,71.43) .. (160,90) ;
\draw    (140,60) .. controls (139.84,79.77) and (179.5,71.43) .. (180,90) ;
\draw    (160,60) .. controls (159.84,79.77) and (139.5,71.43) .. (140,90) ;
\draw    (120,60) -- (120,100) ;
\draw    (200,60) -- (200,90) ;
\draw    (140,90) -- (140,100) ;
\draw    (160,90) -- (160,100) ;
\draw   (110,100) -- (170,100) -- (170,120) -- (110,120) -- cycle ;
\draw    (180,90) -- (180,100) ;
\draw    (140,120) -- (140,130) ;
\draw    (160,120) -- (160,130) ;
\draw   (170,100) -- (210,100) -- (210,120) -- (170,120) -- cycle ;
\draw    (200,90) -- (200,100) ;
\draw    (140,120) -- (140,130) ;
\draw    (160,120) -- (160,130) ;
\draw    (180,120) -- (180,130) ;
\draw    (200,120) -- (200,160) ;
\draw    (180,130) .. controls (179.83,149.77) and (139.5,141.43) .. (140,160) ;
\draw    (140,130) .. controls (139.83,149.77) and (159.51,141.43) .. (160.01,160) ;
\draw    (160,130) .. controls (159.83,149.77) and (179.51,141.43) .. (180.01,160) ;
\draw    (120,120) -- (120,160) ;
\draw    (310,60) .. controls (310,70.37) and (290,70.37) .. (290,80) ;
\draw    (290,60) .. controls (290,70.37) and (310,70.37) .. (310,80) ;
\draw    (330,80) .. controls (330,90.37) and (310,90.37) .. (310,100) ;
\draw    (310,80) .. controls (310,90.37) and (330,90.37) .. (330,100) ;
\draw   (259,100) -- (319,100) -- (319,120) -- (259,120) -- cycle ;
\draw   (319,100) -- (359,100) -- (359,120) -- (319,120) -- cycle ;
\draw    (330,120) .. controls (330,130.37) and (310,130.37) .. (310,140) ;
\draw    (310,120) .. controls (310,130.37) and (330,130.37) .. (330,140) ;
\draw    (310,140) .. controls (310,150.37) and (290,150.37) .. (290,160) ;
\draw    (290,140) .. controls (290,150.37) and (310,150.37) .. (310,160) ;
\draw    (290,80) -- (290,100) ;
\draw    (290,120) -- (290,140) ;
\draw    (330,140) -- (330,160) ;
\draw    (330,60) -- (330,80) ;
\draw    (350,60) -- (350,100) ;
\draw    (350,120) -- (350,160) ;
\draw    (270,60) -- (270,100) ;
\draw    (270,120) -- (270,160) ;
\draw (160,56.6) node [anchor=south] [inner sep=0.75pt]  [font=\small]  {$P$};
\draw (120.01,56.6) node [anchor=south] [inner sep=0.75pt]  [font=\small]  {$A$};
\draw (140,56.6) node [anchor=south] [inner sep=0.75pt]  [font=\small]  {$B$};
\draw (180,56.6) node [anchor=south] [inner sep=0.75pt]  [font=\small]  {$X$};
\draw (200,56.6) node [anchor=south] [inner sep=0.75pt]  [font=\small]  {$X'$};
\draw (221,103.4) node [anchor=north west][inner sep=0.75pt]    {$=$};
\draw (140,110) node  [font=\footnotesize]  {${\textstyle \mathsf{now}( f)}$};
\draw (190,110) node  [font=\footnotesize]  {${\textstyle \mathsf{now}( g)}$};
\draw (289,110) node  [font=\footnotesize]  {${\textstyle \mathsf{now}( f)}$};
\draw (339,110) node  [font=\footnotesize]  {${\textstyle \mathsf{now}( g)}$};
\draw (270,56.6) node [anchor=south] [inner sep=0.75pt]  [font=\small]  {$A$};
\draw (290,56.6) node [anchor=south] [inner sep=0.75pt]  [font=\small]  {$B$};
\draw (310,56.6) node [anchor=south] [inner sep=0.75pt]  [font=\small]  {$P$};
\draw (330,56.6) node [anchor=south] [inner sep=0.75pt]  [font=\small]  {$X$};
\draw (350,56.6) node [anchor=south] [inner sep=0.75pt]  [font=\small]  {$X'$};
\draw (120,163.4) node [anchor=north] [inner sep=0.75pt]  [font=\small]  {$M_{f}$};
\draw (140,163.4) node [anchor=north] [inner sep=0.75pt]  [font=\small]  {$M_{g}$};
\draw (160.01,163.4) node [anchor=north] [inner sep=0.75pt]  [font=\small]  {$S_{0}$};
\draw (180.01,163.4) node [anchor=north] [inner sep=0.75pt]  [font=\small]  {$Y$};
\draw (200,163.4) node [anchor=north] [inner sep=0.75pt]  [font=\small]  {$Y'$};
\draw (270,163.4) node [anchor=north] [inner sep=0.75pt]  [font=\small]  {$M_{f}$};
\draw (290,163.4) node [anchor=north] [inner sep=0.75pt]  [font=\small]  {$M_{g}$};
\draw (310,163.4) node [anchor=north] [inner sep=0.75pt]  [font=\small]  {$S_{0}$};
\draw (330,163.4) node [anchor=north] [inner sep=0.75pt]  [font=\small]  {$Y$};
\draw (350,163.4) node [anchor=north] [inner sep=0.75pt]  [font=\small]  {$Y'$};
\end{tikzpicture}
\caption{The strength axiom (A4).}
\label{strings:strength}
\end{figure}
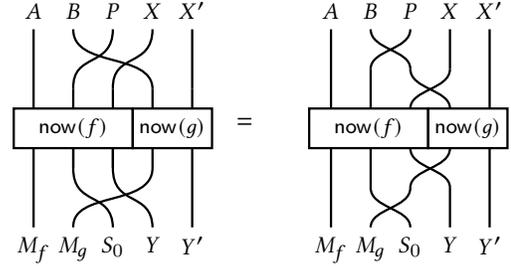

\begin{lemma}\label{lemma:strength}
  The structure $(\STREAM, \fbk)$ with memories satisfies the strength axiom (A4).
  Given two streams
  \begin{itemize}
    \item $f \in \STREAM(\stream{S} \tensor \stream{X}, \stream{S} \tensor \stream{Y})$, and
    \item $g \in \STREAM(\stream{X}',\stream{Y}')$,
  \end{itemize}
  we claim that
  \[\act{\alpha}{\fbk^{S}(f \tensor g)} = \fbk^{S}(f) \tensor g.\]
\end{lemma}
\begin{proof}
  Consequence of \Cref{lemma:strengthlong}, after applying the necessary coherence morphisms.
\end{proof}

\begin{lemma}\label{lemma:slidinglong}
  The structure $(\STREAM, \fbk)$ with memories satisfies the sliding axiom (A5).
  Given two streams
  \begin{itemize}
    \item $f \in \STREAM(\act{A}{\delay \stream{S} \tensor \stream{X}}, \stream{T} \tensor \stream{Y})$, and
    \item $r \in \STREAM(\act{C}{\stream{T}}, \stream{S})$
  \end{itemize}
  we claim that, for each $k \colon B \tensor Q \to C \tensor P$,
  \[\act{k}{\fbk^{\stream{S}}(f^{A \tensor P} ; (r^{C} \tensor \im))} = \fbk^{\stream{T}}(\act{k}{(\delay r^{C} \tensor \im) ; f^{A \tensor P}}).\]
\end{lemma}
\begin{proof}
  First, we note that both sides of the equation (which, from now on, we call $LHS$ and $RHS$, respectively)
  represent streams with different memories.
  \begin{itemize}
    \item $M(LHS) = M(f) \tensor M(r)  \tensor S_{0}$,
    \item $M(RHS) = M(r) \tensor M(f) \tensor T_{0}$.
  \end{itemize}
  We will prove that they are related by dinaturality over the symmetry and the first action of $r$, that is, $\sigma ; (\now(r) \tensor \im)$.
  We know that $\now(LHS)= \now(RHS) ; (\sigma ; (\im \tensor \now(r)))$ by string diagrams (see \Cref{strings:sliding}).
  Using coinduction,
  \[\begin{aligned}
      & \act{(\sigma ; \tid{\now(r)})}{\later(LHS)}  \\
      = & \act{(\sigma ; \tid{\now(r)})}{\fbk^{\stream{S}}(\later(f)^{M(f) \tensor S_{0}} ; (\later(r)^{M(r)} \tensor \im))} \\
      = & \fbk^{\stream{T}}(\act{(\sigma ; \tid{\now(r)})}{(\delay \later(r)^{M(r)} \tensor \im) ; \later(f)^{M(f) \tensor S_{0}}}) \\
      = & \fbk^{\stream{T}}(\act{\sigma}{(\later(\delay r)^{M(r)} \tensor \im) ; \later(f)^{M(f) \tensor S_{0}}}) \\
      = & \later(RHS),
    \end{aligned}\]
  we show that $\act{(\sigma ; \tid{\now(r)})}{\later(LHS)} = \later(RHS)$.
\end{proof}

\begin{figure}
\tikzset{every picture/.style={line width=0.85pt}} %
\begin{tikzpicture}[x=0.75pt,y=0.75pt,yscale=-1,xscale=1]
\draw    (80,40) -- (80,90) ;
\draw    (100,40) -- (100,60) ;
\draw    (120,40) -- (120,60) ;
\draw    (140,40) -- (140,130) ;
\draw   (90,60) -- (130,60) -- (130,80) -- (90,80) -- cycle ;
\draw    (80,90) .. controls (80,110.43) and (100.33,100.43) .. (100,120) ;
\draw    (100,90) .. controls (100,110.43) and (80.33,100.43) .. (80,120) ;
\draw    (120,80) -- (120,130) ;
\draw   (90,130) -- (150,130) -- (150,150) -- (90,150) -- cycle ;
\draw    (80,150) .. controls (80,170.43) and (100.33,160.43) .. (100,180) ;
\draw    (100,150) .. controls (100,170.43) and (80.33,160.43) .. (80,180) ;
\draw    (80,130) -- (80,150) ;
\draw   (90,190) -- (130,190) -- (130,210) -- (90,210) -- cycle ;
\draw    (80,190) -- (80,230) ;
\draw    (140,150) -- (140,230) ;
\draw    (120,210) -- (120,230) ;
\draw    (100,210) -- (100,230) ;
\draw    (190,40) -- (190,80) ;
\draw    (210,40) -- (210,60) ;
\draw    (230,40) -- (230,60) ;
\draw    (250,40) -- (250,150) ;
\draw   (200,60) -- (240,60) -- (240,80) -- (200,80) -- cycle ;
\draw    (230,80) -- (230,100) ;
\draw   (200,150) -- (260,150) -- (260,170) -- (200,170) -- cycle ;
\draw    (190,150) -- (190,170) ;
\draw   (200,190) -- (240,190) -- (240,210) -- (200,210) -- cycle ;
\draw    (230,170) -- (230,190) ;
\draw    (190,190) -- (190,230) ;
\draw    (250,170) -- (250,230) ;
\draw    (230,210) -- (230,230) ;
\draw    (210,210) -- (210,230) ;
\draw    (210,80) .. controls (210,90.37) and (190,90.37) .. (190,100) ;
\draw    (190,80) .. controls (190,90.37) and (210,90.37) .. (210,100) ;
\draw    (190,100) .. controls (190,110.37) and (230,110.37) .. (230,120) ;
\draw    (210,100) .. controls (210,110.37) and (190,110.37) .. (190,120) ;
\draw    (230,100) .. controls (230,110.37) and (210,110.37) .. (210,120) ;
\draw    (190,120) -- (190,130) ;
\draw    (210,120) -- (210,130) ;
\draw    (230,120) -- (230,130) ;
\draw    (230,130) .. controls (230,140.37) and (190,140.37) .. (190,150) ;
\draw    (190,130) .. controls (190,140.37) and (210,140.37) .. (210,150) ;
\draw    (210,130) .. controls (210,140.37) and (230,140.37) .. (230,150) ;
\draw    (210,170) .. controls (210,180.37) and (190,180.37) .. (190,190) ;
\draw    (190,170) .. controls (190,180.37) and (210,180.37) .. (210,190) ;
\draw    (100,80) -- (100,90) ;
\draw    (100,120) -- (100,130) ;
\draw    (80,120) -- (80,130) ;
\draw    (100,180) -- (100,190) ;
\draw    (80,180) -- (80,190) ;
\draw    (120,150) -- (120,190) ;
\draw (80,233.4) node [anchor=north] [inner sep=0.75pt]  [font=\small]  {$M_{f}$};
\draw (100,233.4) node [anchor=north] [inner sep=0.75pt]  [font=\small]  {$M_{r}$};
\draw (80,36.6) node [anchor=south] [inner sep=0.75pt]  [font=\small]  {$A$};
\draw (100,36.6) node [anchor=south] [inner sep=0.75pt]  [font=\small]  {$B$};
\draw (120,36.6) node [anchor=south] [inner sep=0.75pt]  [font=\small]  {$Q$};
\draw (140,36.6) node [anchor=south] [inner sep=0.75pt]  [font=\small]  {$X$};
\draw (110,70) node  [font=\footnotesize]  {$k$};
\draw (110,200) node  [font=\footnotesize]  {${\textstyle \mathsf{now}( r)}$};
\draw (120,140) node  [font=\footnotesize]  {${\textstyle \mathsf{now}( f)}$};
\draw (220,70) node  [font=\footnotesize]  {$k$};
\draw (220,200) node  [font=\footnotesize]  {${\textstyle \mathsf{now}( r)}$};
\draw (230,160) node  [font=\footnotesize]  {${\textstyle \mathsf{now}( f)}$};
\draw (156,133.4) node [anchor=north west][inner sep=0.75pt]    {$=$};
\draw (120,233.4) node [anchor=north] [inner sep=0.75pt]  [font=\small]  {$S_{0}$};
\draw (140,233.4) node [anchor=north] [inner sep=0.75pt]  [font=\small]  {$Y$};
\draw (190,233.4) node [anchor=north] [inner sep=0.75pt]  [font=\small]  {$M_{f}$};
\draw (210,233.4) node [anchor=north] [inner sep=0.75pt]  [font=\small]  {$M_{r}$};
\draw (230,233.4) node [anchor=north] [inner sep=0.75pt]  [font=\small]  {$S_{0}$};
\draw (250,233.4) node [anchor=north] [inner sep=0.75pt]  [font=\small]  {$Y$};
\draw (190,36.6) node [anchor=south] [inner sep=0.75pt]  [font=\small]  {$A$};
\draw (210,36.6) node [anchor=south] [inner sep=0.75pt]  [font=\small]  {$B$};
\draw (230,36.6) node [anchor=south] [inner sep=0.75pt]  [font=\small]  {$Q$};
\draw (250,36.6) node [anchor=south] [inner sep=0.75pt]  [font=\small]  {$X$};
\end{tikzpicture}
\caption{The sliding axiom (A5).}
\label{strings:sliding}
\end{figure}

\begin{lemma}\label{lemma:sliding}
  The structure $(\STREAM, \fbk)$ satisfies the sliding axiom (A5).
  Given two streams
  \begin{itemize}
    \item $f \in \STREAM(\delay \stream{S} \tensor \stream{X}, \stream{T} \tensor \stream{Y})$, and
    \item $r \in \STREAM(\stream{T}, \stream{S})$
  \end{itemize}
  we claim that, $\fbk^{\stream{S}}(f ; (r \tensor \im)) = \fbk^{\stream{T}}((\delay r \tensor \im) ; f)$.
\end{lemma}
\begin{proof}
  Consequence of \Cref{lemma:slidinglong}, after applying the necessary coherence morphisms.
\end{proof}

\begin{theorem}[From \Cref{th:monoidalstreamsfeedback}]\label{th:appendix:monoidalstreamsfeedback}
  \MonoidalStreams{} over a \symmetricMonoidalCategory{}
  $(\catC, \tensor, \monoidalunit)$ form a \(\delay\)-feedback monoidal category
  $(\STREAM, \fbk)$.
\end{theorem}
\begin{proof}
  We have proven that it satisfies all the feedback axioms:
  \begin{itemize}
    \item the tightening axiom by \Cref{lemma:tightening},
    \item the vanishing axiom by \Cref{lemma:vanishing},
    \item the joining axiom by \Cref{lemma:joining},
    \item the strength axiom by \Cref{lemma:strength},
    \item and the sliding axiom by \Cref{lemma:sliding}.
  \end{itemize}
  Thus, it is a feedback structure
\end{proof}
\section{Type theory}\label{appendix:section:typetheory}

As our base type theory, we consider the that of symmetric monoidal categories over some generators forming a multigraph $\mathcal{G}$, as described by Shulman~\cite{shulman2016categorical}.
Our only change will be to consider an \emph{unbiased presentation}, meaning that we consider n-ary tensor products instead of only binary and 0-ary (the monoidal unit).
This reduces the number of rules: the binary case and the 0-ary case of the usual presentation are taken care of by a single n-ary case.

\begin{figure}
\begin{mathpar}
  \infer[\textsc{Gen}]{f \in \mathcal{G}(A_1,\dots,A_n;B) \and
    \Gamma_1 \entails x_1 : A_1 \dots \Gamma_n \entails x_n : A_n}
    {\Shuf(\Gamma_1,\dots,\Gamma_n;\Gamma) \entails f(x_1,\dots,x_n) : B} \\
  \infer[\textsc{Pair}]{\Gamma_1 \entails x_1 : A_1 \and \dots \and \Gamma_n \entails x_n : A_n}{\Shuf(\Gamma,\Delta) \entails [x_1,\dots,x_n] : A_1 \otimes \dots \otimes A_n}  \and \infer[\textsc{Var}]{ }{x : A\entails x : A} \\
  \infer[\textsc{Split}]{\Delta \entails m : A_1 \otimes \dots \otimes A_n \and \Gamma, x_1 : A_1, \dots , x_n : A_n \entails z : C}{ \Shuf(\Gamma_1,\dots,\Gamma_n) \entails \textsc{Split}\ m \to [x_1,\dots,x_n]\ \textsc{in}\ z : C} \\
\end{mathpar}
\caption{Type theory of symm. monoidal categories~\cite{shulman2016categorical}.}
\end{figure}

\begin{definition}[Shuffling contexts~\cite{shulman2016categorical}]
A \emph{shuffle of contexts}
$\defining{linkshuffle}{\ensuremath{\operatorname{Shuf}}}(\Gamma_{1},\dots,\Gamma_{n})$
is the result of a permutation of $\bigsqcup^{n}_{i=0} \Gamma_{n}$ that leaves invariant the
internal order of each $\Gamma_{i}$. Shuffles allow us to derive morphisms that
make use of the symmetry without introducing redundancy in our type theory.
\end{definition}

Substitution is admissible in the type theory for symmetric monoidal categories~\cite{shulman2016categorical}.
It can be inductively defined: we write $P[m/x]$ for the substitution of the variable $x$ by a term $m$ inside the term $p$.  Using substitution, we can state a pair of $\beta/\eta$-reduction equalities for the terms of our type theory.

\begin{itemize}
  \item $\SPLIT{[E_1,\dots,E_{n}]}{x_1,\dots,x_n}{z} \equiv z[E_{1}/x_{1},\dots]$
  \item $\SPLIT{[E_{1},\dots,E_{n}]}{x_{1},\dots,x_{n}}{N[[x_{1},\dots,x_{n}]/u]}$ $\equiv N[[E_{1},\dots,E_{n}]/u]$
\end{itemize}

\subsection{Type theory for a strong monoidal endofunctor}

\begin{definition}[Signature]
  Let $T$ be a set of basic types. We write
  \[\freeendo{T} \coloneqq \{\partial^{n_{1}}t_{1}\partial^{n_{2}}t_{2}\dots\partial^{n_{k}}t_{k} \mid n_{1},\dots,n_{k} \in \mathbb{N}, t_{i} \in T\}\]
  for the free monoid-with-an-endomorphism $\partial : T^{\star}_{\partial} \to T^{\star}_{\partial}$ over $T$.
  The \textbf{signature} for a type theory of monoidal categories with a monoidal endofunctor is given by a pair of
  functions $s,t\colon \mathcal{O} \rightrightarrows \freeendo{T}$, assigning source and target to every generator from a set $\mathcal{O}$.
\end{definition}

\begin{example}
  We will usually include two families of generators in our theory. For each type $t_{0} \in T$ that
  can be copied, we have a $\mathsf{copy} \in \mathcal{O}$ generator with $s(\mathsf{copy}) = t_{0} \in \freeendo{T}$
  and $t(\mathsf{copy}) = t_{0} \cdot t_{0} \in \freeendo{T}$.
\end{example}

\begin{definition}
  The type theory of symmetric monoidal categories with a symmetric monoidal endofunctor over a signature
  $\mathcal{O} \rightrightarrows \freeendo{T}$ is the type theory of symmetric
  monoidal categories extended with an operator $\partial$ on types such that
  \[\partial [] = [] \quad\mbox{and}\quad \partial (\Gamma, x {:} A) = \partial \Gamma , x {:} \partial A,\]
  and the following $\Delay$ introduction rule.
  \begin{mathpar}
    \infer[\defining{linkDelay}{\textsc{Delay}}]{\Gamma \entails x : A}{\partial \Gamma \entails \Delay(x) : \partial A}
  \end{mathpar}
  The delay operator is a monoid homomorphism on types, satisfying
  $\partial I \equiv I$ and
  $\partial (A \otimes B) \equiv \partial A \otimes \partial B$. The \(\Delay\) introduction rule satisfies the following conversion equalities, that state that it acts as a functor preserving the monoidal structure.
  \begin{itemize}
    \item $\Delay(x) \equiv x$ for \(x\) a variable,
    \item $[\Delay(e_{1}), \dots, \Delay(e_{n})] \equiv \Delay([e_{1},\dots,e_{n}])$,
    \item $\begin{aligned}[t]\SPLIT{\Delay(m)}{e_{1},\dots,e_{n}}{\Delay(z)} \\ \equiv \Delay(\SPLIT{m}{e_{1},\dots,e_{n}}{z}),\end{aligned}$
  \end{itemize}
  We choose not to explicitly state the delay rule when writing terms of the type theory (as we do in \Cref{sec:addingfeedback}), as it does not cause ambiguity with the typing. However, we do write it when typechecking.
\end{definition}

\subsection{Type Theory for Delayed Feedback}
We finally augment the type theory of symmetric monoidal categories with a
monoidal endofunctor by adding the following derivation rule.
\begin{mathpar}
  \infer[\defining{linkFbk}{\textsc{Fbk}}]
    {\Gamma , s : \partial S \entails x(s) : S \tensor A}
    {\Gamma \entails \mbox{\textsc{Fbk}}\ s.\ x(s) : A}
\end{mathpar}
Following the axioms of categories with delayed feedback and their normalization
theorem, we introduce rules that follow from the feedback axioms.
Note that these rules simplify multiple applications of feedback into a single application at the head of the term.
\begin{enumerate}
\item $g(\Fbk{s}{x(s)})\ \equiv$ \\
  $\Fbk{s}{\LetIn{[t,b]}{x(s)}{[t,g(b)]}}$
\item
  $[\Fbk{s}{x(s)},\Fbk{t}{y(t)}]\ \equiv\ \Fbk{m}{}{}$ \\
  $\SPLIT{m}{s,t}{}$ \\
  $\SPLIT{x(s)}{u,v}{}$ \\
  $\SPLIT{y(t)}{u',v'}{}$ \\
  $[[u,u'],[v,v']]$
\item
  $\SPLIT{m}{x_{1},\dots,x_{n}}{\Fbk{s}{z}}\ \equiv\ $ \\
  $\Fbk{s}{\SPLIT{m}{x_{1},\dots,x_{n}}{z}}$
\item
  $\SPLIT{(\Fbk{u}{z(u)})}{x_{1},\dots,x_{n}}{m}\ \equiv$ \\
  $\Fbk{u}{\SPLIT{z(u)}{v,n}{}} \\ \SPLIT{n}{x_{1},\dots,x_{n}}{[v,m]}$
\item
  $x\ \equiv\ \Fbk{i}{\SPLIT{i}{}{[t,x]}}$
\item
  $\Fbk{x}{\Fbk{y}{m(x,y)}}\ \equiv\ $\\
  $\Fbk{n}{\SPLIT{n}{x,y}{m(x,y)}}$
\end{enumerate}

Finally, we introduce an equality representing the \emph{sliding axiom} (\Cref{diagram:sliding}).
Having all of these rules together means that we can always rewrite a term in the type theory with feedback as a term of the type theory of symmetric monoidal categories \emph{up to sliding}.
These are precisely the morphisms of the $\St(\bullet)$-construction, the free category with feedback.

\begin{enumerate}
  \item $\Fbk{s}{f(\Delay(h(s)),x)} \equiv$ \\
  $\Fbk{t}{\SPLIT{f(x,t)}{y,s}{[y,h(s)]}}$
\end{enumerate}

\subsection{Categories with copy and syntax sugar}

\begin{definition}\defining{linkwithcopy}
  A \emph{category with copying} is a \symmetricMonoidalCategory{} $(\catC,\otimes,I)$ in which
  every object $A \in \catC$ has a (non-necessarily natural) coassociative and cocommutative comultiplication
  $\delta_{A} = (\blackComonoid)_{A} \colon A \to A \otimes A$, called the ``copy''.
\end{definition}
  Every cartesian category and every kleisli category of a $\Set$-based commutative monad is a category with copying.
In our type theory, this is translated into a $\Copy$ generator acting as follows.
\begin{mathpar}
  \infer[\defining{linkCopy}{\textsc{Copy}}]{\Gamma \entails x : A}{\Gamma \entails \Copy(x) : A \otimes A}
\end{mathpar}

\begin{definition}
  A \emph{category with $\partial$-merging} is a symmetric monoidal category $(\catC,\otimes,I)$ with a symmetric monoidal endofunctor \(\partial \colon \catC \to \catC\) in which every object $A \in \catC$ has an associated morphism $\phi_{A} \colon A \otimes \partial A \to A$.
\end{definition}

In our type theory, this is translated by a $\Fby$ generator acting as follows.

\begin{mathpar}
  \infer[\defining{linkFby}{\defining{linkFby}{\textsc{Fby}}}]{
    \Gamma \entails x : A \and
    \Delta \entails y : \partial(A)}{\Shuf(\Gamma,\Delta) \entails x\ \Fby\ y : A}
\end{mathpar}

We allow three pieces of syntax sugar in our language, suited only for the case of categories with copying.
These make the language more \textsc{Lucid}-like without changing its formal description.

\begin{enumerate}
  \item We allow multiple occurences of a variable, implicily copying it.
  \item We apply \textsc{Delay} rules where needed for type-checking, without
  explicitly writing the rule.
  \item Recursive definitions are syntax for the $\textsc{Fbk}$ rule and the $\textsc{Copy}$
  rule. That is,
  \[M = x(M) \quad\mbox{means}\quad M = \Fbk{m}{\Copy(x(m))}.\]
  \item We use $\Wait$ to declare an implicit feedback loop.
  \[\defining{linkWait}{\ensuremath{\mbox{\textsc{Wait}}}}(x) \quad\mbox{means}\quad \Fbk{y}{[x,y]}.\]
\end{enumerate}
\section{Implementation}\label{section:implementation}

We use the \textbf{Haskell}~\cite{haskellreport} programming language for computations.
We use \hask{Arrows}~\cite{hughes00} to represent monoidal categories with an identity-on-objects monoidal functor from our base category of Haskell types and functions.
Notations for arrows~\cite{paterson01} have been explained in terms of Freyd categories~\cite{power99:freyd}.
In particular, the \hask{loop} notation is closely related to feedback, as it is usually employed to capture \emph{traces}.

Our definition of monoidal streams follows~\Cref{def:monoidalstream}.
\begin{haskellcode}
type Stream c = StreamWithMemory c ()

data StreamWithMemory c n x y where
  StreamWithMemory :: (Arrow c) =>
     c (n , x) (m , y)
    -> StreamWithMemory c m x y
    -> StreamWithMemory c n x y
\end{haskellcode}
Their sequential and parallel composition (\hask{comp} and \hask{tensor}) follow from~\Cref{def:sequentialstream,def:parallelstream}. In \Cref{section:code} we describe both first on the $\now$ part; and then trivially extended by coinduction.
\[\now(f \dcomp_{N} g) = (\sigma \tensor \im_{A}) \dcomp (\im%
    \tensor \now(\fm)) \dcomp  (\sigma \tensor \im_{B}) ; (\im%
    \tensor \now(\gm)).\]
\[\now(f \tensor_{N} g) = (\im\tensor\sigma\tensor\im);(\now(f)\tensor\now(g));(\im \tensor \sigma\tensor\im).\]

\begin{example}[Fibonacci example]\defining{linkexamplefibonacci}{}\label{computation:fibonacci}
  The code for \hask{fibonacci} in~\Cref{section:code} follows the definition in~\Cref{example:fibonacci}.
  We can execute it to obtain the first 10 numbers from the Fibonacci sequence.
  \begin{haskellcode}
> take 10 <$> run fibonacci
Identity [0,1,1,2,3,5,8,13,21,34]
  \end{haskellcode}
\end{example}

\begin{example}[Random walk example]\label{implementation:walk}\defining{linkexamplewalk}{}
  The code for \hask{walk} in~\Cref{section:code} follows the definition in~\Cref{example:walk}.
  We can execute it multiple times to obtain different random walks starting from 0.
  \begin{haskellcode}
> take 10 <$> run walk
[0,1,0,-1,-2,-1,-2,-3,-2,-3]
> take 10 <$> run walk
[0,1,2,1,2,1,2,3,4,5]
> take 10 <$> run walk
[0,-1,-2,-1,-2,-1,0,-1,0,-1]
\end{haskellcode}
\end{example}

\begin{example}[Ehrenfest model]\label{implementation:ehrenfest}
  The code for \hask{ehrenfest} in~\Cref{section:code} follows the definition in~\Cref{example:ehrenfest}.
  We can execute it to simulate the Ehrenfest model.
  \begin{haskellcode}
> take 10 <$> run ehrenfest
[([2,3,4],[1]),([2,3],[1,4]),
([2,3,4],[1]),([2,4],[1,3]),
([2],[1,3,4]),([2,4],[1,3]),
([2],[1,3,4]),([2,3],[1,4]),
([3],[1,2,4]),([2,3],[1,4])]
  \end{haskellcode}
\end{example}

\onecolumn

\subsection{Term derivations}
\begin{example}[Fibonacci]
  {\small
  \begin{mathpar}
    \infer{\infer{\infer{\infer{ }{\entails 0 : \naturals_0} \and
          \infer{\infer{\infer{\infer{ }{f : \naturals \entails f : \naturals}}{f : \naturals \entails \Copy(f) : \naturals \tensor \naturals} \and
              \infer{\infer{ }{f_{1} : \naturals \entails f_{1} : \naturals}  \and \infer{\infer{ }{\entails 1 :\naturals_0} \and \infer{\infer{ }{f_{2} : \naturals \entails f_{2} : \naturals}}{f_{2} : \naturals \entails \Wait(f_{2}) : \delay \naturals}}{f_{2} : \naturals \entails 1\ \Fby\ \Wait(f_{2}) : \naturals}}{f_{1} : \naturals, f_{2} : \naturals \entails f_{1} + 1\ \Fby\ \Wait(f_{2}) : \naturals}}{f : \naturals \entails \SPLIT{\Copy(f)}{f_{1},f_{2}}{} (f_{1} + 1\ \Fby\ \Wait(f_{2})) : \naturals}}{
        f : \delay \naturals \entails \SPLIT{\Copy(f)}{f_{1},f_{2}}{} (f_{1} + 1\ \Fby\ \Wait(f_{2})) : \delay\naturals}}
        {f : \delay \naturals \entails 0\ \Fby\ \SPLIT{\Copy(f)}{f_{1},f_{2}}{} (f_{1} + 1\ \Fby\ \Wait(f_{2})) : \naturals}}
      {f : \delay \naturals \entails \Copy (0\ \Fby\ \SPLIT{\Copy(f)}{f_{1},f_{2}}{} (f_{1} + 1\ \Fby\ \Wait(f_{2}))\ ) : \naturals \tensor \naturals}}
    {\entails : \Fbk{f}{} \Copy (0\ \Fby\ \SPLIT{\Copy(f)}{f_{1},f_{2}}{} (f_{1} + 1\ \Fby\ \Wait(f_{2}))\ ) : \naturals}
  \end{mathpar}}
\end{example}

\begin{example}[Random walk]
  {\small
  \begin{mathpar}
    \infer{\infer{\infer{ }{\entails 0 : \naturals_0} \and \infer{\infer{ }{\entails \uniform{(-1,1)} : \naturals} \and \infer{ }{w : \delay \naturals \entails w : \delay \naturals}}
        {w : \delay \naturals \entails \uniform{(-1,1)} + w : \delay \naturals}}
      {\infer{w : \delay \naturals\entails 0\ \Fby\ (\uniform{(-1,1)} + w) : \naturals}
        {w : \delay \naturals\entails \Copy( 0\ \Fby\ (\uniform{(-1,1)} + w)) : \naturals \tensor \naturals}}}
    {\entails \Fbk{w}{} \Copy(0\ \Fby\ (\uniform{(-1,1)} + w)) : \naturals}
  \end{mathpar}}
\end{example}

\begin{example}[Ehrenfest model]
  {\small
  \begin{mathpar}
    \infer{\infer{\infer{\infer
          {\entails (1,2,3,4) : \textsf{Urn}_0 \\\\ \entails () : \textsf{Urn}_0}
          {\entails [(1,2,3,4),()] \\\\ : \textsf{Urn}_0 \tensor \textsf{Urn}_0} \and
          \infer{\infer{ }
            {u : \textsf{Urn} \tensor \textsf{Urn} \entails\\\\ u : \textsf{Urn} \tensor \textsf{Urn}} \and
            \infer{\infer{\infer{ }{\entails \uniform : \naturals}}
              {\entails \Copy(\uniform) : \naturals \tensor \naturals} \and
              \infer{\infer{\infer{ }{n_1 : \naturals \entails \\\\ n_1 : \naturals} \and
                  \infer{ }{u_1 : \mathsf{Urn} \entails \\\\ u_1 : \mathsf{Urn}}}
                {n_1 : \naturals, u_1:\textsf{Urn} \entails\\\\ \Move(n_{1},u_{1}) : \textsf{Urn}} \and
                \infer{\infer{ }{n_2 : \naturals \entails \\\\  n_2 : \naturals} \and
                  \infer{ }{u_2 : \mathsf{Urn} \entails \\\\ u_2 : \mathsf{Urn}}}
                 {n_2 : \naturals, u_2 :\textsf{Urn} \entails\\\\ \Move(n_{2},u_{2}) : \textsf{Urn}}}
               {n_1 : \naturals, n_2 : \naturals, u_1:\textsf{Urn}, u_2:\textsf{Urn} \entails \\\\
                 [\Move(n_{1},u_{1}), \Move(n_{2},u_{2})] : \textsf{Urn} \tensor \textsf{Urn}}}
            {u_1:\textsf{Urn},u_2:\textsf{Urn} \entails \SPLIT{\Copy(\uniform)}{n_{1},n_{2}}{} \\\\
              [\Move(n_{1},u_{1}), \Move(n_{2},u_{2})] : \textsf{Urn} \tensor \textsf{Urn}}}
          {\infer{
                u : \textsf{Urn} \tensor \textsf{Urn} \entails \SPLIT{u}{u_{1},u_{2}}{} \\\\ \SPLIT{\Copy(\uniform)}{n_{1},n_{2}}{} [\Move(n_{1},u_{1}), \Move(n_{2},u_{2})] : \textsf{Urn} \tensor \textsf{Urn}
              }
{u : \delay (\textsf{Urn} \tensor \textsf{Urn}) \entails \SPLIT{u}{u_{1},u_{2}}{} \\\\ \SPLIT{\Copy(\uniform)}{n_{1},n_{2}}{} [\Move(n_{1},u_{1}), \Move(n_{2},u_{2})] : \delay(\textsf{Urn} \tensor \textsf{Urn})}}}
        {u : \delay (\textsf{Urn} \tensor \textsf{Urn}) \entails
          [(1,2,3,4),()]\ \Fby\ \SPLIT{u}{u_{1},u_{2}}{} \\\\ \SPLIT{\Copy(\uniform)}{n_{1},n_{2}}{} [\Move(n_{1},u_{1}), \Move(n_{2},u_{2})] : \textsf{Urn} \tensor \textsf{Urn}}}
      {u : \delay (\textsf{Urn} \tensor \textsf{Urn}) \entails
      \Copy([(1,2,3,4),()]\ \Fby\ \SPLIT{u}{u_{1},u_{2}}{} \\\\ \SPLIT{\Copy(\uniform)}{n_{1},n_{2}}{} [\Move(n_{1},u_{1}), \Move(n_{2},u_{2})]) : (\textsf{Urn} \tensor \textsf{Urn}) \tensor (\textsf{Urn} \tensor \textsf{Urn})}}
    {\entails \Fbk{u}{}  \Copy([(1,2,3,4),()]\ \Fby\ \SPLIT{u}{u_{1},u_{2}}{} \\\\
      \SPLIT{\Copy(\uniform)}{n_{1},n_{2}}{} [\Move(n_{1},u_{1}), \Move(n_{2},u_{2})]) : \textsf{Urn} \tensor \textsf{Urn}}
  \end{mathpar}}
\end{example}
\newpage
\onecolumn

\subsection{Code}\label{section:code}
The following code has been compiled under GHCi, version 8.6.5.
The ``random'' library may need to be installed.

\begin{minted}{haskell}
-- | Monoidal Streams for Dataflow Programming.

-- Anonymous.

-- The following code implements the category of monoidal streams over a
-- monoidal category with an identity-on-objects functor from the
-- (pseudo)category of Haskell types.

-- During the manuscript, we have needed to perform some computations: for
-- instance, to see that according to our definitions the Fibonacci morphism we
-- describe really computes the Fibonacci sequence. Computations of this kind
-- are difficult and tedious to write and to justify, and the reader may find
-- difficult to reproduce them. Instead of explicitly writing these
-- computations, we implement them and we provide the necessary code so that the
-- reader can verify the result from the computation.

-- Morphisms in a monoidal category are written in "Arrow" notation, using (>>>)
-- for sequential composition and (***) for parallel composition. Coherence
-- morphisms need to be written explicitly, we usually write them at the side of
-- the diagram.

{-# LANGUAGE GADTs                     #-}
{-# LANGUAGE TypeSynonymInstances      #-}
{-# LANGUAGE FlexibleInstances        #-}

module Main where

import Prelude hiding (id)
import Data.Functor.Identity
import Data.List
import Control.Category
import Control.Arrow
import System.Random
import System.IO.Unsafe


-- Fixpoint equation for monoidal streams. Figure 5.
type Stream c = StreamWithMemory c ()

data StreamWithMemory c n x y where
  StreamWithMemory :: (Arrow c) =>
     c (n , x) (m , y)
    -> StreamWithMemory c m x y
    -> StreamWithMemory c n x y


--------------
-- EXAMPLES --
--------------
fibonacci :: Stream (Kleisli Identity) () Int
fibonacci = fbk                             $ runitS
  >>> copy                                >>> lunitinvS *** id
  >>> delay (k1 *** wait) *** id
  >>> delay fby *** id
  >>> plus                                >>> lunitinvS
  >>> k0 *** id
  >>> fby
  >>> copy

walk :: Stream (Kleisli IO) (()) (Int)
walk = fbk
    $ (id *** unif)
  >>> plus                                >>> lunitinvS
  >>> k0 *** id
  >>> fby
  >>> copy
 where
   unif :: Stream (Kleisli IO) () Int
   unif = lift $ Kleisli (\() -> do
      boolean <- randomIO
      return $ if boolean then 1 else -1)

type Urn = [Int]


ehrenfest :: Stream (Kleisli IO) (()) (Urn,Urn)
ehrenfest = fbk                          $ runitS>>> lunitinv *** lunitinv
  >>> (full *** idS) *** (empty *** idS)    >>> runitinv
  >>> (fby *** fby) *** unif
  >>> idS *** copy                          >>> associnv >>> idS *** assoc
  >>> idS *** (sigma *** idS)               >>> idS *** associnv >>> assoc
  >>> move *** move
  >>> copy *** copy                          >>> associnv >>> idS *** assoc
  >>> idS *** (sigma *** idS)               >>> idS *** associnv >>> assoc
  where
    unif :: Stream (Kleisli IO) () Int
    unif = lift $ Kleisli (\() -> randomRIO (1,4))

    empty :: Stream (Kleisli IO) () Urn
    empty = lift $ arr (\() -> [])

    full :: Stream (Kleisli IO) () Urn
    full = lift $ arr (\() -> [1,2,3,4])

    move :: Stream (Kleisli IO) (Urn, Int) Urn
    move = lift $ arr (\(u,i) ->
      if elem i u
        then (delete i u)
        else (insert i u))

--- take 10 <$> run fibonacci
--- take 10 <$> run walk
--- take 10 <$> run ehrenfest



---------------------------
-- THE FEEDBACK CATEGORY --
---------------------------

compS :: (Arrow c) =>
  StreamWithMemory c m x y ->
  StreamWithMemory c n y z ->
  StreamWithMemory c (m , n) x z
compS
  (StreamWithMemory fnow flater)
  (StreamWithMemory gnow glater) =
  StreamWithMemory (sequentialComposition fnow gnow) (compS flater glater)
 where

   -- Definition 5.2.
   -- Sequential composition, "now".
   sequentialComposition :: Arrow c
     => c (m , x) (p , y)
     -> c (n , y) (q , z)
     -> c ((m,n),x) ((p,q),z)
   sequentialComposition f g =
        sigma  *** id    >>> associnv
    >>> id *** f       >>> assoc
    >>> sigma  *** id  >>> associnv
    >>> id *** g       >>> assoc

comp :: (Arrow c) => Stream c x y -> Stream c y z -> Stream c x z
comp f g = lact lunitinv (compS f g)



tensorS :: (Arrow c) =>
  StreamWithMemory c p x y ->
  StreamWithMemory c p' x' y' ->
  StreamWithMemory c (p , p') (x,x') (y,y')
tensorS
  (StreamWithMemory fnow flater)
  (StreamWithMemory gnow glater) =
  StreamWithMemory (parallelCompTosition fnow gnow) (tensorS flater glater)
 where

   -- Definition 5.3. Parallel compTosition.
   parallelCompTosition :: Arrow c
     => c (m,x) (p,z)
     -> c (n,y) (q,w)
     -> c ((m,n),(x,y)) ((p,q),(z,w))
   parallelCompTosition f g =
         associnv                   >>> id *** assoc
     >>> (id *** (sigma *** id))  >>> id *** associnv >>> assoc
     >>> (f *** g)                >>> associnv            >>> id *** assoc
     >>> (id *** (sigma *** id))  >>> id *** associnv >>> assoc

tensor :: Arrow c => Stream c x y -> Stream c x' y' -> Stream c (x,x') (y,y')
tensor f g = lact lunitinv (tensorS f g)


lact :: (Arrow c) => c n m -> StreamWithMemory c m x y -> StreamWithMemory c n x y
lact f (StreamWithMemory now later) = StreamWithMemory ((f *** id) >>> now) later


fbkS :: (Arrow c) =>
  StreamWithMemory c m (s,x) (s,y) ->
  StreamWithMemory c (m, s) x y
fbkS (StreamWithMemory now later) =
  StreamWithMemory (nowFeedback now) (fbkS later)
 where

   -- Definition 5.7. Feedback operation.
   nowFeedback :: (Arrow c) => c (m,(s,x)) (n,(t,y)) -> c ((m,s),x) ((n,t),y)
   nowFeedback f = associnv >>> f >>> assoc

fbk :: (Arrow c) => Stream c (s,x) (s,y) -> Stream c x y
fbk t = lact (arr (\() -> ((),undefined))) (fbkS t)

idS :: (Arrow c) => Stream c x x
idS = StreamWithMemory (id) idS


lift :: (Arrow c) => c x y -> Stream c x y
lift f = StreamWithMemory (id *** f) (lift f)

liftarr :: (Arrow c) => (x -> y) -> Stream c x y
liftarr s = lift $ arr s

instance (Arrow c) => Category (Stream c) where
  id = idS
  (.) f g = comp g f

instance (Arrow c) => Arrow (Stream c) where
  arr = liftarr
  (***) = tensor

instance (Arrow c) => ArrowLoop (Stream c) where
  loop f = fbk $ sigma >>> f >>> sigma


delay :: (Arrow c) => Stream c x y -> Stream c x y
delay f = StreamWithMemory (id *** undefined)  f




------------
-- ARROWS --
------------
assoc :: Arrow c => c (x,(y,z)) ((x,y),z)
assoc = arr $ \(x,(y,z)) -> ((x,y),z)
assocS :: Arrow c => Stream c (x,(y,z)) ((x,y),z)
assocS = lift assoc

associnv :: Arrow c => c ((x,y),z) (x,(y,z))
associnv = arr $ \((x,y),z) -> (x,(y,z))
associnvS :: Arrow c => Stream c ((x,y),z) (x,(y,z))
associnvS = lift $ associnv

lunit :: Arrow c => c ((),a) a
lunit = arr $ \((),a) -> a
lunitS :: Arrow c => Stream c ((),a) a
lunitS = lift $ lunit

lunitinv :: Arrow c => c a ((),a)
lunitinv = arr $ \a -> ((),a)
lunitinvS :: Arrow c => Stream c a ((),a)
lunitinvS = lift $ lunitinv

runit :: Arrow c => c (a,()) a
runit = arr $ \(a,()) -> a
runitS :: Arrow c => Stream c (a,()) a
runitS = lift $ runit

runitinv :: Arrow c => c a (a,())
runitinv = arr $ \a -> (a,())
runitinvS :: Arrow c => Stream c a (a,())
runitinvS = lift $ runitinv


sigma :: Arrow c => c (x,y) (y,x)
sigma = arr $ \(x,y) -> (y,x)

sigmaS :: Arrow c => Stream c (x,y) (y,x)
sigmaS = lift $ sigma


----------------
-- GENERATORS --
----------------
fby :: (Monad t) => Stream (Kleisli t) (a , a) a
fby = StreamWithMemory (Kleisli $ \((),(x,y)) -> pure ((),x)) (lift (arr snd))

copy :: (Monad t) => Stream (Kleisli t) a (a,a)
copy = lift (Kleisli $ \a -> pure (a,a))


k0,k1,k2 :: (Arrow c) => Stream c () Int
k0 = lift $ arr (\() -> 0)
k1 = lift $ arr (\() -> 1)
k2 = lift $ arr (\() -> 2)

plus :: (Arrow c) => Stream c (Int,Int) Int
plus = lift $ arr (\(a,b) -> a + b)

wait :: (Arrow c) => Stream c a a
wait = fbk sigmaS

------------
-- SYSTEM --
------------

class (Monad m) => IOMonad m where unsafeRun :: m a -> m a
instance IOMonad IO where unsafeRun = unsafeInterleaveIO
instance IOMonad Identity where unsafeRun = id


runUnsafeWithMemory :: (IOMonad t) => m -> StreamWithMemory (Kleisli t) m a b -> [a] -> t [b]
runUnsafeWithMemory m (StreamWithMemory (Kleisli now) later) (a:as) = do
  (n , b)<- now (m , a)
  l <- unsafeRun $ runUnsafeWithMemory n later as
  pure (b : l)

runUnsafe :: (IOMonad t) => Stream (Kleisli t) a b -> [a] -> t [b]
runUnsafe = runUnsafeWithMemory ()

run :: (IOMonad t) => Stream (Kleisli t) () a -> t [a]
run s = runUnsafe s (repeat ())

------------------------------------------

main :: IO ()
main = return ()
\end{minted}

\end{document}